\documentclass[12pt,draftcls,onecolumn]{IEEEtran}
\usepackage[boxruled]{algorithm2e}
\usepackage{cite}
\usepackage{graphicx}
\usepackage{epstopdf}
\usepackage{psfrag}
\usepackage{subfigure}
\usepackage{url}
\usepackage{amsmath}
\usepackage{array}
\usepackage{amssymb}
\usepackage{amsfonts}
\usepackage{epstopdf}
\newtheorem{proposition}{Proposition}
\newtheorem{lemma}{Lemma}

\newtheorem{remark}{Remark}
\newtheorem{definition}{Definition}
\newtheorem{theorem}{Theorem}
\newtheorem{example}{Example}

\title{Precoding Based Network Alignment using Transform Approach for Acyclic Networks with
Delay}
\author{ Teja Damodaram Bavirisetti, Abhinav Ganesan, K. Prasad and B. Sundar Rajan
\thanks{This work was supported partly by the DRDO-IISc program on Advanced Research in Mathematical Engineering through a research grant, and partly by the INAE Chair Professorship grant to B.~S.~Rajan. Parts of this paper appeared in the Proceedings of IEEE Information Theory Workshop (ITW) 2011, IEEE International Symposium on Information Theory (ISIT) 2012, and IEEE Global Communications Conference (GLOBECOM) 2012.}
\thanks{Teja Damodaram Bavirisetti is currently with Broadcom Technologies, Bangalore, India. Abhinav Ganesan, K. Prasad and B.~Sundar~Rajan are with the Department of Electrical Communication Engineering, Indian Institute of Science, Bangalore-560012, India (e-mail: dbaviris@broadcom.com, \{abhig\_88, prasadk5, bsrajan\}@ece.iisc.ernet.in).}
}
\begin{document}

\maketitle
\thispagestyle{empty}	

\begin{abstract}
The algebraic formulation for linear network coding in acyclic networks with the links having integer delay is well known. Based on this formulation, for a given set of connections over an arbitrary acyclic network with integer delay assumed for the links, the output symbols at the sink nodes, at any given time instant, is a $\mathbb{F}_{p^m}$-linear combination of the input symbols across different generations where, $\mathbb{F}_{p^m}$ denotes the field over which the network operates ($p$ is prime and $m$ is a positive integer). We use finite-field discrete fourier transform (DFT) to convert the output symbols at the sink nodes, at any given time instant, into a $\mathbb{F}_{p^m}$-linear combination of the input symbols generated during the same generation {\em without making use of memory at the intermediate nodes}.  We call this as transforming the acyclic network with delay into {\em $n$-instantaneous networks} ($n$ is sufficiently large). We show that under certain conditions, there exists a network 
code satisfying sink demands in the usual (non-transform) approach if and only if there exists a network code satisfying sink demands in the transform approach. When the zero-interference conditions are not satisfied, we propose three Precoding Based Network Alignment (PBNA) schemes for three-source three-destination multiple unicast network with delays ($3$-S $3$-D MUN-D) termed as PBNA using transform approach and time-invariant local encoding coefficients (LECs), PBNA using time-varying LECs, and PBNA using transform approach and block time-varying LECs.  We derive sets of necessary and sufficient conditions under which throughputs close to $\frac{n'+1}{2n'+1}$, $\frac{n'}{2n'+1}$, and $\frac{n'}{2n'+1}$  are achieved for the three source-destination pairs in a $3$-S $3$-D MUN-D employing PBNA using transform approach and time-invariant LECs, and PBNA using transform approach and block time-varying LECs where, $n'$ is a positive integer. For PBNA using time-varying LECs, we obtain a sufficient condition 
under 
which a throughput demand of $\frac{n_1}{n}$, $\frac{n_2}{n}$, and $\frac{n_3}{n}$ can be met for the three source-destination pairs in a $3$-S $3$-D MUN-D where, $n_1$, $n_2$ and $n_3$ are positive integers less or equal to the positive integer $n$. This condition is also necessary when $n_1+n_3=n_1+n_2=n$ where, $n_1 \geq n_2 \geq n_3$.

\end{abstract}	
\section{Introduction}
\label{sec1}
The notion of Network Coding was introduced in \cite{ACLY} where the capacity of wireline multicast networks is characterized. Scalar linear network coding was found to achieve the capacity of multicast networks \cite{CLY}. The existence problem of scalar linear network coding for networks without delay (i.e., instantaneous networks) was converted into an algebraic problem in \cite{KoM}. In the meanwhile, it was shown that \cite{Leh} there exist solvable non-multicast networks where scalar linear network coding is insufficient. In addition, \cite{Leh} also showed that determining the existence of linear network coding solution for multiple unicast networks is NP-hard in general. In \cite{MEHK}, it was conjectured that vector linear network coding suffices to solve networks with arbitrary message demands.  Subsequently, Dougherty et al. \cite{DFZ}  disproved the conjecture by showing that there exists networks where vector linear network coding does not achieve network capacity and that nonlinear network 
coding are required in general. However, the practicality of linear network codes led to construction of suboptimal network codes for Multiple Unicast networks based on linear programming \cite{TRLKM}.  
 
The concept of interference alignment originally introduced in interference channels \cite{CaJ} was applied by Das et al. \cite{DVJM,ADVJM} in a three-source three-destination instantaneous multiple unicast network ($3$-S $3$-D I-MUN), where the zero interference conditions of Koetter et al. \cite{KoM} cannot be met, to achieve a rate close to half for each source-destination pair. Since precoding matrices are used at the sources for interference alignment and exploited for network coding in $3$-S $3$-D I-MUN, it came to be known as Precoding Based Network Alignment (PBNA) \cite{MRMJ}. Though PBNA is not optimal in general for a  $3$-S $3$-D I-MUN \cite{ADVJM}, it provides a simple and systematic manner of network code construction that can guarantee (under certain conditions) an asymptotic rate of half for every source-destination pair when the zero interference conditions cannot be met.

A set of sufficient conditions for feasibility of PBNA in a $3$-S $3$-D I-MUN were obtained in \cite{DVJM}. However, the set of conditions were infinite and hence, impossible to check. Moreover, the sufficient conditions were constrained by the use of  particular precoding matrices at the sources. These motivated the work of Meng et al. \cite{MRMJ} where, a finite set of conditions are obtained for feasibility of PBNA in a $3$-S $3$-D I-MUN that are both necessary and sufficient. We call these finite set of conditions as the {\em ``reduced feasibility conditions''}. The highlight of their result is that PBNA with arbitrary precoding matrices is feasible iff PBNA is feasible with the choice of precoding matrices as in \cite{DVJM} (with the number of symbol extensions being  greater than or equal to five). The derivation of the result involved taking into account graph related properties.

The case of acyclic networks with delays was abstracted in \cite{KoM} as acyclic networks where each link in the network has an integer delay associated with it. In the current work, we look at a technique similar to \cite{DVJM} for providing throughput guarantees in certain acyclic networks with delays where the zero-interference conditions cannot be satisfied while  {\em not making use of any memory at the intermediate nodes} (i.e., nodes other than the sources and sinks). The set of all $\mathbb F_{p^m}$-symbols generated by the sources at the same time instant are said to constitute a {\em generation} where, $\mathbb{F}_{p^m}$ denotes the field over which the network operates ($p$ is a prime number and $m$ is a positive integer). The output symbols at the sink nodes, at any given time instant, is a $\mathbb{F}_{p^m}$-linear combination of the input symbols across different generations. We convert the output symbols at the sink nodes, at any given time instant, into a $\mathbb{F}_{p^m}$-linear 
combination of the input symbols generated during the same generation, by using techniques similar to Multiple Input Multiple Output-Orthogonal Frequency Division Multiplexing (MIMO-OFDM) \cite{RaC}. We call this technique as the {\em transform technique}, since we use Discrete Fourier Transform (DFT) over a finite field towards achieving this instantaneous behaviour in the network. 

As a first step towards guaranteeing a minimum throughput when the zero-interference conditions cannot be satisfied in an acyclic network with delay, we consider a three-source three-destination multiple unicast network with delays ($3$-S $3$-D MUN-D) with the source-destination pair denoted by $S_i$-$T_i$, $i =1,2,3$. We also assume a min-cut of one between source $S_i$ and destination $T_i$. Under this set-up, we derive a sufficient condition under which {\em PBNA using  time-varying local encoding coefficients (LECs)} is feasible whenever the throughput demands for the three source-destination pairs are given by $\frac{n_1}{n}$, $\frac{n_2}{n}$ and $\frac{n_3}{n}$ where, $n_1$, $n_2$, $n_3$, and $n$ are positive integers with $n_1,n_2,n_3\leq n$. This condition is also necessary when $n_1+n_3=n_1+n_2=n$  where, it is assumed without loss of generality that $n_1 \geq n_2 \geq n_3$. The condition is purely algebraic. But, this condition is often difficult to verify in practice. However, when {\em time-
invariant} LECs are used, our transform technique aids in obtaining network alignment matrices of the form similar to \cite{DVJM}, and the derived set of necessary and sufficient conditions under which PBNA is feasible under the transform technique are simpler to verify for a given number of symbol extensions $n=2n'+1$. We term this PBNA scheme as {\em PBNA using transform approach and time-invariant LECs}. Under this PBNA scheme, throughputs of $\frac{n'+1}{2n'+1}$, $\frac{n'}{2n'+1}$, and $\frac{n'}{2n'+1}$ are achieved for $S_1-T_1$, $S_2-T_2$, and  $S_3-T_3$ respectively, where $n'$ is a positive integer. So, for large $n'$, each of the throughputs is close to half. However, these conditions are applicable only to the case of precoding over a fixed number of symbol extensions, i.e., if the feasibility test fails over a symbol extension of length $2n'+1$, it is not known if the test would fail for a symbol extension of length greater than $2n'+1$. Hence, on the look-out for an elegant set of conditions 
that 
would help check the 
feasibility of PBNA in a $3$-S $3$-D MUN-D over any number of symbol extensions (like in \cite{MRMJ}), we propose a PBNA scheme for $3$-S $3$-D MUN-D which is different from PBNA using transform approach and time-invariant LECs, and PBNA using time-varying LECs. This scheme is termed as {\em PBNA using transform approach and block time-varying LECs} and we show that its feasibility conditions are the same as the reduced feasibility conditions of Meng et al. The drawback in PBNA using transform approach and block time-varying LECs is that the decoding delay is higher compared to PBNA using time-varying LECs, and PBNA using transform approach and time-invariant LECs. Formally, we define block time-varying LECs as follows.
\begin{definition}
A $3$-S $3$-D MUN-D is said to use block time varying LECs when the LECs are varied with every time block of length $k>1$ and remain constant within each time block.
\end{definition}

The contributions of this paper are summarized as follows. 
\begin{itemize}
\item  Given an acyclic network with delay, we convert the output symbols at the sink nodes at any given time instant into a $\mathbb{F}_{p^m}$-linear combination of the input symbols generated during the same generation, using finite-field DFT. We call this as transforming the acyclic network with delay into {\em $n$-instantaneous networks} where, $n$ is sufficiently large.

\item Using a constructive proof, we show that there exists a network code
(satisfying a certain property) that achieves the sink demands in the usual
(non-transform) approach if and only if there exists a network code satisfying
sink demands in the transform approach .

\item {\em PBNA with time-varying LECs:} For $3$-S $3$-D MUN-D, which do not satisfy the zero-interference conditions, we obtain a sufficient condition (called the {\em feasibility condition} for  PBNA with time-varying LECs) under which a throughput close to $\frac{n_1}{n}$, $\frac{n_2}{n}$, and $\frac{n_3}{n}$ are achieved for the source-destination pairs $S_1$-$T_1$, $S_2$-$T_2$, and $S_3$-$T_3$ respectively using time-varying LECs where, $n_1$, $n_2$, and $n_3$ are positive integers less than or equal to $n$. This condition is also necessary when $n_1+n_3=n_1+n_2=n$  where, without loss of generality, it is assumed that $n_1 \geq n_2 \geq n_3$.

\item {\em PBNA using transform approach and time-invariant LECs:} Assuming time-invariant LECs, for a given number of symbol extensions $n=2n'+1$, we use our transform technique to achieve throughputs close to $\frac{n'+1}{2n'+1}$, $\frac{n'}{2n'+1}$, and $\frac{n'}{2n'+1}$ for $S_1-T_1$, $S_2-T_2$, and  $S_3-T_3$ respectively under certain conditions, along with the use of alignment strategies. When $n'$ is large, the throughputs are close to half. The derived set of necessary and sufficient conditions under which PBNA using transform approach and time-invariant LECs is feasible are easier to verify when compared to the feasibility condition for PBNA with time-varying LECs. The set of necessary and sufficient conditions for this scheme can be derived as a special case of that of PBNA with time-varying LECs.

\item {\em PBNA using transform approach and block time-varying LECs:}  Using transform techniques and block time-varying LECs, a PBNA scheme different from the above two is proposed. The highlight of this scheme is that the derived set of necessary and sufficient conditions for feasibility of this PBNA scheme are shown to be the same as the reduced feasibility conditions for $3$-S $3$-D I-MUN which are independent of the number of symbol extensions $2n'+1\geq 5$ over which the independent symbols of each source are precoded. However, the decoding delay is higher in this scheme compared to the two other PBNA schemes proposed in this paper.
\end{itemize}

A comparison of the three proposed PBNA schemes is summarized in Table \ref{tab-comparison}.
\begin{table*} \label{tab-comparison}
\centering
\normalsize
\caption{Comparison of the three proposed PBNA schemes where PBNA $1$ denotes PBNA using transform approach and time-invariant LECs, PBNA $2$ denotes PBNA using time-varying LECs, and PBNA $3$ denotes PBNA using transform approach and block time-varying LECs}
\begin{tabular}{|c|c|c|c|}
\hline
 &PBNA $1$ & PBNA $2$ & PBNA $3$\\
\hline
Decoding Delay for &  &  & \\
the first $p$ symbols, & $2n'+1$ & $n=2n'+1$ &  $k(2n'+1)$\\
$p \leq k$ for some positive integer $k$ & & & \\
\hline
Dependence of the derived  & Dependent & Dependent & Independent\\
feasibility conditions on $n'/n$ &(on $n'$) & (on $n$) & (of $n$)\\
\hline
Existence of $3$-S $3$-D MUN-D  & Cannot exist & Can exist& \\
where one PBNA scheme is feasible & (Proposition \ref{prop-TIL_inf},  & (Example \ref{eg_TVL_F},  & Not known\\
when the other two are not. & Section \ref{subsec4}) &  in Section \ref{subsec4}) & \\
\hline
\end{tabular}
\end{table*}

The organization of this paper is as follows. In Section \ref{sec2}, after a brief overview of the system model for acyclic networks with delays using time-invariant LECs \cite{KoM}, we derive the system model for acyclic networks with delays, using time-varying LECs. Section \ref{sec3} presents the transform technique using which we convert the usual convolutional behaviour of the network with delay into instantaneous behaviour. In Section \ref{sec3}, we also prove the interchangeability of solving the usual (non-transform) network code existence problem and the counterpart in the transform technique. In Section \ref{subsec1},  PBNA using transform approach and time-invariant LECs is described for $3$-S $3$-D MUN-D where the zero-interference conditions cannot be satisfied. PBNA with time-varying LECs is described in Section \ref{subsec2}, and PBNA using transform approach and block time-varying LECs is described in Section \ref{sec5}. The feasibility conditions of the three PBNA schemes are compared in 
Section \ref{subsec4}. In Section \ref{sec6}, we discuss the potential of on-off schemes in achieving half-rate for every source-destination pair in a $3$-S $3$-D MUN-D. We conclude our paper in Section \ref{sec7} with discussion and directions for further research.

{\em Notations:} The cardinality of a set $E$ is denoted by $|E|$. A superscript of $t$ accompanying any variable (for example, $\epsilon^{(t)}$) or any matrix (for example, $M^{(t)}$) denotes that they are a function of time $t$, unless mentioned otherwise. The $i^\text{th}$ row, $j^\text{th}$ column element of a matrix $A$ is denoted by $[A]_{ij}$. The notation $\text{Col}(P) \subset \text{Col}(Q)$ denotes that the columns of the matrix $P$ are a subset of the columns of the matrix $Q$. The notation $\text{Span}(P)$ indicates the sub-space spanned by the columns of the matrix $P$. The notation $\text{Span}(P) \subset \text{Span}(Q)$ denotes that the space spanned by the columns of the matrix $P$ is a sub-space of the space spanned by the columns of the matrix $Q$. The determinant of a square matrix $A$ is denoted by $det(A)$. An identity matrix of size $\mu \times  \mu$ is denoted by $I_{\mu}$. For three-source three-destination unicast networks we shall use the term destination to denote sink. The Galois 
Field of cardinality $p^m$ is denoted by $GF(p^m)$ where, $p$ is a prime number and $m$ is a positive integer. The notation $a|b$ denotes that $a$ divides $b$ where, $a$ and $b$ belong to a ring. The notation $a \nmid b$ denotes that $a$ does not divide $b$. For positive integers $a$ and $b$, $\text{LCM}(a,b)$ denotes the least common multiple of $a$ and $b$. The notation $f(A)$ where, $A$ is a matrix, denotes that $f$ is a function of elements of the matrix $A$. The notation $diag(a_1,a_2,\cdots,a_n)$ denotes a diagonal matrix whose $i^{\text{th}}$ diagonal entry is given by $a_i$, for $i=1,2,\cdots,n$.
\section{System Model}
\label{sec2}
First, we shall briefly review the system model from \cite{KoM}. We consider a network represented by a Directed Acyclic Graph (DAG) $\cal G$ $=$ $(V, E)$, where $V$ is the set of nodes and $E$ is the set of directed links. We assume that every directed link between a pair of nodes represents an error-free link and has a capacity of one ${\mathbb F}_{p^m}$ symbol per link-use. Multiple links between two nodes are allowed and the $i^{th}$ directed link from $v_1$ $\in$ $V$ to $v_2$ $\in$ $V$ is denoted by $(v_1,v_2,i)$. The head and tail of a link $e$ $=$ $(v_1,v_2,i)$ are denoted by \mbox{$v_2$ $=$ $\mbox{head}(e)$} and $v_1$ $=$ $\mbox{tail}(e)$. Without loss of generality, we assume that a link between a pair of nodes has a unit delay (if the link has any other non-zero integer delay, we could introduce an appropriate number of dummy nodes in between the pair of nodes which are then connected by links of unit delays). Let ${\cal X}(v)=\{X(v,1),X(v,2),...,X(v,\mu_v)\}$ be the collection of discrete random 
processes that are generated at the node $v$. Let $\underline{X_v}$ $=$ $[X(v,1)~X(v,2)~...~X(v,\mu_v)]^T$. The random process transmitted through link $e$ is denoted by $Z(e)$. Communication is to be established between selected nodes in the network, i.e., we are required to
replicate a subset of the random process in ${\cal X}(v)$ at some different node $v^\prime$. This subset is denoted by ${\cal X}(v,v^\prime)$. A connection $c$ is defined as a triple \mbox{$(v,v^\prime,{X}(v,i))$ $\in$ $V \times V \times {\cal X}(v,v^\prime)$}, for some $i\in \{1,2, \cdots,\mu_v \}$\footnote{The definition of connection adopted here is different from that in \cite{KoM}.}. For the connection $c$, $v$ is called the source and $v^\prime$ is called the sink of $c$, i.e., $v$ $=$ $\mbox{source}(c)$ and $v^\prime$ $=$ $\mbox{sink}(c)$ ($\mbox{source}(c)$ $\neq$ $\mbox{sink}(c)$). The collection of $\nu_{v^\prime}$ random processes ${\cal Y}(v^\prime)$ $=$ $\{Y(v^\prime,1),Y(v^\prime,2),...,Y(v^\prime,\nu_{v^\prime})\}$ denotes the output at sink $v^\prime$. Let $\underline{Y_{v^\prime}}$ $=$ $[Y(v^\prime,1)~Y(v^\prime,2)~...~Y(v^\prime,\nu_{v^\prime})]^T$.

The input random processes $X(v,i)$, output random processes $Y(u,j)$ and  random processes $Z(e)$ transmitted on the link $e$ are considered as a power series in a delay parameter $D$, i.e., $X(v,i)$ $=$ $\sum_{t=0}^{\infty}X^{(t)}(v,i)D^t$, $Y(u,j)$ $=$ $\sum_{t=0}^{\infty}Y^{(t)}(u,j)D^t,$ and $Z(e)$ $=$ $\sum_{t=0}^{\infty}Z^{(t)}(e)D^t$.

Let $\cal G$ $=$ $(V, E)$ be an acyclic network with arbitrary finite integer
delay on its links. $\cal G$ is taken to be a $\mathbb{F}_{p^m}$-linear network \cite{KoM} where,
for all links the random process $Z(e)$ on a link $e=(v,u,i)$ $\in$ $E$
satisfies
\begin{align}
\nonumber
Z^{(t+1)}(e) = \sum_{j=1}^{\mu_v} \alpha_{j,e} X^{(t)}(v,j) +
\sum_{e^\prime:\text{head}(e^\prime)=\text{tail}(e)} \beta_{e^\prime,e}
Z^{(t)}(e^\prime)
\end{align}where, $\alpha_{j,e}$ and $\beta_{e^\prime,e}$ belong to
$\mathbb{F}_{p^m}$.
The output at any sink node $v^\prime$, is taken to be
\begin{align}
\label{eqnoutput}
Y^{(t+1)}(v^\prime,j) =\sum_{e^\prime:\mbox{{\small head}}(e^\prime)=v^\prime}
\epsilon_{e^\prime,j} Z^{(t)}(e^\prime)
\end{align}
where $\epsilon_{e^\prime,j}$ $\in$ $\mathbb{F}_{p^m}$. The coefficients,
$\alpha_{j,e}$, $\beta_{e^\prime,e}$ and  $\epsilon_{e^\prime,j}$ are also
called {\em local encoding coefficients} (LECs). The vector consisting of all LECs is
denoted by $\underline{\varepsilon}.$ Note that in \cite{KoM}, the definition
for the output processes at any given time instant at any sink involves linear
combinations of the received processes and output processes across different previous time
instants, and hence the variables involved in such linear combinations together
performed the function of decoding the received processes at the sinks to the
demanded input processes. However, in (\ref{eqnoutput}), at every sink, we only define a
preprocessing of the received symbols corresponding to the
previous time instant alone. The outputs $Y^{(t+1)}(v^\prime,j)$ as $t$ varies,
will then be used by sink-$v^\prime$ to decode the demanded input processes using
sufficient delay elements for feed-forward and feedback operations.  The LECs
are time-invariant unless mentioned otherwise.

We assume some ordering among the sources so that the random process generated
by the sources can be denoted, without loss of generality, as
$\underline{X_1}(D)$, $\underline{X_2}(D)$, $...$, $\underline{X_s}(D)$, where
$s$ denotes the number of sources and $\underline{X_i}(D)$ is a $\mu_i \times 1$
column vector given by
\begin{align*}
 \underline{X_i}(D)=[{X_{i1}}(D) ~{X_{i2}}(D) ~\ldots {X_{i\mu_i}}(D)]^T.
\end{align*}
\noindent Similarly, we assume some ordering among the sinks so that the output
random process at the sinks can be denoted, without loss of generality, as
$\underline{Y_1}(D)$, $\underline{Y_2}(D)$, $...$, $\underline{Y_r}(D)$, where
$r$ denotes the number of sinks and $\underline{Y_i}(D)$ is a $\nu_i \times 1$
column vector given by
\begin{align*}
 \underline{Y_i}(D)=[{Y_{i1}}(D) ~{Y_{i2}}(D) ~\ldots {Y_{i\nu_i}}(D)]^T.
\end{align*}
Let
\begin{align*}
Y(D)&=[\underline{Y_1}(D)^T ~\underline{Y_2}(D)^T~... ~\underline{Y_r}(D)^T]^T\\
&=[{y_1}(D) ~{y_2}(D)~... ~{y_{\nu}}(D)]^T, \\
X(D)&=[\underline{X_1}(D)^T ~\underline{X_2}(D)^T~... ~\underline{X_s}(D)^T]^T\\
&=[{x_1}(D) ~{x_2}(D)~... ~{x_{\mu}}(D)]^T
\end{align*}
where, ${x_{k_1}}(D)$ and ${y_{k_2}}(D)$ represent the input and output random process of some source-$i$ and sink-$j$ respectively, and $\mu=\sum_{i=1}^{s}\mu_i$ and $\nu=\sum_{i=1}^{r}\nu_i$. We now have \cite{KoM}
\begin{equation}
\label{Trmx}
Y(D)=M(D)X(D)
\end{equation}where, $M(D)$ denotes the {\em network transfer matrix} of size
${\nu \times \mu}$ with elements from $\mathbb{F}_{p^m}[D],$ the ring of polynomials
in the delay parameter $D$ with coefficients from $\mathbb{F}_{p^m}$. Now, $M(D)$ can also be
written as
\begin{align}\label{Mmtrx}
M(D)=\begin{bmatrix}
       M_{11}(D) & M_{21}(D) &  \cdots & M_{s1}(D) \\
       M_{12}(D) & M_{22}(D) &  \cdots & M_{s2}(D) \\
       \vdots & \vdots & \vdots &  \vdots \\
       M_{1r}(D) & M_{2r}(D) &  \cdots & M_{sr}(D) \\
     \end{bmatrix}.
\end{align}where $M_{ij}(D)$ denote the network transfer matrix from source-$i$ to sink-$j$ and is of size $\nu_{j}\times\mu_{i}$.
Let $d'_{max}$ and $d'_{min}$ denote the maximum and the minimum of all the path delays from source-$i$ to
sink-$j$, for all $(i,j)$, between which a path exists. Let
\begin{equation*}
d_{max}=d'_{max}-d'_{min}
\end{equation*}
Then, $M(D)$ can be written as
\begin{equation*}
M(D)=\sum_{d=d'_{min}}^{d'_{max}}M^{(d)}D^{d}=\left(\sum_{d=0}^{d_{max}}M^{(d)}
D^{d}\right)D^{d'_{min}},
\end{equation*}
where $M^{(d)}\in\mathbb{F}_{p^m}^{\nu \times \mu}$ represents the matrix-coefficients of $D^d$.

Since $D^{d'_{min}}$ just adds a constant additional delay to all the outputs, without loss of generality, we can take $M(D)$ to be
\begin{equation}\label{Mdisp}
M(D)=\sum_{d=0}^{d_{max}}M^{(d)}D^{d}.
\end{equation}
Hence, $M_{ij}(D)$ can be alternatively written as
\begin{equation}\label{Mdisp1}
M_{ij}(D)=\sum_{d=0}^{d_{max}}M^{(d)}_{ij}D^{d}.
\end{equation}
For each sink-$j,$ we also define $M_{j}(D)$ to be the $\nu_j \times \mu$
submatrix of $M(D)$ that captures the transfer function between all the sources
and the sink-$j,$ i.e.,
\begin{equation}
\label{eqnno3}
M_{j}(D)=\left[M_{1j}(D)~~M_{2j}(D)~~...~~M_{sj}(D)\right].
\end{equation}
In the network $\cal{G}$, let ${\cal C}_{j}$ denote the set of all connections
to sink-$j$. Let ${\cal C}=\cup_{j=1}^{r} {\cal C}_{j}$. 

\begin{definition}
An acyclic network with delay is said to be solvable if the demands of the sinks, as specified by the set of connections, can be met. 
\end{definition}
The following lemma from \cite{KoM} states the conditions for solvability of acyclic networks with
delay.
\begin{lemma}[\cite{KoM}]
 \label{zero-int}
An acyclic network with delay is solvable iff there exists an assignment to the LECs $\underline \varepsilon$ such that the following conditions are
satisfied.
\begin{enumerate}
\item {\em Zero-Interference:} ${M}_{ij}^{(d)}(l_i)=0$, for all pairs (source-$
i$, sink-$ j$) of nodes such that (source-$ i$, sink-$j$,
$\underline{X_i}^{(l_i)}(D)$) $\not\in {\cal C}_{j}$ for all $0\leq d \leq
d_{max}$, where ${M}_{ij}^{(d)}(l_i)$ denotes the $l_i^{\text{th}}$ column of
${M}_{ij}^{(d)}$ and $\underline{X_i}^{(l_i)}(D)$ denotes the $l_i^{\text{th}}$
element of  $\underline{X_i}(D)$. 
\item {\em Invertibility:} For every sink-$j$, the square submatrix $M'_j(D)$ of
$M_j(D)$ formed by juxtaposition of the columns of $M_{ij}(D)$, for all $i$ other than those involved in the zero-interference conditions, is invertible over $\mathbb{F}_{p^m}(D)$, the field of rationals over
$\mathbb{F}_{p^m}$.
\end{enumerate}
 \end{lemma}

A network code for $({\cal G},{\cal C})$ is defined to be a \textit{feasible network code} if it achieves the given set of demands at the sinks i.e., if the above zero-interference and the invertibility conditions are satisfied.

\subsection{System Model for time-varying LECs}
When the LECs are time-varying, we can't express the input-output relation as in
(\ref{Trmx}). Hence, first, we need to derive the input-output relation
involving transfer matrices which are dependent on varying LECs. Retaining the
notations as already introduced, we only point out the changes in the system
model here.

For a given DAG $\cal G$ with integer delay on its links, define the adjacency
matrix of $\cal G$ at time $t$ as the $|E|\times|E|$ matrix $K^{(t)},$  whose
elements are given by
\[
[K^{(t)}]_{ij}=\left\{
\begin{array}{ll}
\beta_{e_i,e_j}^{(t)}& ~\mbox{head}(e_i)=\mbox{tail}(e_j) \\
0 & ~\mbox{otherwise}.
\end{array}
\right.
\]
Let the entries of $\mu \times |E|$ matrix $A^{(t)}$, at time $t$, be given by

\vspace{-0.3cm}
{\small \[
[A^{(t)}]_{ij}=\left\{
\begin{array}{ll}
\alpha_{l,e_j}^{(t)}& ~x_i=X_{\text{tail}(e_j)l} \text{ for some $l$, $1 \leq l \leq \mu_{\text{tail}(e_j)}$}  \\
0 & ~\mbox{otherwise}.
\end{array}
\right.
\]}where, the tail of an edge originating from a source is identified by the source number. Also, let the entries of $\nu \times |E|$ matrix $B^{(t)}$, at time $t$, be
given by

\vspace{-0.3cm}
{\small \[
[B^{(t)}]_{ij}=\left\{
\begin{array}{ll}
\epsilon_{e_j,l}^{(t)}& ~y_i=Y_{\text{head}(e_j)l} \text{ for some $l$, $1 \leq l \leq \nu_{\text{head}(e_j)}$} \\
0 & ~\mbox{otherwise}.
\end{array}
\right.
\]}where, the head of an edge terminating at a sink is identified by the sink number. Denote the matrix of LECs from time instant $t_1$ to time instant $t_2$ $(t_2 \geq t_1)$  by $\underline{\varepsilon}^{(t_1,t_2)}$, i.e.,
\begin{align*}
\underline{\varepsilon}^{(t_1,t_2)} = \left[\underline{\varepsilon}^{(t_1)} 
~~\underline{\varepsilon}^{(t_1+1)}~ ~~\ldots~~ ~~\underline{\varepsilon}^{(t_2)}\right]
\end{align*}
where $\underline{\varepsilon}^{(t_i)}$ denotes the LECs at time $t_i$. 
Since the LECs are time varying, define a time-varying network transfer matrix given by

{
\vspace{-0.3cm}
\begin{align}
\nonumber
&M(D,t)^T \hspace{-0.1cm}=\hspace{-0.1cm}  \left(A^{(t-1)}I~D\hspace{-0.1cm}+A^{(t-2)}K^{(t-1)}D^2 \hspace{-0.1cm}+A^{(t-3)}K^{(t-2)}K^{(t-1)}D^3\right.\\
\nonumber
& \left.+\ldots +A^{(t-d_{max})}K^{(t-(d_{max}-1))}..K^{(t-2)}K^{(t-1)}D^{d_{max}} \right){B^{(t)}}^T\\
\label{eqn-define-component_M}
&~~~~~~~~~\triangleq\sum_{d=0}^{d_{max}}{M^{(d)}}^T({\underline{\varepsilon}^{(t-d,t)}})D^{d}
\end{align}}where the matrices ${M^{(d)}}^T({\underline{\varepsilon}^{(t-d,t)}})$ are a matrix functions of ${\underline{\varepsilon}^{(t-d,t)}}$, and  ${M^{(0)}}^T=\bf{0}$, i.e., the zero matrix, as each link in the network is assumed to have a unit delay.
The matrix $M(D,t)^T$ can also be written as

{\small \vspace{-0.3cm}
\begin{align}
\label{eqn-M_matrix_TV}
M(D,t)^T=\begin{bmatrix}
       M_{11}(D,t) & M_{21}(D,t) &  \cdots & M_{s1}(D,t) \\
       M_{12}(D,t) & M_{22}(D,t) &  \cdots & M_{s2}(D,t) \\
       \vdots & \vdots & \vdots &  \vdots \\
       M_{1r}(D,t) & M_{2r}(D,t) &  \cdots & M_{sr}(D,t) \\
     \end{bmatrix}
\end{align}}where, $M_{ij}(D,t)$ defines a time-varying network transfer matrix of size $\nu_j \times \mu_i$ from source-$i$ to sink-$j$. The matrix $M_{ij}(D,t)$ can also be written in terms of the delay parameter $D$ as
\begin{align}
\label{eqn-Mij_matrix_TV}
 M_{ij}(D,t)=\sum_{d=0}^{d_{max}}{M_{ij}^{(d)}}({\underline{\varepsilon}^{(t-d,t)}})D^{d}.
\end{align}

We shall now derive the relation between the input and the output symbols. 
\begin{definition}
 The impulse response $h_{k_1,k_2}(t',d)$ of the network between a source generating $x_{k_1}(D)$ and a sink whose output is $y_{k_2}(D)$ is defined as the value of the output symbol  $y^{(t')}_{k_2}$ when 

\vspace{-0.1cm}
{\[
x^{(t)}_{k}=\left\{
\begin{array}{ll}
1& k=k_1~\text{ and }~t=t'-d\\
0 & k\neq k_1~\text{ or } ~t \neq t'-d
\end{array}
\right.
\]}where, $1$ and $0$ denote the multiplicative and additive identity of the field $\mathbb{F}_p$ respectively.
\end{definition}

So, for a given value of $d$, if 

\vspace{-0.1cm}
{\[
x^{(t)}_{k}=\left\{
\begin{array}{ll}
a_{k_1}& k=k_1~\text{ and }~t=t'-d \\
0 & k\neq k_1~\text{ or }~t \neq t'-d
\end{array}
\right.
\]}where $a_{k_1} \in \mathbb{F}_{p^m}$ then, the value of the output symbol  $y^{(t')}_{k_2}$ is given by $a_{k_1}h_{k_1,k_2}(t',d)$ as the intermediate nodes linearly combine the symbols on its incoming links. If 

\vspace{-0.1cm}
{\[
x^{(t)}_{k}=\left\{
\begin{array}{ll}
a_k& k=1,2,\cdots,\mu,~\text{ and }~t=t'-d \\
0 & ~t\neq t'-d
\end{array}
\right.
\]}then, the value of the output symbol  $y^{(t')}_{k_2}$ is given by $\sum_{k=1}^{\mu}a_{k}h_{k,k_2}(t',d)$ as the intermediate nodes linearly combine the symbols on its incoming links. Now, observe that for a given $d$ and $t'$, the values of $h_{k,k_2}(t',d)$ for $k=1,2,\cdots,\mu$, is given by the $k_2^{\text{th}}$-row of the matrix ${M^{(d)}}^T({\underline{\varepsilon}^{(t'-d,t')}})$ which directly follows from the definition of ${M^{(d)}}^T({\underline{\varepsilon}^{(t'-d,t')}})$ in (\ref{eqn-define-component_M}). Hence, from (\ref{eqn-M_matrix_TV}) and (\ref{eqn-Mij_matrix_TV}), the relation between the output and the input symbols follows as
\begin{equation}
\label{trmx_tvarleks}
\underline{Y_{j}}^{(t)}=\sum_{i=1}^{s}\sum_{d=0}^{d_{max}}
M_{ij}^{(d)}({\underline{\varepsilon}^{(t-d,t)}})\underline{X_{i}}^{(t-d)}.
\end{equation}

The above input-output relation can also be seen by observing that acyclic networks with delay employing time-varying LECs are analogous to multiple-transmitter multiple-receiver MIMO channel with linear time-varying impulse response between every transmitter and every receiver \cite{And}.

\section{Transform Techniques for Acyclic Networks with Delay}
\label{sec3}
In this section, we show that the output symbols at all the sinks which was
originally a $\mathbb{F}_{p^m}$-linear combination of the input symbols across the
different generations, at any given time instant, can be transformed into a
$\mathbb{F}_{p^m}$-linear combination of the input symbols across the same
generation.

Consider a matrix $A$ of size ${n\nu\times n\mu}$ given by
\begin{align*}
\begin{bmatrix}
  A^{(0)} & A^{(1)} & \cdots & A^{(L-1)} & A^{(L)} & 0 & 0 & \cdots & 0 \\
  0 & A^{(0)} & \cdots & A^{(L-2)} & A^{(L-1)} & A^{(L)} & 0 &\cdots & 0 \\
  \vdots & \vdots & \vdots & \vdots & \vdots & \vdots &\vdots & \vdots & \vdots
\\
  A^{(1)} & A^{(2)} & \cdots & A^{(L)} & 0 & 0& 0 & \cdots & A^{(0)} \\
\end{bmatrix}
\end{align*}
where $A^{(i)}$ for all $i$, $0\leq i \leq L$, are matrices of size $\nu\times\mu$ whose elements belong to $\mathbb{F}_{p^m}$ and $n>>L$. Note that the $(i+1)^{\text{th}}$ row of matrices is a circular shift of the $i^{\text{th}}$ row of matrices in $A$. We assume that $n$ divides $p^m-1$. The choice of $n$ is such that, there exists an $\alpha\in\mathbb{F}_{p^m}$ such that $n$ is the smallest integer for which $\alpha^{n}=1$. This is indeed possible \cite{Bla}. Define matrices $\hat{A}^{(j)}$ of size $\nu\times\mu$, for $0\leq j\leq n-1$, as
\begin{align*} 
\hat{A}^{(j)}=\sum_{i=0}^{L}\alpha^{(n-1-j)i}A^{(i)}.
\end{align*}Let $F$ be the finite-field DFT matrix given by

{\small
\begin{align} \label{eqn-F-matrix}
F=&\begin{bmatrix}
    1 & 1 & 1 & \cdots & 1 \\
    1 & \alpha & \alpha^{2} & \cdots & \alpha^{n-1} \\
    1 & \alpha^{2} & \alpha^{4} & \cdots & \alpha^{2(n-1)} \\
    \vdots & \vdots & \vdots & \vdots & \vdots \\
    1 & \alpha^{n-1} & \alpha^{2(n-1)} & \cdots & \alpha^{(n-1)(n-1)} \\
  \end{bmatrix}.
\end{align}}
\noindent Define the matrix $Q_{\mu}$ as

{\small
\begin{align}
\label{diagqmu}
Q_{\mu}=&\begin{bmatrix}
    I_{\mu} & I_{\mu} & I_{\mu} & \cdots & I_{\mu} \\
    I_{\mu} & \alpha{I_{\mu}} & \alpha^{2}I_{\mu} & \cdots & \alpha^{n-1}I_{\mu}
\\
    I_{\mu} & \alpha^{2}I_{\mu} & \alpha^{4}I_{\mu} & \cdots &
\alpha^{2(n-1)}I_{\mu} \\
    \vdots & \vdots & \vdots & \vdots & \vdots \\
    I_{\mu} & \alpha^{n-1}I_{\mu} & \alpha^{2(n-1)}I_{\mu} & \cdots &
\alpha^{(n-1)(n-1)}I_{\mu}
  \end{bmatrix}.
\end{align}}Similarly, we can define  matrix $Q_{\nu}$.
The following theorem will be useful in establishing the results subsequently.
\begin{theorem}\label{diagthm}
The matrix $A$ can be block diagonalized  as
\begin{equation*}
A=Q_{\nu}\hat{A}Q_{\mu}^{-1},
\end{equation*}
where, $\hat{A}$ is given by
\begin{equation*}
\hat{A}=\begin{bmatrix}
          \hat{A}^{(n-1)} & 0 & 0 & \cdots & 0 \\
          0 & \hat{A}^{(n-2)} & 0 & \cdots & 0 \\
          \vdots & \vdots & \vdots & \cdots & \vdots \\
          0 & 0 & \cdots & 0 & \hat{A}^{(0)} \\
        \end{bmatrix}.
\end{equation*}
\end{theorem}
\begin{proof}
 Proof is given in Appendix \ref{appen_thm1}.
\end{proof}

Now, consider an arbitrary acyclic network with delay. From (\ref{Trmx}) and
(\ref{Mmtrx}),
\begin{align}\label{yj}
\underline{Y_j}(D)=\sum_{i=1}^{s}M_{ij}(D)\underline{X_i}(D).
\end{align}

Consider a transmission scheme where, in order to transmit $n$ ($>>d_{max}$) generations of input symbols at each source, the last $d_{max}$ generations (which we call the {\em cyclic prefix}) is transmitted first followed by the $n$ generations of input symbols. Hence, $n+d_{max}$ time slots at each source are used to transmit $n$ generations. From (\ref{yj}) and (\ref{Mdisp}), the output symbols at any time instant $t$ can be written as
\begin{align*}
\underline{Y_{j}}^{(t)}=\sum_{i=1}^{s}\sum_{d=0}^{d_{max}} M_{ij}^{(d)}\underline{X_{i}}^{(t-d)}.
\end{align*}Evaluating $\underline{Y_{j}}^{(t)}$  at the time instants \mbox{$t=-d_{max}, \cdots,(n-1)$}, we have (\ref{bigeqn}).
\begin{figure*}
\scriptsize
\begin{align}\label{bigeqn}
\begin{bmatrix}
  \underline{Y_j}^{(n-1)} \\
  \underline{Y_j}^{(n-2)} \\
  \vdots \\
  \underline{Y_j}^{(0)} \\
  \underline{Y_j}^{(-1)} \\
  \vdots \\
  \underline{Y_j}^{(-d_{max})} \\
\end{bmatrix}
=\sum_{i=1}^{s}\begin{bmatrix}
  M_{ij}^{(0)} & M_{ij}^{(1)} & \cdots  & M_{ij}^{(d_{max})} & 0 & 0 & \cdots &
0 & 0  \\
  0 & M_{ij}^{(0)} & \cdots & M_{ij}^{(d_{max}-1)} & M_{ij}^{(d_{max})} & 0 &
\cdots & 0& 0 \\
  \vdots & \vdots & \vdots & \vdots & \ddots & \ddots & \ddots & \vdots & \vdots
\\
  0 & 0 & \cdots & 0 &  M_{ij}^{(0)} & M_{ij}^{(1)} & \cdots &
M_{ij}^{(d_{max}-1)} & M_{ij}^{(d_{max})}\\
0 & 0 & \cdots & 0 & 0 &  M_{ij}^{(0)} & \cdots &  M_{ij}^{(d_{max}-2)} &
M_{ij}^{(d_{max-1})}\\
\vdots & \vdots & \vdots & \vdots & \vdots & \vdots & \vdots & \vdots & \vdots
\\
0 & 0 & \cdots & 0 & 0 & 0 & 0 & 0  & M_{ij}^{(0)}\\
\end{bmatrix}
\begin{bmatrix}
                 \underline{X_i}^{(n-1)} \\
                 \underline{X_i}^{(n-2)} \\
                 \vdots \\
                 \underline{X_i}^{(0)} \\
                 \underline{X_i}^{(n-1)} \\
                 \vdots \\
                 \underline{X_i}^{(n-d_{max})} \\
               \end{bmatrix}
\end{align}
\hrule
\end{figure*}
\begin{figure*}
\scriptsize
\begin{align}\label{bigeqn2}
\begin{bmatrix}
  \underline{Y_j}^{(n-1)} \\
  \underline{Y_j}^{(n-2)} \\
  \vdots \\
  \underline{Y_j}^{(0)} \\
\end{bmatrix}
=\sum_{i=1}^{s}\underbrace{\begin{bmatrix}
  M_{ij}^{(0)} & M_{ij}^{(1)} & \cdots & M_{ij}^{(d_{max}-1)} &
M_{ij}^{(d_{max})} & 0 & \cdots & 0 & 0 & 0  \\
  0 & M_{ij}^{(0)} & \cdots & M_{ij}^{(d_{max}-2)} & M_{ij}^{(d_{max}-1)} &
M_{ij}^{(d_{max})} & \cdots & 0 & 0& 0 \\
  \vdots & \vdots & \vdots & \vdots & \vdots & \vdots & \vdots & \vdots & \vdots
& \vdots \\
  M_{ij}^{(1)} & M_{ij}^{(2)} & \cdots & M_{ij}^{(d_{max})}& 0 & 0 & \cdots & 0
& 0 &  M_{ij}^{(0)} \\
\end{bmatrix}}_{M_{ij}}
\begin{bmatrix}
                 \underline{X_i}^{(n-1)} \\
                 \underline{X_i}^{(n-2)} \\
                 \vdots \\
                 \underline{X_i}^{(0)} \\
         \end{bmatrix}
\end{align}
\hrule
\end{figure*}
Discarding the first $d_{max}$ outputs at sink-$j$, (\ref{bigeqn}) can be re-written as (\ref{bigeqn2}). Using Theorem \ref{diagthm}, (\ref{bigeqn2}) can be re-written as
\begin{align}\label{eqn-transform}
\underline{Y_j}^{n}=\sum_{i=1}^{s} Q_{\nu_{j}}\hat{M}_{ij} Q_{\mu_{i}}^{-1}\underline{X_i}^{n}
\end{align}
where,
\begin{align*}
\underline{Y_j}^{n}=\begin{bmatrix}
  \underline{Y_j}^{(n-1)} \\
  \underline{Y_j}^{(n-2)} \\
  \vdots \\
  \underline{Y_j}^{(0)} \\
\end{bmatrix};~~~~
\underline{X_i}^{n}=\begin{bmatrix}
                 \underline{X_i}^{(n-1)} \\
                 \underline{X_i}^{(n-2)} \\
                 \vdots \\
                 \underline{X_i}^{(0)} \\
         \end{bmatrix};\\
\hat{M}_{ij}=\begin{bmatrix}
          {\hat{M}_{ij}^{(n-1)}} & 0 & 0 & \cdots & 0 \\
          0 & {\hat{M}_{ij}^{(n-2)}} & 0 & \cdots & 0 \\
          \vdots & \vdots & \vdots & \cdots & \vdots \\
          0 & 0 & 0 & \cdots & {\hat{M}_{ij}^{(0)}} \\
        \end{bmatrix}.
\end{align*}where, {\small $\hat{M}_{ij}^{(t)} = \sum_{d=0}^{d_{max}}\alpha^{(n-1-t)d} M_{ij}^{(d)}$}. At each source-$i$, $\underline{X'_i}^{n}=Q_{\mu_{i}}\underline{X_i}^{n}$ is transmitted instead of $\underline{X_i}^{n}$. At each sink-$j$, the received symbols are denoted by $\underline{Y'_j}^{n}$. Let $\underline{Y_j}^{n}=Q_{\nu_{j}}^{-1}\underline{Y'_j}^{n}$. Then, from (\ref{eqn-transform}) we have,
\begin{align}
\nonumber
\underline{Y'_j}^{n}=&\sum_{i=1}^{s} Q_{\nu_{j}}\hat{M}_{ij}
Q_{\mu_{i}}^{-1}\underline{X'_i}^{n}\\
\nonumber
\Rightarrow ~\underline{Y_j}^{n}=&Q_{\nu_{j}}^{-1}\sum_{i=1}^{s} Q_{\nu_{j}}\hat{M}_{ij}
Q_{\mu_{i}}^{-1}Q_{\mu_{i}}\underline{X_i}^{n}\\
\label{eqnrefer1unicast}
\Rightarrow ~ \underline{Y_j}^{n}=&\sum_{i=1}^{s} \hat{M}_{ij}\underline{X_i}^{n}.
\end{align}
For $0\leq t \leq n-1$, (\ref{eqnrefer1unicast}) can be re-written as
\begin{equation}
\label{instant}
\underline{Y_j}^{(t)}=\sum_{i=1}^{s}\hat{M}_{ij}^{(t)}\underline{X_i}^{(t)}.
\end{equation}
Hence, each element of $\underline{Y_j}^{(t)}$ is a ${\mathbb F}_{p^m}$-linear
combination of the input symbols across the same generation. We now say that we have transformed the acyclic network with delay into {\em $n$-instantaneous networks}.

\begin{remark}
Note that the linear processing of multiplying by matrices $Q_{\mu_{i}}$ at source-$i$ and $Q_{\nu_{j}}^{-1}$ at sink-$j$ are done in a distributed fashion which is necessary because the sources and sinks are distributed in the actual network.
\end{remark}

\begin{remark}
One can observe that transmitting $\underline{X'_i}^{n}=Q_{\mu_{i}}\underline{X_i}^{n}$ implies taking DFT across $n$ generations of each of the $\mu_i$ random-processes generated at source-$i$. Similarly, the pre-multiplication by $Q_{\nu_{j}}^{-1}$ at sink-$j$ simply implies taking IDFT across $n$ generations of each of the $\nu_j$ random-processes received. The entire processing, including addition of cyclic prefix at source-$i$ and removal of cyclic prefix at sink-$j$ is shown in a block diagram in Fig. \ref{blk_diag}.
\end{remark}
\begin{figure*}[htbp]
\hrule
\centering
\vspace{0.4cm}
\subfigure[Linear Processing at Source-$i$] {\includegraphics[totalheight=1.7in]{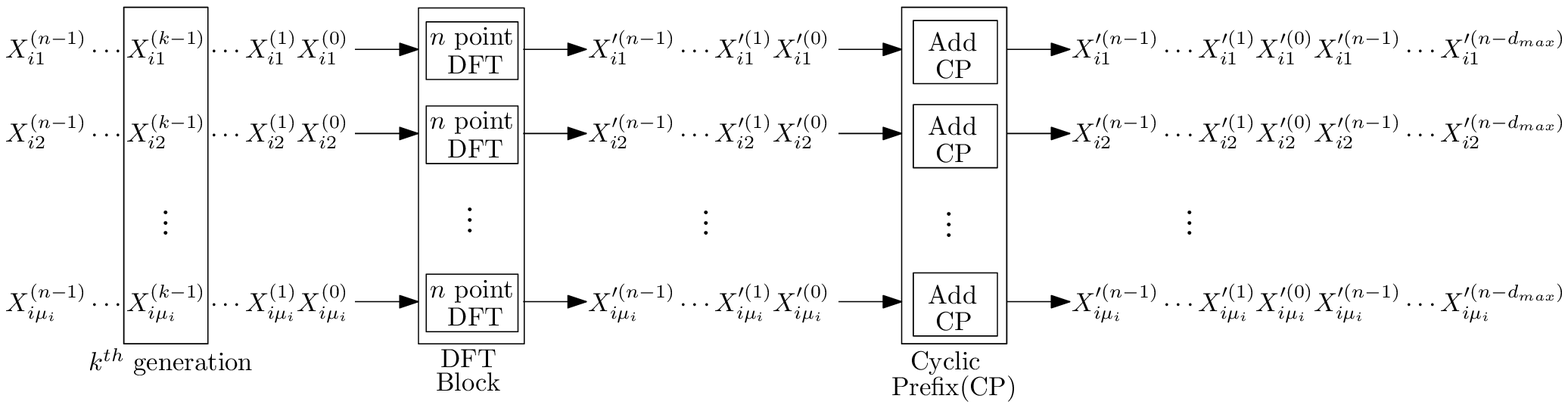}}\\
\subfigure[Linear Processing at Sink-$j$] {\includegraphics[totalheight=1.7in]{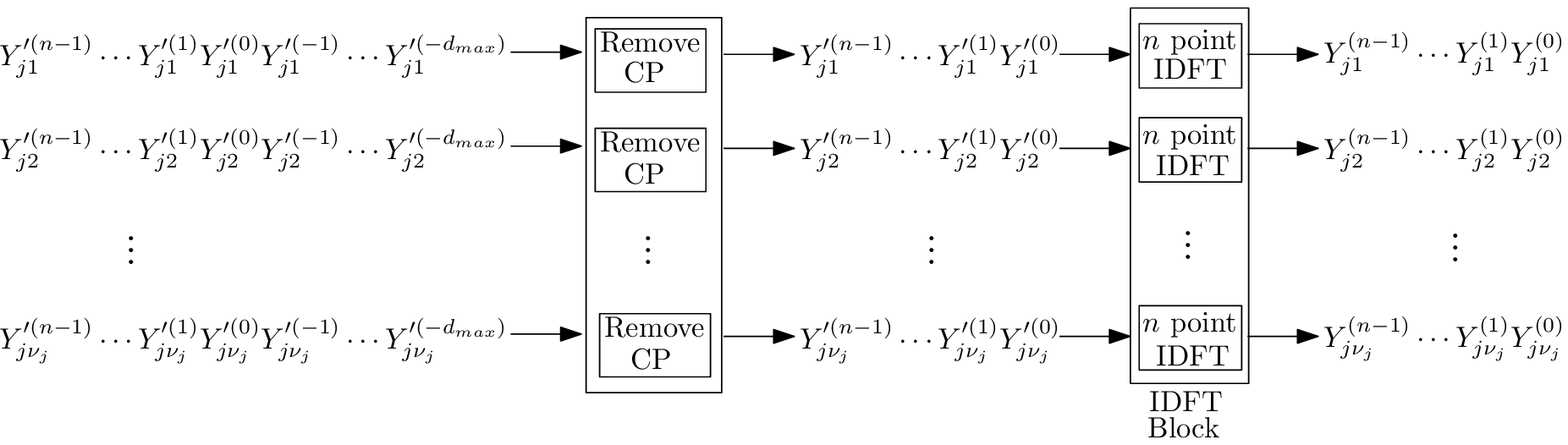}}
\caption{Block Diagram to illustrate linear processing at Source-$i$ and Sink-$j$.}	
\label{blk_diag}	
\end{figure*}

Now, (\ref{instant}) is re-written as
\begin{equation*}
\underline{Y_j}^{(t)}=\sum_{i=1}^{s}\sum_{l_{i}=1}^{\mu_i}\hat{M}_{ij}^{(t)}
(l_i)\underline{X_i}^{(t)}(l_i)
\end{equation*}
where, $\hat{M}_{ij}^{(t)}(l_i)$ denotes the $l_{i}^{\text{th}}$ column of
$\hat{M}_{ij}^{(t)}$ and $\underline{X_i}^{(t)}(l_i)$ denotes the 
$l_{i}^{\text{th}}$ element of $\underline{X_i}^{(t)}$.

Similar to the zero-interference and invertibility conditions in Lemma \ref{zero-int}, we have the following theorem for solvability of the network implementing the transform technique.

\begin{theorem}
\label{thm2}
An acyclic network $({\cal G, C})$ with delay, incorporating the transform
techniques, is solvable iff there exists an assignment to $\underline{\varepsilon}$ such that the following conditions are satisfied.
\begin{enumerate}
  \item \textit{Zero-Interference: }$\hat{M}_{ij}^{(t)}(l_i)=0$ for all pairs (source-$ i$, sink-$j$) of nodes such that (source-$ i$, sink-$j$,$\underline{X_i}^{(t)}(l_i)$) $\notin {\cal C}_{j}$ for $0\leq t \leq n-1$.
  \item \textit{Invertibility: }If ${\cal C}_{j}$ comprises the connections 

\[\begin{array}{l}
\left\{\mbox{(source-$ i_1$, sink-$j$, {\small $\underline{X_{i_1}}^{(t)}(l_{i_1})$})},\right.\\
\mbox{ ~(source-$ i_2$, sink-$j$, {\small $\underline{X_{i_2}}^{(k)}(l_{i_2})$})},\\
\hspace{2cm}\vdots\\
\left. \mbox{ ~(source-$ i_{s'}$, sink-$j$, {\small $\underline{X_{i_{s'}}}^{(t)}(l_{i_{s'}})$})} \right\}
\end{array}
\]then, the sub-matrix $[\hat{M}_{i_1j}^{(t)}(l_{i_1})  ~\cdots ~\hat{M}_{i_{s'}j}^{(t)}(l_{i_{s'}}) ]$ is a nonsingular $\nu_j\times\nu_j$ matrix for $0\leq t \leq n-1$.
\end{enumerate}
\end{theorem}
\begin{proof}
Proof is given in Appendix \ref{appen_thm2}.
\end{proof}

The network code which satisfies the invertibility and the zero-interference
conditions for $({\cal G},{\cal C})$ in the transform approach using a suitable
choice of $\alpha$ for the DFT operations is defined as a \textit{feasible
transform network code} for $({\cal G},{\cal C}).$

\subsection{Existence of a network code in the transform approach}

In this section, we prove that under certain conditions there exists a feasible network code for a given $({\cal G},{\cal C})$ if and only if there exists a feasible transform network code. Towards that end, we prove Lemma \ref{nonzeros} which is given below. We first define the polynomial $f(D)$ which will be used henceforth throughout this paper.
\begin{equation}
\label{eqnno1}
f(D)=\prod_{j=1}^{r}det\left(M_j'(D)\right)
\end{equation}where, $M_j'(D)$ is the square submatrix of $M_j(D)$ indicating the source processes that are demanded by sink-$j$.
\begin{lemma}
\label{nonzeros}
Suppose there exists a feasible network code for $({\cal G},{\cal C})$ over some field $\mathbb{F}_{p^m}.$ For some $\alpha \in \mathbb{F}_{p^a}$ (for some positive integer $a$), the LECs defined by the feasible network code for $({\cal G},{\cal C})$ (viewed in the extension field $\mathbb{F}_{p^b}$ where, $b=\text{LCM($a$,$m$)}$) result in a feasible transform network code for $({\cal G},{\cal C})$ if and only if $f(\alpha^t) \neq 0$ for all $0 \leq t \leq n-1.$ 
\end{lemma}
\begin{proof}
Proof is given in Appendix \ref{appen_nonzeros}.
\end{proof}

We now prove the following theorem which concerns with the relationship between
the existence of a feasible network code and a feasible transform network code
for $({\cal G},{\cal C}).$
\begin{theorem}
\label{thmexistence}
Let $({\cal G},{\cal C})$ be the given acyclic delay network with the set of
connections ${\cal C}$ demanded by the sinks. There exists a feasible transform
network code for $({\cal G},{\cal C})$ if and only if there exists a feasible
network code for $({\cal G},{\cal C})$ such that $(D-1)\nmid f(D)$, i.e., $f(1) \neq 0$.
\end{theorem}
\begin{proof}
Proof is given in Appendix \ref{appen_thmexistence}.
\end{proof}

We now present an example acyclic network in which there exists a feasible network code, using which we obtain a feasible transform network code for some choice of $n \geq 7.$
\begin{example}
\label{exm1}
\begin{figure}[htbp]
\centering
\includegraphics[totalheight=2.7in]{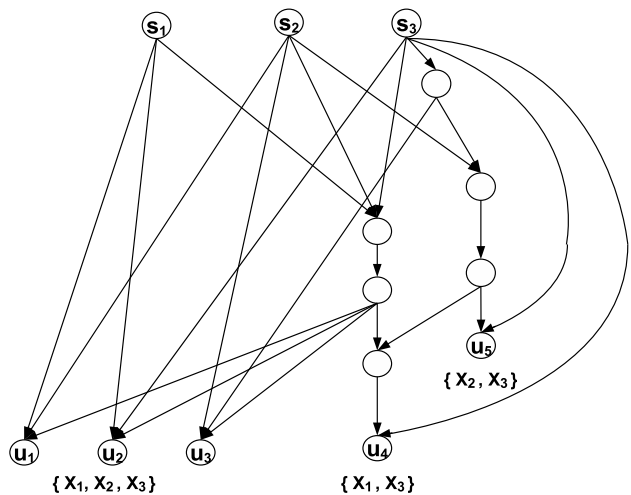}
\caption{A unit-delay network with $3$ sources and $5$ sinks}	
\label{fig1}	
\end{figure}

Consider the network $\cal G$ shown in Fig. \ref{fig1}. This is a unit-delay network (where each edges have a delay of one unit associated with it) taken from \cite{PrR2}. For $1\leq i \leq 3,$ each source $s_i$ has an information sequence $x_i(D).$ This network has non-multicast demands, with sinks $u_j$, $1\leq j \leq 3$, requiring all three information sequences, while sink $u_4$ requires $\left\{x_1(D),x_3(D)\right\}$ and $u_5$ demands $\left\{x_2(D),x_3(D)\right\}.$ Let $\cal C$ denote these set of demands. A feasible network code for $({\cal G},{\cal C})$ over $\mathbb{F}_2$ as obtained in \cite{PrR2} can be obtained by using $1$ as the local encoding coefficient coefficient at all non-sink nodes. The transfer matrix  $M_{u_j}(D),$ the invertible submatrix $M_{u_j}'(D)$ of $M_{u_j}(D),$ and their determinants for the sinks $u_j$, $1 \leq j \leq 5$, are tabulated in Table \ref{tab1}. 
\begin{table*}
\label{tab1}
\centering
\normalsize
\caption{}
\begin{tabular}{|c|c|c|c|}
\hline
\textbf{Sink} & \textbf{Network transfer} & \textbf{Invertible submatrix} &
\textbf{Determinant of} $\boldsymbol{M_{u_j}'(D)},$ \\
&  \textbf{matrix} $\boldsymbol{M_{u_j}(D)}$ &
$\boldsymbol{M_{u_j}'(D)}$\textbf{ of }$\boldsymbol{M_{u_j}(D)}$ &  
$\boldsymbol{det(M_{u_j}'(D))}$ \\ 
\hline
$u_1$ & $\left(\begin{array}{ccc} D & 0 & 0 \\ 0 & D & 0 \\ D^3 & D^3 & D^3
\end{array}\right)$ & $M_{u_1}(D)$ & $D^5$ \\
\hline 
$u_2$ & $\left(\begin{array}{ccc} D & 0 & 0 \\ 0 & 0 & D \\ D^3 & D^3 & D^3
\end{array}\right)$ & $M_{u_2}(D)$ & $D^5$ \\
\hline
$u_3$ & $\left(\begin{array}{ccc} 0 & D & 0 \\ 0 & 0 & D^2 \\ D^3 & D^3 & D^3
\end{array}\right)$ & $M_{u_3}(D)$ & $D^6$ \\
\hline
$u_4$ &  $\left(\begin{array}{ccc} D^4 & 0 & D^4+D^5 \\ 0 & 0 & D
\end{array}\right)$ & $\left(\begin{array}{cc} D^4 & D^4+D^5 \\ 0 & D
\end{array}\right)$ & $D^5$\\
\hline
$u_5$ & $\left(\begin{array}{ccc} 0 & D^3 & D^4 \\ 0 & 0 & D \end{array}\right)$
& $\left(\begin{array}{ccc} D^3 & D^4 \\ 0 & D \end{array} \right)$ & $D^4$ \\
\hline
\end{tabular}
~~\\
~~\\
~\\
\end{table*}

We therefore have $f(D) = D^{25}.$ Note that $f(1)\neq 0$ and $d_{max}=4$ for this network. Therefore, with $n=2^m-1$ for any positive integer $m \geq 3,$ i.e., $\alpha$ being the primitive element of $\mathbb{F}_{2^m},$ we will then have $f(\alpha^t) \neq 0$ for any $0 \leq t \leq n-1.$ By Lemma \ref{nonzeros}, we then have a feasible transform network code for $({\cal G},{\cal C}).$ 
\end{example}

\subsection{Comparison of complexity of the proposed transform approach and the non-transform approach}
Based on the constructive proof of Theorem \ref{thmexistence}, a large field might be required for the existence of a suitable value for $\alpha$ that defines the necessary transform for the network, under the condition that the rate-loss $\left(\frac{d_{max}}{n}\right)$ due to the transform approach be less. The transformed network would then have to be operated over this large field, i.e., the matrices $\hat{M}_{ij}^{(t)}$ have elements from this large field (which is at least a degree $n$ extension over the base field over which the non-transform network code is defined). It is known that (see \cite{GV}, for example) inverting a $\nu_j\times \nu_j$ matrix (at some sink-$j$) takes  $O(\nu_j^3)$ computations, however over the extension field. In the process of computing these inverses, the information symbols corresponding to the $n$ generations are obtained by Gauss-Jordan elimination. In terms of base field computational complexity, the complexity of computing the inverse of the transfer matrix becomes $O 
\left(\nu_j^3n(\log n)(\log\log n)\right)$, as each multiplication in the extension field involves $O\left(n(\log n)(\log \log n)\right)$ computations over the base field  \cite{wiki} (it is equivalent to multiplying two polynomials of degree at least $n-1$ over the base field). The total complexity of recovering the input symbols at all the $n$ generations is then $O\left(n^2\nu_j^3(\log n)(\log \log n)\right)$.

On the other hand, if the non-transform network code is used as such, the transfer matrices $M'_j(D)$ consist of polynomials of degree upto $d_{max}$ in $D$ over the base field. Again, it is known (see \cite{GV}, for example) that finding the inverse of such a matrix has complexity $O(\nu_j^3d_{max})$. To do a fair comparison with the transform case, we consider decoding of $n$-generations ($n$ being large as in the transform case) of information.  Note that inversion of the matrix $M'_j(D)$ does not give us the information polynomials directly. A naive method of obtaining the each information polynomial would then require $\nu_j^2$ multiplications of polynomials over the base field (each of which has complexity $O\left(n(\log n)(\log \log n)\right)$, assuming that $\nu_jd_{max} < n.$) and one division between polynomials (again with complexity $O\left(n(\log n)(\log \log n)\right)$). Therefore, the total complexity involved in recovering the information sequences would then be $O\left(\nu_j^3n(\log n)(\log \
log n)\right) + O\left(\nu_jn(\log n)(\log \log n)\right) + O(\nu_j^3d_{max})$ computations.

Thus, we see that there is an advantage in the complexity of decoding in the non-transform network compared to the transform network (inspite of using the least possible size for the extension field). Therefore, based on the constructive proof of Theorem \ref{thmexistence}, complexity reduction is not an advantage of the transform process. We however note that the construction shown in the proof of Theorem \ref{thmexistence} need not be the only method to construct a transform network code. It remains an open problem to see if a suitable DFT matrix can be defined over a field of smaller size than that suggested by the proposed construction. In that case, the complexity of our transform technique could be lesser than the usual non-transform technique.

The transform technique is useful for PBNA in $3$-S $3$-D MUN-D which will be discussed in the subsequent sections. In the next section, we shall consider the feasibility problem of PBNA using time-varying LECs and PBNA using time-invariant LECs where we apply the transform techniques to obtain precoding matrices similar to the ones in \cite{DVJM}.
\section{PBNA Using Time-Invariant and Time-Varying LECs IN $3$-S $3$-D MUN-D}
\label{sec4}
In \cite{DVJM}, it was shown that for a class of $3$-S $3$-D I-MUN, it is possible to achieve a throughput close to $1/2$ for every source-destination pair via network alignment. In this section, we deal with $3$-S $3$-D MUN-D where each source-destination pair has a min-cut of $1$. In Section \ref{subsec1}, we employ the results from Section \ref{sec3} and show that, even when the zero-interference conditions of Lemma \ref{zero-int} (or Theorem \ref{thm2}) cannot be satisfied, for a class of $3$-S $3$-D MUN-D, we can achieve a throughputs of $\frac{n'+1}{2n'+1}$, $\frac{n'}{2n'+1}$ and $\frac{n'}{2n'+1}$ (for some positive integer $n'$) for the three source-destination pairs by making use of network alignment. The throughputs are close to half when $n'$ is large. This scheme is termed as PBNA using transform approach and time-invariant LECs. In Section \ref{subsec2}, we proceed to generalize the conditions for feasibility of network alignment using time-varying LECs, i.e., we obtain a sufficient condition under which throughputs of $\frac{n_1}{n}$, $\frac{n_2}{n}$ and $\frac{n_3}{n}$ can be achieved for the three source-destination pairs where, $n_1$, $n_2$ and $n_3$ are positive integers less than or equal to $n$. The condition is also a necessary one when $n_1+n_3=n_1+n_2=n$ where it is assumed, without loss of generality, that $n_1 \geq n_2 \geq n_3$. 

Let the random process injected into the network by source $S_i$, $i=1,2,3$, be $X_i(D)$. Source $S_i$ needs to communicate only with destination $T_i$. Here, $\mu_i = 1$ and $\nu_j = 1$, $i,j=1,2,3$.

We shall consider the following two cases separately.
\begin{enumerate}
 \item The min-cut between $S_i$ and $T_j$ is greater than or equal to
$1$, for all $i \neq j$.
 \item The min-cut between $S_i$ and $T_j$ is equal to $0$, for some $i
\neq j$.
\end{enumerate}
$\\$
\textbf{Case 1:} The min-cut between source-$i$ and sink-$j$ is greater than or
equal to $1$, for all $i \neq j$.

\subsection{PBNA using transform approach and time-invariant LECs}
\label{subsec1}
Consider a transmission scheme where, in order to transmit $2n'+1$ ($>> d_{max}$) generations of input symbols at each source, the cyclic prefix comprising $d_{max}$ generations is transmitted first, followed by the $2n'+1$ generations of input symbols. Let $Q_1X_i^{2n'+1}$ be the input symbols transmitted by source $i$, where,
\begin{align}
\nonumber
X_i^{2n'+1} ~= ~[X_i^{(2n')} ~X_i^{(2n'-1)} ~\cdots ~X_i^{(0)}]^T
\end{align}
Also, let 
\begin{align*}
 X_1^{2n'+1}=V_1{X_1^\prime}^{n'+1}, X_2^{2n'+1}=V_2{X_2^\prime}^{n'} \text{, and } X_3^{2n'+1}=V_3{X_3^\prime}^{n'}
\end{align*}
where, $V_1$ is a $(2n'+1)\times(n'+1)$ matrix,  $V_2$ is a $(2n'+1)\times n'$ matrix,  $V_3$ is a $(2n'+1)\times n'$ matrix, and
\begin{align*}
{X_1^\prime}^{n'+1}&=[{X_1^\prime}^{(0)} ~{X_1^\prime}^{(1)} ~\cdots
~{X_1^\prime}^{(n')}]^T\\
{X_2^\prime}^{n'} &=[{X_2^\prime}^{(0)} ~{X_2^\prime}^{(1)} ~\cdots
~{X_2^\prime}^{(n'-1)}]^T\\
{X_3^\prime}^{n'} &=[{X_3^\prime}^{(0)} ~{X_3^\prime}^{(1)} ~\cdots
~{X_3^\prime}^{(n'-1)}]^T.
\end{align*}
The quantities ${X_1^\prime}^{n'+1}$, ${X_2^\prime}^{n'}$ and ${X_3^\prime}^{n'}$ denote the $(n'+1)$, $n'$, and $n'$ independent input symbols generated by $S_1$, $S_2$ and $S_3$ respectively. From (\ref{eqnrefer1unicast}) we have, for $j=1,2,3$,
\begin{align}
\nonumber
{Y_j}^{2n'+1}=  \hat{M}_{1j} V_1 {X_1^\prime}^{n'+1} + \hat{M}_{2j} V_2
{X_2^\prime}^{n'} + \hat{M}_{3j} V_3 {X_3^\prime}^{n'},
\end{align}
where, ${Y_j}^{2n'+1}$ denotes the $(2n'+1)$ output symbols at sink-$j$.
The objective is to recover the $(n'+1)$ independent input symbols of $S_1$, $n'$ independent input symbols of $S_2$, and $n'$ independent input symbols of $S_3$ at $T_1$, $T_2$, and $T_3$ from $Y_1^{2n'+1}$, $Y_2^{2n'+1}$, and $Y_3^{2n'+1}$ respectively.

For acyclic networks without delay, the network alignment concept in \cite{DVJM} involved varying LECs at every time instant. But with delays it is possible, in some cases, to achieve network alignment even with time-invariant LECs. This is what we show in this sub-section.

First, note that the elements of the matrices $\hat{M}_{ij}$ are functions of $\underline{\varepsilon}$.
\begin{lemma}
\label{lemma2}
Determinants of the matrices $\hat{M}_{ij}$, $i,j =1,2,3$, are non-zero polynomials in $\underline{\varepsilon}$.
\end{lemma}
\begin{proof}
Proof is given in Appendix \ref{appen_lemma2}.
\end{proof}

Let
\begin{align}
\nonumber
\hat{U}&=\hat{M}_{12}^{-1}\hat{M}_{32}\hat{M}_{31}^{-1}\hat{M}_{21}\hat{M}_{23}^{-1}\hat{M}_{13},\\
\label{mat_T}
\hat{R}&=\hat{M}_{13}\hat{M}_{23}^{-1}, ~\hat{S}=\hat{M}_{12}\hat{M}_{32}^{-1}.
\end{align} Now, choose
\begin{align}
\label{thm2_pf7}
&V_1 = [W ~\hat{U}W ~\hat{U}^2W ~\cdots \hat{U}^{n'}W]\\
\label{thm2_pf8}
&V_2 = [\hat{R}W ~\hat{R}\hat{U}W ~\hat{R}\hat{U}^2W ~\cdots ~\hat{R}\hat{U}^{n'-1}W]\\
\label{thm2_pf9}
&V_3 = [\hat{S}\hat{U}W ~\hat{S}\hat{U}^2W ~\cdots ~\hat{S}\hat{U}^{n'}W]
\end{align}where, $W=[1 ~1 ~\cdots ~1]^T$ (all ones vector of size $(2n'+1) \times 1$). Since the transform approach requires that $2n'+1 | p^m-1$, we shall find it useful in stating the exact relationship between $2n'+1$ and $p$ which will be used in the result that follows.

\begin{lemma}\label{lemma_newly_added}
 The positive integer $2n'+1$ divides $p^m-1$ for some positive integer $m$ iff $p \nmid 2n'+1$.
\end{lemma}
\begin{proof}
Proof is given in Appendix \ref{appen_lemma_newly_added}.
\end{proof}

\begin{theorem}
\label{thm3}
The input symbols ${X_1^\prime}^{n'+1}$, ${X_2^\prime}^{n'}$ ,and ${X_3^\prime}^{n'}$ can be exactly recovered at $T_1$, $T_2$, and $T_3$ from the output symbols $Y_1^{2n'+1}$, $Y_2^{2n'+1}$, and $Y_3^{2n'+1}$ respectively subject to $p \nmid 2n'+1$, if the following conditions hold.
\begin{align}
\label{thm3_cond1}
&\mbox{Rank}[V_1 ~~\hat{M}_{11}^{-1}\hat{M}_{21}V_2]=2n'+1\\
\label{thm3_cond2}
&\mbox{Rank}[\hat{M}_{12}^{-1}\hat{M}_{22}V_2 ~~V_1]=2n'+1\\
\label{thm3_cond3}
&\mbox{Rank}[\hat{M}_{13}^{-1}\hat{M}_{33}V_3 ~~V_1]=2n'+1
\end{align}
\end{theorem}
\begin{proof}
Proof is given in Appendix \ref{appen_thm3}.
\end{proof}

When the conditions of the above Theorem are satisfied, we say that PBNA using transform approach and time-invariant LECs is feasible. When PBNA using transform approach and time-invariant LECs is feasible, throughputs of $\frac{(n'+1)}{(2n'+1)}$, $\frac{n'}{(2n'+1)}$, and $\frac{n'}{(2n'+1)}$ are achieved for the source-destination pairs $S_1-T_1$, $S_2-T_2$, and $S_3-T_3$ respectively. When $n'$ is large, the throughputs are close to half. The throughput loss due to the addition of cyclic prefix is not accounted for, since it is assumed that $2n'+1 >> d_{max}$. 

It will be shown in Section \ref{subsec4} that the conditions of Theorem \ref{thm3} are also necessary conditions for feasibility of PBNA using transform approach and time-invariant LECs, i.e., the choice of the precoding matrices in (\ref{thm2_pf7})-(\ref{thm2_pf9}) do not restrict the conditions for network alignment.

\begin{remark} 
In a $3$-S $3$-D I-MUN considered in \cite{DVJM}, it was not possible to achieve network alignment without changing the LECs with time.  When there is no delay, the matrices $\hat{U}$, $\hat{R}$, and $\hat{S}$ given in (\ref{mat_T}), would simply be equal to $f(\underline{\varepsilon})I_{2n'+1}$ (where, $f(\underline{\varepsilon})$ is some polynomial in $\underline{\varepsilon}$) and hence, the matrices $V_1$, $V_2$ and $V_3$ as given in (\ref{thm2_pf7})-(\ref{thm2_pf9}) are themselves not full-rank matrices. Hence, $\underline{\varepsilon}$ was varied with time in \cite{DVJM}. However, in the case of delay it is easy to see from the structure of the matrix $\hat{M}_{ij}$ that the matrices $\hat{U}$, $\hat{R}$, and $\hat{S}$ are not necessarily scaled identity matrices. 
\end{remark}

The following example, taken from \cite{DVJM} (but considered with delays), illustrates the existence of a network where network alignment is feasible with time-invariant LECs.

\begin{example} \label{eg-PBNA_TINV_LEC}
Consider the network shown in Fig. \ref{NA_eg}. Each link is taken to have unit-delay. 
\begin{figure}[htbp]
\centering
\includegraphics[totalheight=2.7in]{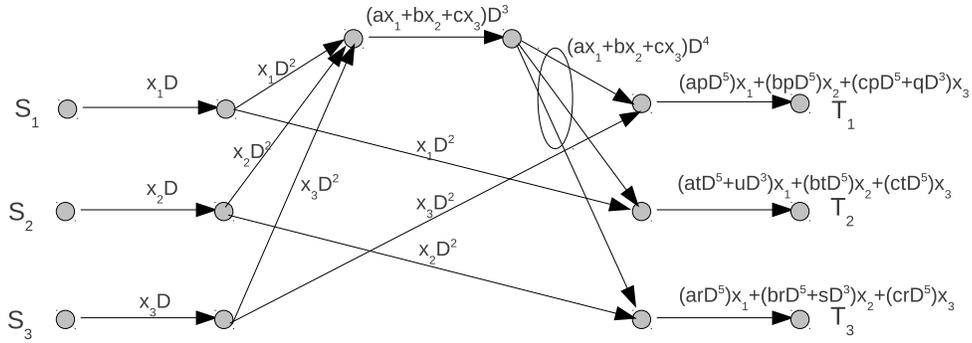}
\caption{A $3$-S $3$-D MUN-D where PBNA using transform approach and time-invariant LECs is feasible.}	
\label{NA_eg}	
\end{figure}
In accordance with the LECs denoted as in the figure, the transfer matrices $M_{ij}(D)$ are as given below.
\begin{align*}
&M_{11}(D)=M_{11}^{(5)}D^{5}=apD^{5},\\
&M_{12}(D)=M_{12}^{(3)}D^{3}+M_{12}^{(5)}D^{5}=uD^{3}+atD^{5},\\
&M_{13}(D)=M_{13}^{(5)}D^{5}=arD^{5},\\
&M_{21}(D)=M_{21}^{(5)}D^{5}=bpD^{5},\\
&M_{22}(D)=M_{22}^{(5)}D^{5}=btD^{5}\\
&M_{23}(D)=M_{23}^{(3)}D^{3}+M_{23}^{(5)}D^{5}=sD^{3}+brD^{5},\\
&M_{31}(D)=M_{31}^{(3)}D^{3}+M_{31}^{(5)}D^{5}=qD^{3}+cpD^{5},\\
&M_{32}(D)=M_{32}^{(5)}D^{5}=ctD^{5},\\
&M_{33}(D)=M_{33}^{(5)}D^{5}=crD^{5}.
\end{align*}
\noindent As explained in Section \ref{sec2}, without loss, the network transfer function between $S_i-T_j$ can be taken to be equal to $M_{ij}D^{-2}$. Note that the network does not satisfy the zero-interference conditions of Lemma \ref{zero-int}. Here, $d_{max}=2$. Consider the following (random) assignment to the LECs.
\begin{align*}
&a=b=c=p=r=t=1\\
&s=1+{\beta}^2+{\beta}^3+{\beta}^4+{\beta}^5\\
&q=1+{\beta}+{\beta}^2\\
&u=1+{\beta}^4
\end{align*}where, $\beta$ is a primitive element of $GF(2^6)$ whose minimal polynomial is given by $(1+x+x^6)$. The DFT parameter $\alpha$ is given by $\alpha=\beta^9$ and the number of symbol extensions is given by $2n'+1=7$. The transformed network transfer matrices are given by
\begin{align*}
&\hat{M}_{11}=diag\left(1,\alpha^2,\cdots,\alpha^{12} \right),\\
&\hat{M}_{12}=diag\left(\left(1+{\beta}^4\right)+1,\left(1+{\beta}^4\right)+\alpha^2,\cdots,\left(1+{\beta}^4\right)+\alpha^{12} \right),\\
&\hat{M}_{13}=diag\left(1,\alpha^2,\cdots,\alpha^{12} \right),\\
&\hat{M}_{21}=diag\left(1,\alpha^2,\cdots,\alpha^{12} \right),\\
&\hat{M}_{22}=diag\left(1,\alpha^2,\cdots,\alpha^{12}\right),\\
&\hat{M}_{23}=diag\left(\sum_{j=0}^5{\beta}^j+1,\sum_{j=0}^5{\beta}^j+\alpha^2,\cdots,\sum_{j=0}^5{\beta}^j+\alpha^{12} \right),\\
&\hat{M}_{31}=diag\left(\left(1+{\beta}+{\beta}^2\right)+1,\left(1+{\beta}+{\beta}^2\right)+\alpha^2,\cdots,\left(1+{\beta}+{\beta}^2\right)+\alpha^{12} \right),\\
&\hat{M}_{32}=diag\left(1,\alpha^2,\cdots,\alpha^{12} \right),\\
&\hat{M}_{33}=diag\left(1,\alpha^2,\cdots,\alpha^{12} \right).
\end{align*}

It can be verified using the software {\em Mathematica}
\footnote{A Galois Field package for {\em Mathematica} is available at \cite{mathematica_pack}.} that the rank conditions of Theorem \ref{thm3} (in (\ref{thm3_cond1})-(\ref{thm3_cond3})) are satisfied using the above assignment to the LECs and $\alpha$.
\end{example}

\subsection{PBNA with time-varying LECs}
\label{subsec2}
The feasibility problem for PBNA with time-varying LECs is stated as follows. Source $S_i$ demands a throughput of $\frac{n_i}{n}$ where, $n$ is a positive integer and $n_i$, $i = 1,2,3$, are positive integers less than or equal to $n$. Without loss of generality, we assume that $n_1 \geq n_2 \geq n_3$. We need to determine if the throughput demands can be met through a PBNA scheme similar to the one described in the previous sub-section while permitting the use of time-varying LECs. The solution to this problem will also generalize Theorem \ref{thm3}. Moreover, there can exist $3$-S $3$-D MUN-D where PBNA using transform approach and time-invariant LECs is infeasible for all $n'$ while PBNA using time-varying LECs is feasible for some positive integer tuple $(n_1,n_2,n_3,n)$. Example \ref{eg_TVL_F} in Section \ref{subsec4} is an instance of such a network.

We shall observe in this sub-section that, unlike in the case of time-invariant LECs, the network cannot be decomposed into instantaneous networks using the transform method. Throughout the sub-section we shall assume that the LECs and the other variables that we shall encounter belong to the algebraic closure of the field $\mathbb{F}_p$ which is denoted by $\overline{\mathbb{F}_{p}}$. Clearly, once an assignment to the LECs and variables are made, they 
belong to a finite extension of $\mathbb{F}_p$.

Consider a transmission scheme, where we take $n$ ($>> d_{max}$) generations of input symbols at each source and first transmit last $d_{max}$ generations (i.e., the cyclic prefix) followed by the $n$ generations of input symbols. Let $X_i^{n}$ be the input symbol that needs to be transmitted by $S_i$ where,
\begin{align}
\nonumber
X_i^{n} ~= ~[X_i^{(n-1)} ~X_i^{(n-2)} ~\cdots ~X_i^{(0)}]^T.
\end{align}
Let { $X_1^{n}=V_1{X_1^\prime}^{n_1}$}, {$X_2^{n}=V_2{X_2^\prime}^{n_2}$}, and {$X_3^{n}=V_3{X_3^\prime}^{n_3}$} where, $V_1$ is a $n\times n_1$ matrix,  $V_2$ is a $n\times n_2$ matrix,  $V_3$ is a $n\times n_3$ matrix, and
\begin{align}
\nonumber
{X_1^\prime}^{n_1}&=[{X_1^\prime}^{(0)} ~{X_1^\prime}^{(1)} ~\cdots
~{X_1^\prime}^{(n_1-1)}]^T,\\
\nonumber
{X_2^\prime}^{n_2}&=[{X_2^\prime}^{(0)} ~{X_2^\prime}^{(1)} ~\cdots
~{X_2^\prime}^{(n_2-1)}]^T,\\
\nonumber
{X_3^\prime}^{n_3}&=[{X_3^\prime}^{(0)} ~{X_3^\prime}^{(1)} ~\cdots
~{X_3^\prime}^{(n_3-1)}]^T.
\end{align}
The quantities ${X_1^\prime}^{n_1}$, ${X_2^\prime}^{n_2}$, and ${X_3^\prime}^{n_3}$ denote the $n_1$, $n_2$, and $n_3$ independent input symbols generated by $S_1$, $S_2$, and $S_3$ respectively. Thus, the independent input symbols are coded over $n$ time slots by the matrices $V_1$, $V_2$, and $V_3$ before they are transmitted over the network after the addition of cyclic prefix. Now, from (\ref{trmx_tvarleks}) and following the same steps as involved in writing (\ref{bigeqn}) and (\ref{bigeqn2}),
for $j =1,2,3$, we get
\begin{align}
\nonumber
{Y_j}^{n}=  {M}_{1j} V_1 {X_1^\prime}^{n_1} + {M}_{2j} V_2 {X_2^\prime}^{n_2} +{M}_{3j} V_3 {X_3^\prime}^{n_3}
\end{align}where, ${Y_j}^{n}$ denotes the $n$ output symbols at $T_j$ and ${M}_{ij}$ is as given in  (\ref{bigeqn1_unicast}) (at the top of the next page). The structure of $M_{ij}$ is such that it becomes a circulant matrix when the LECs are time-invariant, i.e., $\underline{\varepsilon}^{(-d_{max})}=\underline{\varepsilon}^{(-d_{max}+1)}
=\ldots= \underline{\varepsilon}^{(n-1)}$.
\begin{figure*}
\scriptsize
\begin{align}
\label{bigeqn1_unicast}
&\left[\begin{array}{ccccc}
  M_{ij}^{(0)}({\underline{\varepsilon}^{(n-1,n-1)}}) &
M_{ij}^{(1)}({\underline{\varepsilon}^{(n-2,n-1)}}) & \cdots &
M_{ij}^{(d_{max}-1)}({\underline{\varepsilon}^{(n-d_{max},n-1)}}) &
M_{ij}^{(d_{max})}({\underline{\varepsilon}^{(n-1-d_{max},n-1)}})\\
  0 & M_{ij}^{(0)}({\underline{\varepsilon}^{(n-2,n-2)}}) & \cdots &
M_{ij}^{(d_{max}-2)}({\underline{\varepsilon}^{(n-d_{max},n-2)}}) &
M_{ij}^{(d_{max}-1)}({\underline{\varepsilon}^{(n-1-d_{max},n-2)}})\\
  \vdots & \vdots & \vdots & \vdots & \vdots\\
  M_{ij}^{(1)}({\underline{\varepsilon}^{(-1,0)}}) &
M_{ij}^{(2)}({\underline{\varepsilon}^{(-2,0)}}) & \cdots &
M_{ij}^{(d_{max})}({\underline{\varepsilon}^{(-d_{max},0)}}) &  0\\
\end{array}\right.\\
\nonumber
&~\underbrace{\hspace{10.25cm}\left.\begin{array}{cccc}
	0 & 0  & \cdots & 0 \\
	M_{ij}^{(d_{max})}({\underline{\varepsilon}^{(n-2-d_{max},n-2)}}) & 0
& \cdots& 0 \\
	\vdots & \vdots  & \vdots & \vdots \\
	0 & 0  & \cdots &  M_{ij}^{(0)}({\underline{\varepsilon}^{(0,0)}})\\
      \end{array}\right]}_{M_{ij}}
\end{align}
\hrule
\end{figure*}
The objective is to determine if the $n_i$ independent input symbols of $S_i$ can be recovered at $T_i$, from $Y_i^{n}$.

Note that the matrices $M_{ij}$ are not a circulant matrices and therefore, cannot be simultaneously diagonalized in general. Let \mbox{$\underline{\varepsilon}'=\left[\underline{\varepsilon}^{(-d_{max})} ~~\underline{\varepsilon}^{(-d_{max}+1)} ~~\cdots ~~\underline{\varepsilon}^{(n-1)}\right]$}.
\begin{lemma}
 \label{lemma3}
Determinant of the matrix ${M}_{ij}$, for all $(i,j)$, is a non-zero polynomial in $\underline{\varepsilon}'$.
\end{lemma}
\begin{proof}
Proof is given in Appendix \ref{appen_lemma3}.
\end{proof}

As a direct consequence of the above lemma, the inverses of $M_{ij}$s exist. Now, let the elements of $V_1$ be given by
\begin{align}
\label{na2_V1}
[V_1]_{ij}=\theta_{ij}; ~i =1,2,\cdots,n, ~j =1,2,\cdots,n_1
\end{align}
where $\theta_{ij}$ is a variable that takes values from $\overline{{\mathbb F}_p}$. Let
\begin{align}
\label{na2_V2_and_V3}
 V_2 = M_{23}^{-1}M_{13}V_1A \text{ and }V_3 =M_{32}^{-1}M_{12}V_1B
\end{align}
\noindent where, the elements of the matrices $A$ and $B$, of sizes $n_1 \times n_2$ and $n_1 \times n_3$, are given by $[A]_{ij}=a_{ij}$ and $[B]_{ij}=b_{ij}$ respectively ($a_{ij}$ and $b_{ij}$ are variables that take values from $\overline{{\mathbb F}_p}$). Let
\begin{align} \label{eqn-na2_matU}
U=M_{12}^{-1}M_{32}M_{31}^{-1}M_{21}M_{23}^{-1}M_{13}.
\end{align} 
Let $f^{(k)}_1(V_1,\underline{\varepsilon}',A)$ and $f^{(k)}_2(V_1,\underline{\varepsilon}',A)$  denote the determinants of {\small $(n_1+n_2) \times (n_1+n_2)$} submatrices of {\small $[V_1 ~~{M}_{11}^{-1}{M}_{21}V_2]$} and {\small $[{M}_{12}^{
-1}{M}_{22}V_2 ~~V_1]$} respectively, for {\small $k = 1,2,\ldots,{n \choose n_1+n_2}$}. Similarly, let $f^{(k)}_3(V_1,\underline{\varepsilon}',B)$ denote the determinants of \mbox{\small $(n_1+n_3) \times (n_1+n_3)$} submatrices of {\small$[{M}_{13}^{-1}{M}_{33}V_3 ~~V_1]$}, for {\small $k =1,2,\ldots,{n \choose n_1+n_3}$}. Now, define

{\small \vspace{-0.3cm}
\begin{align}
\nonumber
f_1(V_1,\underline{\varepsilon}',A)&=1-\prod_{k=1}^{{n \choose n_1+n_2}}\left(1-\delta_1^{(k)}f^{(k)}_1(V_1,\underline{\varepsilon}',A)\right)\\
\nonumber
f_2(V_1,\underline{\varepsilon}',A)&=1-\prod_{k=1}^{{n \choose n_1+n_2}}\left(1-\delta_2^{(k)}f^{(k)}_2(V_1,\underline{\varepsilon}',A)\right)\\
\label{newqn1}
f_3(V_1,\underline{\varepsilon}',B)&=1-\prod_{k=1}^{{n \choose n_1+n_3}}\left(1-\delta_3^{(k)}f^{(k)}_3(V_1,\underline{\varepsilon}',B)\right)\\
\nonumber
f_4(\underline{\varepsilon}')&=\prod_{(i,j)\in\{1,2,3\}}det(M_{ij})\\
\nonumber
f(V_1,\underline{\varepsilon}',A,B) &=
f_1(V_1,\underline{\varepsilon}',A)f_2(V_1,\underline{\varepsilon}',A)f_3(V_1,\underline{\varepsilon}',B) f_4(\underline{\varepsilon}')
\end{align}}where, $\delta_i^{(k)} \in \overline{\mathbb{F}_{q}}$, for all $(i,k)$.
\noindent Denote the elements of a matrix $C$ of size $n_2 \times n_3$ by $[C]_{ij}=c_{ij}$ where, $c_{ij}$ is a variable that takes values from ${\mathbb{F}_{q}}$, for all $(i,j)$. For $i = 1,2,\cdots,n_2$ and $j = 1,2,\cdots,n_3$, let 
\begin{align} \label{eqn-gij}
&g_{ij}(V_1,\underline{\varepsilon}',A,B,C)=[UV_1AC]_{ij}-[V_1B]_{ij}.
\end{align}
\noindent Let
$g^{(nr)}_{ij}(V_1,\underline{\varepsilon}',A,B,C)$ and $g^{(dr)}_{ij}(V_1,\underline{\varepsilon}',A,B,C)$ denote the numerator and denominator respectively of the rational-polynomial $g_{ij}(V_1,\underline{\varepsilon}',A,B,C)$. Similarly, let $f^{(nr)}(V_1,\underline{\varepsilon}',A,B)$ and $f^{(dr)}(V_1,\underline{\varepsilon}',A,B)$ denote the numerator and denominator respectively of the rational polynomial $f(V_1,\underline{\varepsilon}',A,B)$. We shall denote  $g_{ij}(V_1,\underline{\varepsilon}',A,B,C)$ and $f(V_1,\underline{\varepsilon}',A,B)$ respectively as $g_{ij}$ and $f$ for short. Similar notation is used for the numerator and denominator of the respective rational polynomials.

\begin{theorem}
\label{thm4}
For an acyclic $3$-S $3$-D MUN-D, the input symbols ${X_1^\prime}^{n_1}$, ${X_2^\prime}^{n_2}$, and ${X_3^\prime}^{n_3}$
can be exactly recovered at $T_1$, $T_2$, and $T_3$ from the output symbols $Y_1^{n}$, $Y_2^{n}$, and $Y_3^{n}$  respectively if the ideal generated by the polynomials  $g^{(nr)}_{ij}$, $i =1,2,..,n$ and $j = 1,2,..,n_3$, and
$\left(1-\delta f^{(nr)}f^{(dr)} \prod_{(i,j)}g^{(dr)}_{ij}\right)$ does not include $1$, where $\delta$ is a variable that can take value from $\overline{{\mathbb F}_{q}}$. The condition is also necessary when $(n_1+n_2)=(n_1+n_3)=n$.
\end{theorem}
\begin{proof}
Proof is given in Appendix \ref{appen_thm4}.
\end{proof}
When the conditions of the above Theorem are satisfied, we say that PBNA using time-varying LECs is feasible. When PBNA using time-varying LECs is feasible, as $n >> d_{max}$, throughputs close to $\frac{n_1}{n}$, $\frac{n_2}{n}$, and $\frac{n_3}{n}$ are achieved for the source-destination pairs $S_1-T_1$, $S_2-T_2$, and $S_3-T_3$ respectively. As seen from the proof of the above theorem, if the throughput demands are such that $n_1+n_3 > n$ or $n_1+n_2 > n$ then, PBNA using time-varying LECs is infeasible.

\begin{remark}
\label{remark2}
Theorem \ref{thm4} implies that network alignment is feasible if there exists an assignment to the LECs and the other variables such that $f^{(nr)}f^{(dr)}\prod_{(i,j)}g^{(dr)}_{ij}$ takes a non-zero value and $g^{(nr)}_{ij}$ take values of zero for network alignment to be feasible. Firstly, $f^{(nr)}$ has to be a non-zero polynomial which requires that $V_1$ be a full-rank matrix. This is true from the choice of $V_1$. Also, ${M}_{11}^{-1}{M}_{21}V_2$, ${M}_{12}^{-1}{M}_{22}V_2$ and ${M}_{13}^{-1}{M}_{33}V_3$ should also be full-rank. Since the matrices ${M}_{ij}$ are invertible, it is equivalent to checking if $V_2$ and $V_3$ are full-rank. This is also true because $V_1$ is a full-rank matrix, and by choosing $A$ and $B$ as matrices that select the first $n_2$ columns of $V_1$ and the first $n_3$ columns of $V_1$ respectively, $V_2$ and $V_3$ become full-rank. Hence, the determinants of all the $n_2 \times n_2$ and $n_3 \times n_3$ sub-matrices of $V_2$ and $V_3$ respectively are non-zero polynomials. So, 
we have at least ensured that, by proper choice, ${M}_{11}^{-1}{M}_{21}V_2$, ${M}_{12}^{-1}{M}_{22}V_2$ and ${M}_{13}^{-1}{M}_{33}V_3$ are full-rank matrices.
\end{remark}
\begin{remark}
The network alignment matrices in Section \ref{subsec1} can be derived as a special case of the network alignment matrices in Section \ref{subsec2} and Theorem \ref{thm3} can be derived as a special case of Theorem \ref{thm4} as explained below. Choose
$\underline{\varepsilon}^{(-d_{max})}=\underline{\varepsilon}^{(-d_{max}+1)}
=...= \underline{\varepsilon}^{(2n')}=\underline{\varepsilon}$ and $n=2n'+1$, $n_1=n'+1$, $n_2=n'$, and $n_3=n'$. Also, choose the
variables $\theta_{ij}$ such that $V_1$ in (\ref{na2_V1}) takes the form of $V_1$ in (\ref{thm2_pf7}). Choose $A$ and $B$, respectively, to be the selection matrices which select the first $n'$ columns and the last $n'$ columns of the matrices pre-multiplying them. Let $C=I_{n'}$. Since the input symbols at the sources were precoded by $Q_1$ and the output symbols at the destinations were pre-multiplied by $Q_1^{-1}$, the effective transfer matrix between $S_i$ and $T_j$ is given by $\hat{M}_{ij}$. Hence, $(UV_1AC-V_1B)$ becomes equal to the zero matrix. It can also be easily seen that the full-rank conditions in Theorem \ref{thm3} are the same as stating that the ideal generated by $\left(1-\delta f^{(nr)}f^{(dr)}\prod_{(i,j)}g^{(dr)}_{ij}\right)$ should not include $1$.
\end{remark}

A systematic method of verifying the condition in Theorem \ref{thm4} is by computing the reduced Groebner basis for the given ideal with a chosen monomial ordering. The condition is satisfied iff $1$ is an element of the reduced Groebner basis \cite{CLO-book}. However, in general, Groebner basis algorithms are known to have large exponential complexity in the number of variables and solving multivariate polynomial equations is known to be NP-hard over any field \cite{CLO-book} \cite{CKPS}. Hence, the conditions of Theorem \ref{thm3} are easier to check than the condition of Theorem \ref{thm4}.

\textbf{Case 2:} The min-cut between source-$i$ and sink-$j$ is equal to $0$,
for some $i \neq j$. 

This means that at least one of the matrices $M_{ij}$, for $i \neq j$, is a zero-matrix. The choices of $V_1$ , $V_2$ and $V_3$ and the conditions for network alignment for this case are similar to the ones presented in Section \ref{subsec2}. The only major difference will be the absence of conditions on the lines that there must exist an assignment to the LECs and the other variables such that the rational polynomials $g_{ij}$ take values of zero for network alignment to be feasible. There are various possibilities in this case. We present feasibility conditions for one possibility (i.e. min-cut between $S_2$-$T_1$ is equal to $0$) and the rest are fairly straight-forward to derive. We assume the same set-up as in Section \ref{subsec2}.

{\em Min-cut between $S_2$-$T_1$ is equal to $0$}: This implies that $M_{21}=0$. Let the elements of $V_1$ be given by

{
\vspace{-0.3cm}
\begin{align}
\label{mincut0_poss1_V1}
[V_1]_{ij}=\theta_{ij}, ~ i = 1,2,\cdots,n, ~j =1,2,\cdots,n_1
\end{align}}where, $\theta_{ij}$ is a variable that takes values from $\overline{{\mathbb F}_{q}}$. Let

{
\vspace{-0.3cm}
\begin{align}
\label{mincut0_poss1_V2_and_V3}
 V_2 = M_{23}^{-1}M_{13}V_1A \text{ and} ~~V_3 =M_{32}^{-1}M_{12}V_1B
\end{align}}where, the elements of the matrices $A$ and $B$, of sizes $n_1 \times n_2$ and $n_1 \times n_3$, are given by $[A]_{ij}=a_{ij}$ and $[B]_{ij}=b_{ij}$ respectively ($a_{ij}$ and $b_{ij}$ are variables that take values from $\overline{{\mathbb F}_{q}}$). The following theorem provides the conditions under which network alignment is feasible.

\begin{theorem}
\label{thm5}
For an acyclic $3$-S $3$-D MUN-D, when the min-cut between $S_2$-$T_1$ is equal to $0$ and the min-cut between the other sources and destinations are not zero, the input symbols ${X_1^\prime}^{n_1}$, ${X_2^\prime}^{n_2}$, and ${X_3^\prime}^{n_3}$ can be exactly recovered at $T_1$, $T_2$, and $T_3$ from the output symbols $Y_1^{n}$, $Y_2^{n}$, and $Y_3^{n}$ respectively, if

{
\vspace{-0.4cm}
\begin{align*}
&\mbox{Rank}[V_1 ~~{M}_{11}^{-1}{M}_{31}V_3]=n_1+n_2,\\
&\mbox{Rank}[{M}_{12}^{-1}{M}_{22}V_2 ~~V_1]=n_1+n_2,\\
&\mbox{Rank}[{M}_{13}^{-1}{M}_{33}V_3 ~~V_1]=n_1+n_3.
\end{align*}}The above conditions are also necessary when $(n_1+n_2)=(n_1+n_3)=n$.
\end{theorem}
\begin{proof}
 Proof is given in Appendix \ref{appen_thm5}.
\end{proof}

When the conditions of the above Theorem are satisfied, throughputs close to $\frac{n_1}{n}$, $\frac{n_2}{n}$, and $\frac{n_3}{n}$ are achieved for the source-destination pairs $S_1-T_1$, $S_2-T_2$, and $S_3-T_3$ respectively.

In the next section, we introduce PBNA using transform approach and block time-varying LECs where we show that the reduced feasibility conditions of Meng et al. \cite{MRMJ} for feasibility of PBNA in 3-S 3-D I-MUN are also necessary and sufficient for feasibility of PBNA using transform approach and block time-varying LECs in $3$-S $3$-D MUN-D.

\section{PBNA using transform approach and block time-varying LECs}
\label{sec5}
In this section, we propose a PBNA scheme for $3$-S $3$-D MUN-D, different from those given in Section \ref{subsec1} and \ref{subsec2}. The min-cut between $S_i-T_i$, for all $i$, is assumed to be equal to $1$. We restrict ourselves to the field $\mathbb F_{2^m}$ in this section and also the following section. For the PBNA scheme proposed in this section, we shall show that the feasibility condition is the same as that proposed for instantaneous networks in \cite{MRMJ}. In addition, the feasibility condition is independent of the number of symbol extensions over which the independent input symbols are precoded unlike in the case of the other two proposed PBNA schemes.

Consider the following transmission where, every source $S_i$ is required to transmit a $k(2n'+1)$-length block of symbols $(k >>d_{max})$ given by 
{\small $
 [X^{(0)}_i ~X^{(1)}_i \cdots ~X^{(k(2n'+1)-1)}_i]^T
$} for some positive integer $n'>0$. Partition the block of symbols into $(2n'+1)$ blocks, each of length $k$ symbols. For each block of $k$ symbols, we add a cyclic prefix of length $d_{max}$. The partitioning of the input symbols and the addition of cyclic prefix (CP) are shown in Fig. \ref{fig:partition-CP-zeros}.

\begin{figure}[htbp]
\centering
\includegraphics[height=4.0in,width=5.5in]{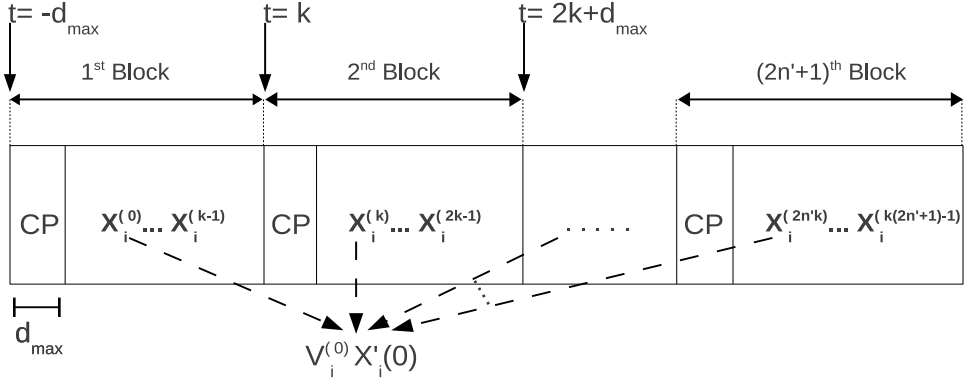}
\vspace{-2.2cm}
\caption{The figure demonstrates the transmission of $(2n'+1)$ blocks of symbols, involving addition of CP for every block at $S_i$. The pre-multiplication of each block of symbols by $F$ (not explicitly shown in the figure) is done after the precoding step and before the addition of CP.}
\label{fig:partition-CP-zeros}
\vspace{0.5cm}
\hrule
\end{figure}
The LECs of the network are varied with every $(k+d_{max})$ time instants starting from the time instant $t$$=$$-d_{max}$. Therefore, when $S_i$ transmits its first block of data as shown in Fig. \ref{fig:partition-CP-zeros}, the LECs remain constant and when it starts the transmission of the second block of data, the LECs encountered in the network are different.

At each destination $T_i$, the first $d_{max}$ outputs in each received block of length $(k+d_{max})$ symbols, starting from time instant $t=-d_{max}$, is discarded. Denote the LECs during $l^{\text{th}}$-block transmission by $\underline \varepsilon_l$, for $1 \leq l \leq (2n'+1)$. Now, consider the second block of output symbols (i.e., $l=2$) at $T_j$ after discarding the cyclic prefix. Since the LECs remain constant during one block of transmission, from (\ref{trmx_tvarleks}) and (\ref{bigeqn}), we get (\ref{bigeqn3}). As in (\ref{bigeqn2}), (\ref{bigeqn3}) is re-written as (\ref{bigeqn4}).
\begin{figure*}
\scriptsize
\begin{align}\label{bigeqn3}
\begin{bmatrix}
   {Y_j}^{(2k+d_{max}-1)} \\
   {Y_j}^{(2k+d_{max}-2)} \\
  \vdots \\
   {Y_j}^{(k+d_{max})} \\
  \end{bmatrix}
&=\sum_{i=1}^{s}\begin{bmatrix}
  M_{ij}^{(0)}(\underline \varepsilon_2) & M_{ij}^{(1)}(\underline \varepsilon_2) & \cdots  & M_{ij}^{(d_{max})}(\underline \varepsilon_2) & 0 & 0 & \cdots & 0 & 0  \\
  0 & M_{ij}^{(0)}(\underline \varepsilon_2) & \cdots & M_{ij}^{(d_{max}-1)}(\underline \varepsilon_2) & M_{ij}^{(d_{max})}(\underline \varepsilon_2) & 0 & \cdots & 0& 0 \\
  \vdots & \vdots & \vdots & \vdots & \ddots & \ddots & \ddots & \vdots & \vdots \\
  0 & 0 & \cdots & 0 &  M_{ij}^{(0)}(\underline \varepsilon_2) & M_{ij}^{(1)}(\underline \varepsilon_2) & \cdots & M_{ij}^{(d_{max}-1)}(\underline \varepsilon_2) & M_{ij}^{(d_{max})}(\underline \varepsilon_2)\\
\end{bmatrix}\\
\nonumber
&~~~~~~~\times \begin{bmatrix}
                  {X_i}^{(2k-1)}
                  ~~{X_i}^{(2k-2)} 
                 ~~\cdots 
                  ~~{X_i}^{(k)} 
                  ~~{X_i}^{(2k-1)} 
                 ~~\cdots
                  ~~{X_i}^{(2k-d_{max})}
               \end{bmatrix}^T
\end{align}
\vspace{-0.2cm}
\hrule
\vspace{-0.6cm}
\end{figure*}
\begin{figure*}
\scriptsize
\begin{align}\label{bigeqn4}
\begin{bmatrix}
    {Y_j}^{(2k+d_{max}-1)} \\
   {Y_j}^{(2k+d_{max}-2)} \\
  \vdots \\
   {Y_j}^{(k+d_{max})} \\
\end{bmatrix}\hspace{-0.1cm}
=\hspace{-0.1cm}\sum_{i=1}^{s}\hspace{-0.1cm}\underbrace{\begin{bmatrix}
  M_{ij}^{(0)}(\underline \varepsilon_2) & M_{ij}^{(1)}(\underline \varepsilon_2) & \cdots & M_{ij}^{(d_{max}-1)}(\underline \varepsilon_2) & M_{ij}^{(d_{max})}(\underline \varepsilon_2) & 0 & \cdots & 0 & 0 & 0  \\
  0 & M_{ij}^{(0)}(\underline \varepsilon_2) & \cdots & M_{ij}^{(d_{max}-2)}(\underline \varepsilon_2) & M_{ij}^{(d_{max}-1)}(\underline \varepsilon_2) & M_{ij}^{(d_{max})}(\underline \varepsilon_2) & \cdots & 0 & 0& 0 \\
  \vdots & \vdots & \vdots & \vdots & \vdots & \vdots & \vdots & \vdots & \vdots & \vdots \\
  M_{ij}^{(1)}(\underline \varepsilon_2) & M_{ij}^{(2)}(\underline \varepsilon_2) & \cdots & M_{ij}^{(d_{max})}(\underline \varepsilon_2)& 0 & 0 & \cdots & 0 & 0 &  M_{ij}^{(0)}(\underline \varepsilon_2) \\
\end{bmatrix}}_{M_{ij}(\underline \varepsilon_2)}
\hspace{-0.2cm}\begin{bmatrix}
                  {X_i}^{(2k-1)} \\
                  {X_i}^{(2k-2)} \\
                 \vdots \\
                  {X_i}^{(k)} \\
         \end{bmatrix}
\end{align}
\vspace{-0.25cm}
\hrule
\vspace{-0.55cm}
\end{figure*}
Using Theorem \ref{diagthm}, ${M_{ij}(\underline \varepsilon_2)}$ can be diagonalized to $\hat{M}_{ij}(\underline \varepsilon_2)$, where $k$ is chosen so that $k|2^m-1$.  Similarly, the $l^{\text{th}}$-block of output symbols, after discarding the cyclic prefix, can be written in terms of the matrix $\hat{M}_{ij}(\underline \varepsilon_l)$, for $1 \leq l \leq (2n'+1)$. We note that 
{
\begin{align}
 \label{eqn-struc-trans-mat}
\hat{M}_{ij}(\underline \varepsilon_l) = &diag\left( M_{ij}(\underline \varepsilon_l,1), ~M_{ij}(\underline \varepsilon_l,\alpha),\cdots, ~M_{ij}(\underline \varepsilon_l,\alpha^{k-1}) \right).
\end{align}}where, $M_{ij}(\underline \varepsilon_l,\alpha^q)$ denotes the transfer function $M_{ij}(D)$ evaluated at $D=\alpha^q$ and $\varepsilon=\varepsilon_l$, for $q=0,1,\cdots,(k-1)$ . Let $X'^{(n'+1)k}_1$, $X'^{n'k}_2$, and $X'^{n'k}_3$  denote the $(n'+1)k$-length, $n'k$-length, and $n'k$-length independent symbols generated by $S_1$, $S_2$, and $S_3$ respectively. Partition each of the independent input symbols into $k$ blocks of equal length. Denote the $q^{\text{th}}$-block of independent input symbols of $S_i$ by $X'_i{(q)}$, for $0\leq q \leq k-1$, which is a column vector of lengths $(n'+1)$ for $S_1$,  $n'$ for $S_2$, and  $n'$ for $S_3$. The symbols $X'_i{(q)}$ are precoded onto $X^{k(2n'+1)}_i$ as follows. Define {\small $X^{(q \oplus k)}_i = \left[ X^{(q)}_i ~X^{(q+k)}_i ~X^{(q+2k)}_i ~\cdots ~X^{(q+2nk)}_i\right]^T$}, for $0 \leq q \leq k-1$. Let $V^{(q)}_i$ denote the precoding matrices at $S_i$, for $0 \leq q \leq k-1$. These matrices, for all $q$, are of size $(2n'+1)\times(n'+1)$, $(2n'
+1)\times n'$ and $(2n'+1)\times n'$ for $i=1,2,$ and $3$ respectively. Now, the symbols to be transmitted by $S_i$, before the pre-multiplication of each block by $F$ (where, the matrix $F$ is 
the DFT matrix defined in (\ref{eqn-F-matrix})) and the addition of CP to every block, are given by $X^{(q \oplus k)}_i=V^{(q)}_iX'_i{(q)}$. In brief, the $q^{\text{th}}$ element of every block to be transmitted by $S_i$, before the pre-multiplication of each block by $F$ and the addition of CP to every block, are obtained by precoding the $q^{\text{th}}$ block of independent symbols $X'_i{(q)}$. The instance of $q=0$ is shown in Fig. \ref{fig:partition-CP-zeros}.

After discarding the CP and pre-multiplying by $F^{-1}$ at $T_j$, we obtain $(2n'+1)k$-output symbols. These are partitioned into $k$-blocks, each of length $(2n'+1)$-symbols. Each block is given by {\small$Y^{(q \oplus k)}_i = \left[ Y^{(q)}_i ~Y^{(q+k)}_i ~Y^{(q+2k)}_i ~\cdots ~Y^{(q+2n'k)}_i\right]^T$}, for  $0 \leq q \leq k-1$. The input-output relation is now given by

{\small \vspace{-0.3cm}
\begin{align}
\label{ip-op-main}
Y^{(q \oplus k)}_i& = \sum_{i=1}^{3} diag\left( M_{ij}(\underline \varepsilon_1,\alpha^q), ~M_{ij}(\underline \varepsilon_2,\alpha^q), \cdots, ~M_{ij}(\underline \varepsilon_{2n'},\alpha^q), ~M_{ij}(\underline \varepsilon_{(2n'+1)},\alpha^q) \right)V^{(q)}_iX'_i{(q)}.
\end{align}}For  $0 \leq q \leq k-1$, define the matrix
{
\begin{align*}
 M_{{ij}}^{q}=diag\left( M_{ij}(\underline \varepsilon_1,\alpha^q), ~\cdots, ~M_{ij}(\underline \varepsilon_{(2n'+1)},\alpha^q) \right).
\end{align*}}

\subsection{Feasibility of PBNA using Transform Approach and Block Time Varying LECs }
We assume that the min-cut between $S_i-T_j$ is not zero for all $i \neq j$. The proof technique for feasibility of PBNA in the case of min-cut between $S_i-T_j$ being zero for some $i \neq j$ will be similar to that used for non-zero min-cut.  

PBNA using transform approach and block time-varying LECs requires that the following conditions be satisfied for $0 \leq q \leq k-1$.
{
\begin{align}
\nonumber
&\text{Span}({M}_{31}^{q}V^{(q)}_3) \subset \text{Span}({M}_{21}^{q}V^{(q)}_2),\text{Span}({M}_{32}^{q}V^{(q)}_3) \subset \text{Span}({M}_{12}^{q}V^{(q)}_1),\\
\nonumber
&\text{Span}({M}_{23}^{q}V^{(q)}_2) \subset \text{Span}({M}_{13}^{q}V^{(q)}_1),\\
\label{eqn-span-cond}
&\mbox{Rank}[{M}_{11}^{q}V^{(q)}_1 ~~{M}_{21}^{q}V^{(q)}_2]=\mbox{Rank}[V^{(q)}_1 ~~{{M}_{11}^{q}}^{-1}{M^{q}}_{21}V^{(q)}_2]=2n'+1 \\
\nonumber
&\mbox{Rank}[{M}_{22}^{q}V^{(q)}_2 ~~{M}_{12}^{q}V^{(q)}_1]=\mbox{Rank}[{{M}_{12}^{q}}^{-1}{M}_{22}^{q}V^{(q)}_2 ~~V^{(q)}_1]=2n'+1\\
\nonumber
&\mbox{Rank}[{M}_{33}^{q}V^{(q)}_3 ~~{M}_{13}^{q}V^{(q)}_1]=\mbox{Rank}[{{M}_{13}^{q}}^{-1}{M}_{33}^{q}V^{(q)}_3 ~~V^{(q)}_1]=2n'+1.
\end{align}}From Lemma $1$ in \cite{Yeu-book}, $m$ can always be chosen large enough so that the above rank conditions are satisfied, if the corresponding determinants are non-zero polynomials. 

We first note that recovering  $X'_i{(0)}$, for all $i$, represents the feasibility problem of PBNA in the instantaneous version of the original $3$-S $3$-D MUN-D. Suppose that we cannot recover $X'_i{(0)}$, for all $i$. But, if we can recover $X'_i{(q)}$, for all $q \neq 0$ and for all $i$, we can still achieve throughputs of $\frac{(n'+1)(k-1)}{(2n'+1)k}$, $\frac{n'(k-1)}{(2n'+1)k}$, $\frac{n'(k-1)}{(2n'+1)k}$ for  $S_1-T_1$, $S_2-T_2$ and $S_3-T_3$ respectively. This means that as $n$ and $k$ become arbitrarily large, a throughput close to $\frac{1}{2}$ can be achieved for every source-destination pair. However, in this section we show that if $X'_i{(0)}$, for some $i=i_1$, cannot be recovered then, $X'_{i_1}{(q)}$ is not recoverable for any $q$. Conversely, we also show that if $X'_i{(0)}$, for all $i$, can be recovered then $X'_{i}{(q)}$ is recoverable for all $q$ and $i$. 

\begin{definition}
PBNA in $3$-S $3$-D MUN-D using Transform Approach and Block Time Varying LECs is said to be feasible if $X'_i{(q)}$ can be recovered from $Y^{(q \oplus k)}_i$, for $i=1,2,3$, {$q = 1,2,\cdots, k-1$}, and  for every $n'>1$.
\end{definition} 

Henceforth in this section, PBNA in $3$-S $3$-D MUN-D using transform approach and block time-varying LECs shall be simply referred to as PBNA in $3$-S $3$-D MUN-D. We now proceed to prove that the reduced feasibility conditions of Meng et al. for feasibility of PBNA in $3$-S $3$-D I-MUN are also necessary and sufficient for PBNA in $3$-S $3$-D MUN-D.

PBNA in $3$-S $3$-D MUN-D is feasible iff there exists a choice of $(n'+1) \times n'$ matrices $A^{(q)}$ and $B^{(q)}$, $V^{(q)}_1$,  and a $n' \times n'$ matrix $C^{(q)}$, for $0 \leq q \leq k-1$, all with entries from $\mathbb F_{2^m}$, such that
%
{
\begin{align}
\nonumber
&\mbox{det}[V^{(q)}_1 ~~~~{M^q_{11}}^{-1}{M^q_{21}}{M^q_{23}}^{-1}M^q_{13}V^{(q)}_1A^{(q)}] \neq 0,\\
\nonumber
&\mbox{det}[{M^q_{12}}^{-1}M_{22}{M^q_{23}}^{-1}M^q_{13}V^{(q)}_1A^{(q)} ~~~~V^{(q)}_1] \neq 0,\\
\label{eqn-NA-nec-suff}
&\mbox{det}[{M^q_{13}}^{-1}M^q_{33}{M^q_{32}}^{-1}M^q_{12}V^{(q)}_1B^{(q)} ~~~~V^{(q)}_1] \neq 0,\\
\nonumber
&U^{(q)}V^{(q)}_1AC=V^{(q)}_1B
\end{align}}where, {\small $U^{(q)}={M^q_{12}}^{-1}M^q_{32}{M^q_{31}}^{-1}M^q_{21}{M^q_{23}}^{-1}M^q_{13}$}. The above conditions are obtained from the network alignment conditions in (\ref{eqn-span-cond}) and following the same steps as in the proof of \mbox{Theorem \ref{thm4}}.
For $0 \leq q \leq k-1$ , define 
%
{
\begin{align}
\nonumber
&\eta(q)=\frac{M_{21}(\underline\varepsilon,\alpha^q)  M_{32}(\underline\varepsilon,\alpha^q)M_{13}(\underline\varepsilon,\alpha^q)}{{M_{31}(\underline\varepsilon,\alpha^q)}{M_{23}(\underline\varepsilon,\alpha^q)}{M_{12}(\underline\varepsilon,\alpha^q)}},\\
\nonumber
&b_1(q)=\frac{M_{21}(\underline\varepsilon,\alpha^q) M_{13}(\underline\varepsilon,\alpha^q)}{M_{11}(\underline\varepsilon,\alpha^q) M_{23}(\underline\varepsilon,\alpha^q)}, \\
\label{eqn_defn_b}
&b_2(q)=\frac{M_{22}(\underline\varepsilon,\alpha^q) M_{13}(\underline\varepsilon,\alpha^q)}{M_{12}(\underline\varepsilon,\alpha^q) M_{23}(\underline\varepsilon,\alpha^q)},\\
\nonumber
&b_3(q)=\frac{M_{33}(\underline\varepsilon,\alpha^q) M_{12}(\underline\varepsilon,\alpha^q)}{M_{13}(\underline\varepsilon,\alpha^q) M_{32}(\underline\varepsilon,\alpha^q)}.
\end{align}}As in \cite{MRMJ}, we shall consider the two cases of $\eta(0)$ not being a constant\footnote{The terminology of $\eta(q)$ or $b_i(q)$ being a constant or not is understood to be with respect to $\underline \varepsilon$ and henceforth, this shall not be explicitly mentioned.} and being a constant, separately.

\textbf{Case 1:} $\eta(0)$ is not a constant.

The precoding matrices which are similar to those in \cite{DVJM} \cite{MRMJ} are given by
%
{
\begin{align}
 \nonumber
&V^{(q)}_1 = [W ~~U^{(q)}W ~~{U^{(q)}}^2W ~~\cdots ~~{U^{(q)}}^{n'}W],\\
\nonumber
&V^{(q)}_2 = [R^{(q)}W ~~R^{(q)}U^{(q)}W ~~R^{(q)}{U^{(q)}}^2W ~~\cdots ~~R^{(q)}{U^{(q)}}^{n'-1}W],\\
\label{eqn-precod-mats}
&V^{(q)}_3 = [S^{(q)}{U^{(q)}}W ~~S^{(q)}{U^{(q)}}^2W ~~\cdots ~~S^{(q)}{U^{(q)}}^{n'}W]
\end{align}}where, {\small $R=M^{q}_{13}{M^{q}_{23}}^{-1}$, $S=M^{q}_{12}{M^{q}_{32}}^{-1}$}, for $0 \leq q \leq k-1$, and {\small $W=[1 ~1 ~\cdots ~1]^T$} (all ones vector of size $(2n'+1) \times 1$). The above choice of precoding matrices satisfy the last condition in (\ref{eqn-NA-nec-suff}) though not necessarily the other conditions in (\ref{eqn-NA-nec-suff}). 


The following theorem of Meng et al. gives the reduced feasibility conditions for $3$-S $3$-D I-MUN.

\begin{theorem}[\cite{MRMJ} (Reduced Feasibility Conditions)] \label{thm_MRMJ}
 $X'_i{(0)}$ can be recovered from $Y^{(0 \oplus k)}_i$, for all $i$, iff
%
{
\begin{align}
\label{eqn-Meng}
b_i(0) \notin {\cal S}^{(0)} = \left\{ 1, \eta(0), \eta(0)+1, \frac{\eta(0)}{\eta(0)+1} \right\}.
\end{align}}
\end{theorem}

The following theorem shows that PBNA in $3$-S $3$-D MUN-D is feasible iff $b_i(0) \notin {\cal S}^{(0)}$.
 
\begin{theorem} \label{thm_GC}
When $\eta(0)$ is not a constant, $X'_i{(q)}$ can be recovered from $Y^{(q \oplus k)}_i$, for $q=1,2,\cdots,k-1$, iff  $X'_i{(0)}$ can be recovered from $Y^{(0 \oplus k)}_i$.
\end{theorem}
\begin{proof}
Proof is given in Appendix \ref{appen_thm_GC}.
\end{proof}

In brief, the above theorem proves that the reduced feasibility conditions of Meng et al. for feasibility of PBNA in $3$-S $3$-D I-MUN are also necessary and sufficient for feasibility of PBNA in $3$-S $3$-D MUN-D when $\eta(0)$ is not a constant.

\textbf{Case 2:} $\eta(0)$ is a constant. 

When $\eta(0)$ is a constant, Theorem $1$ of \cite{MRMJ} states that $X'_i{(0)}$ can be recovered from $Y^{(0 \oplus k)}_i$ iff $b_i(0)$ is not a constant, for $i =1,2,3$. Similar to Theorem $1$ of \cite{MRMJ} we have the following lemma.

\begin{lemma} \label{lem2}
PBNA in $3$-S $3$-D MUN-D is feasible iff $b_i(q)$ is not a constant, for $i =1,2,3$, and $1 \leq q \leq k-1$.
\end{lemma}
\begin{proof}
Proof is the same as for $q=0$ case in \cite{MRMJ}.
\end{proof}
The following proposition in combination with Theorem $1$ of \cite{MRMJ} and Lemma \ref{lem2} shows that PBNA in a $3$-S $3$-D MUN-D is feasible iff PBNA in the $3$-S $3$-D I-MUN is feasible.

\begin{proposition}
 $b_i(q)$, for  $1 \leq q \leq k-1$, is a constant iff $b_i(0)$ is a constant.
\end{proposition}
\begin{proof}
The proof follows using similar arguments as in the ``If Part''  and ``Only If Part'' in the proof of Theorem \ref{thm_GC}.
\end{proof}

The feasibility conditions for PBNA in $3$-S $3$-D MUN-D for the case of zero min-cut between $S_i-T_j$ for some $(i,j)$ are also the same as that for $3$-S $3$-D I-MUN as given in \cite{MRMJ}. For example, when the min-cut between $S_2-T_1$ is zero as considered in Case $2$ of the previous section, re-define $b_1(q)=\frac{M_{31}(\underline\varepsilon,\alpha^q) M_{12}(\underline\varepsilon,\alpha^q)}{M_{11}(\underline\varepsilon,\alpha^q) M_{32}(\underline\varepsilon,\alpha^q)}$. In such a case, PBNA is feasible iff $b_i(0)$ is not a constant for $i = 1,2,3$.

\begin{remark}
The PBNA scheme proposed in this section is different from PBNA using transform approach and time-invariant LECs, and PBNA using time-varying LECs where, the  independent symbols were precoded onto a single block of data which is transmitted after the addition of CP. In PBNA using transform approach and block time-varying LECs, the independent symbols are precoded across multiple blocks of data which are demarcated by separate CPs. A drawback in this scheme is that the decoding delay is higher compared to PBNA using transform approach and time-invariant LECs, and PBNA using time-varying LECs for the same values of $n'$ and large values of $k$. In order to decode the first $p$ symbols, for $p \leq k$, the decoding delay for PBNA using transform approach and block time-varying LECs is equal to $k(2n'+1)$ whereas for both PBNA using transform approach and time-invariant LECs, and PBNA using time-varying LECs (with $n=2n'+1$) the decoding delay is equal $2n'+1$.
\end{remark}

\section{Comparison of Feasibility of the PBNA Schemes in Section \ref{subsec1}, Section \ref{subsec2}, and Section \ref{sec5}} \label{subsec4}
In this section, we tie-up the feasibility of the PBNA schemes in Section \ref{subsec1} and Section \ref{sec5}. We also provide one example each for the cases where 1) the feasibility test fails for all the PBNA schemes proposed, 2) PBNA using time-varying LECs is feasible for some $(n_1,n_2,n_3)$ while the other two proposed PBNA schemes fail.

Consider the feasibility problem of PBNA using transform approach and time-invariant LECs described in Section \ref{subsec1}. Consider the case of non-zero min-cut between every $S_i-T_j$ with $b_i(q)$ as defined in (\ref{eqn_defn_b}) and $\eta(0)$ not being a constant. Using the results of \cite{MRMJ}, we shall show that the conditions of Theorem \ref{thm3} are also necessary for feasibility of PBNA using transform approach and time-invariant LECs, i.e., the conditions are not restricted by the choice of the precoding matrices in (\ref{thm2_pf7})-(\ref{thm2_pf9}). We also show that the derived set of necessary and sufficient conditions for feasibility of PBNA using transform approach and block time-varying LECs, i.e., (\ref{eqn-Meng}) is also a necessary condition for feasibility of PBNA using transform approach and time-invariant LECs.

\begin{proposition}
If the conditions of Theorem \ref{thm3} are not satisfied then, PBNA using transform approach and time-invariant LECs is infeasible for any choice of precoding matrices (i.e., even when the matrices are not restricted to those in (\ref{thm2_pf7})-(\ref{thm2_pf9})).
\end{proposition}
\begin{proof}
Note that the diagonal elements of $\hat{M}^{-1}_{11}\hat{M}_{21}\hat{R}$, $\hat{M}^{-1}_{12}\hat{M}_{22}\hat{R}$, and $\hat{M}^{-1}_{13}\hat{M}_{33}\hat{S}$, where $\hat{R}$ and $\hat{S}$ are defined in (\ref{mat_T}), are given by $b_1(q)$, $b_2(q)$, and $b_3(q)$ respectively, for $0 \leq q \leq 2n'+1$. Therefore, for the conditions of  Theorem \ref{thm3} to be satisfied, the columns of the matrices $[V_1 ~~\hat{M}^{-1}_{11}\hat{M}_{21}V_2]$, $[\hat{M}^{-1}_{12}\hat{M}_{22}V_2 ~~V_1]$, and $[\hat{M}^{-1}_{13}\hat{M}_{33}V_3 ~~V_1]$ must be linearly independent over the field of rational polynomials in $\underline\varepsilon$ where the choice of the precoding matrices are given by (\ref{thm2_pf7})-(\ref{thm2_pf9}). Alternatively, the conditions of  Theorem \ref{thm3} will fail iff for some $i$ and for $0 \leq q \leq 2n'+1$,
{
\begin{align}
\label{eqn-bi-TIL-cond}
& b_i(q)  \in  \left\{\left. \frac{f(\eta(q))}{g(\eta(q))} \right|f(x), g(x) \in \mathbb F_{2^m}(\underline \varepsilon)[x],f(x)g(x) \neq 0,\right. \\
\nonumber
&\left. \hspace{5.0cm} gcd(f(x),g(x))=1, ~deg(f) \leq n', ~deg(g) \leq n'-1 \right\}.
\end{align}} Note that the functions $f(x)$ and $g(x)$ must be the same for \mbox{$0 \leq q \leq 2n+1$}. If $b_i(q)  \in  \left\{\frac{f(\eta(q))}{g(\eta(q))}\right\}$, for $0 \leq q \leq 2n'+1$, such that the denominators of $f(x)$ and $g(x)$ are not constants, then the denominators can be subsumed in the numerators of $g(x)$ and $f(x)$ respectively. Hence, (\ref{eqn-bi-TIL-cond}) can be re-stated as
{
\begin{align} \label{eqn-bi-TIL-cond1}
& b_i(q)  \in  \left\{\left. \frac{f(\eta(q))}{g(\eta(q))} \right|f(x), g(x) \in \mathbb F_{2^m}[\underline \varepsilon][x],f(x)g(x) \neq 0,\right.\\
\nonumber
&\left. \hspace{5.0cm}  gcd(f(x),g(x))=1, ~deg(f) \leq n', ~deg(g) \leq n'-1 \right\}.
\end{align}} Using (\ref{eqn-bi-TIL-cond1}) and following exactly the same steps as in the proofs of Lemma $5$, Lemma $8$ and Theorem $2$ of \cite{MRMJ}, it can be shown that when (\ref{eqn-bi-TIL-cond1}) is satisfied, choice of any other precoding matrices would still not satisfy (\ref{thm2_pf4})-(\ref{thm2_pf6}) (given in Appendix \ref{appen_thm3}) which are necessary conditions for feasibility of PBNA using transform approach and time-invariant LECs.
\end{proof}

The following proposition states that the necessary and sufficient condition for feasibility of PBNA using transform approach and block time-varying LECs in (\ref{eqn-Meng}) is also a necessary condition for feasibility of PBNA using transform approach and time-invariant LECs.
\begin{proposition} \label{prop-TIL_inf}
 PBNA using transform approach and time-invariant LECs in $3$-S $3$-D MUN-D is infeasible if $b_i{(0)} \in {\cal S}^{(0)}$ for every $i = 1,2,3$ where, ${\cal S}^{(0)}$ is defined in (\ref{eqn-Meng}).
\end{proposition}
\begin{proof}
The proposition is just a re-statement of the fact that if $b_i{(0)} \in {\cal S}^{(0)}$, for instance  $b_i{(0)}=\frac{\eta(0)}{\eta(0)+1}$ then, $b_i{(q)}=\frac{\eta(q)}{\eta(q)+1}$ for $0 \leq q \leq 2n'+1$ which is proved using similar arguments as in the ``If'' part in the proof of Theorem \ref{thm_GC} (given in Appendix \ref{appen_thm_GC}). Thus, if $b_i{(0)} \in {\cal S}^{(0)}$ then (\ref{eqn-bi-TIL-cond1}) will be satisfied which implies that PBNA using transform approach and time-invariant LECs is infeasible.
\end{proof}

Hence, whenever PBNA using transform approach and block time-varying LECs is infeasible, PBNA using transform approach and time-invariant LECs is also infeasible \footnote{This can also be proved for the case of $\eta(0)$ being a constant and for the case of zero min-cut between $S_i-T_j$ for some $(i,j)$. Both follow directly from the fact that $b_i(0)$ is a constant iff $b_i(q)$ is a constant for all $q \in \{0,1,\cdots,2n'+1\}$.}. Conversely, if PBNA using transform approach and time-invariant LECs is feasible then, PBNA using transform approach and block time-varying LECs is feasible because the reduced feasibility conditions of (\ref{eqn-Meng}) will be satisfied. For example, PBNA using transform approach and block time-varying LECs is also feasible for the network considered in Example \ref{eg-PBNA_TINV_LEC}. The sufficiency of the condition in (\ref{eqn-Meng}) for feasibility of PBNA using transform approach and time-invariant LECs remains open. 

So, the next natural question is whether PBNA using time-varying LECs is feasible for some $(n_1,n_2,n_3)$ (where, $n_1,n_2,n_3 \neq 0$) or not when the other two PBNA schemes fail. This question is difficult to answer in generality. However, we show through examples the existence of $3$-S $3$-D MUN-D such that all the three PBNA schemes are infeasible and also the existence of $3$-S $3$-D MUN-D such that PBNA using time-varying LECs is feasible for some $(n_1,n_2,n_3,n)$ while the other two PBNA schemes are infeasible. The following example taken from \cite{ADVJM}, but with delays incorporated, is an instance where all the PBNA schemes described in the previous sections are infeasible.

\begin{example} \label{eg_TVL_NF}
Consider the network shown in Fig. \ref{NA_eg2}. Each link is taken to have unit-delay. 
\begin{figure}[htbp]
\centering
\includegraphics[width=3.4in,height=4.4in]{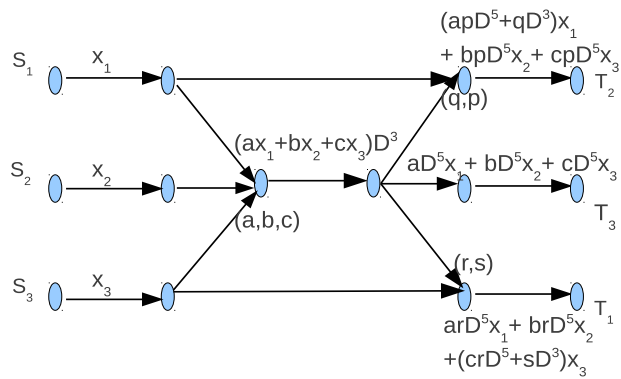}
\vspace{-5.5cm}
\caption{A $3$-S $3$-D MUN-D where, $(1)$ PBNA using transform approach and time-invariant LECs, and PBNA using transform approach and block time-varying LECs are infeasible, and $(2)$ PBNA using time-varying LECs is infeasible for all positive integer-tuples $(n_1,n_2,n_3,n)$.}	
\label{NA_eg2}	
\end{figure}The local encoding coefficients at each node are indicated in the figure. Here, $d_{max}=2$.

Note that $b_1(0)= 1$ and hence, PBNA using transform approach and time-invariant LECs, and PBNA using transform approach and block time-varying LECs are infeasible for all $n'>1$. We shall now show that PBNA using time-varying LECs is infeasible for all $(n_1>0,n_2>0,n_3>0)$ and $n>0$. Let $a_i$ denote the LEC $a$ at time instant $i$ and similarly denote the other LECs. Note that $b_i$ inside a matrix will denote the LEC $b$ at time instant $i$. Now, we have

{\footnotesize \vspace{-0.3cm}
\begin{align*}
&M_{11}=
\begin{bmatrix}
  0 & 0 & a_{n-3}r_{n-1} & 0 & \cdots & 0\\
  0 & 0 & 0 & a_{n-4}r_{n-2} & \cdots & 0\\
  \vdots & \vdots & \vdots & \vdots & \vdots & \vdots \\
  0 & 0 & 0 & 0 & \cdots & a_0r_2\\
  a_{-1}r_{1} & 0 & 0 & 0 & \cdots & 0\\
  0 & a_{-2}r_{0} & 0 & 0 & \cdots & 0\\ 
\end{bmatrix}\\
& M^{-1}_{11}=
\begin{bmatrix}
  0 & 0 & 0 & \cdots & \frac{1}{a_{-1}r_{1}} & 0\\
  0 & 0 & 0 & \cdots & 0 &  \frac{1}{a_{-2}r_{0}}\\
  \frac{1}{a_{n-3}r_{n-1}} & 0 & 0 & \cdots & 0 & 0\\
  0 & \frac{1}{a_{n-4}r_{n-2}} & 0 & \cdots & 0 & 0 \\
  \vdots & \vdots & \vdots & \vdots & \vdots & \vdots \\
  0 & 0 & \cdots & \frac{1}{a_{0}r_{2}} & 0 & 0\\ 
\end{bmatrix}
\end{align*}}Similarly,

{\footnotesize \vspace{-0.3cm}
\begin{align*}
&M^{-1}_{23}=
\begin{bmatrix}
  0 & 0 & 0 & \cdots & \frac{1}{b_{-1}} & 0\\
  0 & 0 & 0 & \cdots & 0 &  \frac{1}{b_{-2}}\\
  \frac{1}{b_{n-3}} & 0 & 0 & \cdots & 0 & 0\\
  0 & \frac{1}{b_{n-4}} & 0 & \cdots & 0 & 0 \\
  \vdots & \vdots & \vdots & \vdots & \vdots & \vdots \\
  0 & 0 & \cdots & \frac{1}{b_{0}} & 0 & 0\\  
\end{bmatrix}.
\end{align*}}The other transfer matrices involved in determining the feasibility of PBNA using time-varying LECs are given by 

{\footnotesize \vspace{-0.3cm}
\begin{align*}
&M_{21}=
\begin{bmatrix}
 0 & 0 & b_{n-3}r_{n-1} & 0 & \cdots & 0\\
  0 & 0 & 0 & b_{n-4}r_{n-2} & \cdots & 0\\
  \vdots & \vdots & \vdots & \vdots & \vdots & \vdots \\
  0 & 0 & 0 & 0 & \cdots & b_0r_2\\
  b_{-1}r_{1} & 0 & 0 & 0 & \cdots & 0\\
  0 & b_{-2}r_{0} & 0 & 0 & \cdots & 0\\ 
\end{bmatrix}\\
&M_{13}=
\begin{bmatrix}
 0 & 0 & a_{n-3} & 0 & \cdots & 0\\
  0 & 0 & 0 & a_{n-4} & \cdots & 0\\
  \vdots & \vdots & \vdots & \vdots & \vdots & \vdots \\
  0 & 0 & 0 & 0 & \cdots & a_0\\
  a_{-1} & 0 & 0 & 0 & \cdots & 0\\
  0 & a_{-2} & 0 & 0 & \cdots & 0\\ 
\end{bmatrix}.
\end{align*}}Hence, \mbox{$M^{-1}_{11}M_{21}M^{-1}_{23}M_{13}=I_n$}. Thus, the matrix
\begin{align*}
 [V_1 ~~M^{-1}_{11}M_{21}V_2]=[V_1 ~~M^{-1}_{11}M_{21}M^{-1}_{23}M_{13}V_1A]=[V_1 ~~V_1A]
\end{align*}is not full-rank. This violates (\ref{thm3_pf4}) (given in Appendix \ref{appen_thm4}) and hence, the condition of Theorem \ref{thm4} is not satisfied.
\end{example}

The following example considers a modified version of the network dealt in Example \ref{eg_TVL_NF} and is an instance where PBNA using time-varying LECs is feasible for some $(n_1,n_2,n_3,n)$ while the other two the PBNA schemes are infeasible.
\begin{example} \label{eg_TVL_F}
This example considers a network whose topology is the same as that in Fig. \ref{NA_eg2} where all the links except the incoming link at $T_3$ have unit delays. The incoming link at $T_3$ is assumed to have a delay of $4$ time units. Hence, $d_{max}$ is equal to $5$.

Note that $b_1(0)= 1$ and hence, PBNA using transform approach and time-invariant LECs, and PBNA using transform approach and block time-varying LECs are infeasible for all $n'>1$. We shall now show that PBNA using time-varying LECs is feasible for some $(n_1>0,n_2>0,n_3>0)$ and $n>0$. The notation used for the LECs is the same as that in Example \ref{eg_TVL_NF}. Now, the network transfer matrices are given by 

{\footnotesize \vspace{-0.3cm}
\begin{align*}
&M_{11}=
\begin{bmatrix}
  0 & 0 & a_{n-3}r_{n-1} & 0 & \cdots & 0\\
  0 & 0 & 0 & a_{n-4}r_{n-2} & \cdots & 0\\
  \vdots & \vdots & \vdots & \vdots & \vdots & \vdots \\
  0 & 0 & 0 & 0 & \cdots & a_0r_2\\
  a_{-1}r_{1} & 0 & 0 & 0 & \cdots & 0\\
  0 & a_{-2}r_{0} & 0 & 0 & \cdots & 0\\ 
\end{bmatrix},
M_{21}=
\begin{bmatrix}
  0 & 0 & b_{n-3}r_{n-1} & 0 & \cdots & 0\\
  0 & 0 & 0 & b_{n-4}r_{n-2} & \cdots & 0\\
  \vdots & \vdots & \vdots & \vdots & \vdots & \vdots \\
  0 & 0 & 0 & 0 & \cdots & b_0r_2\\
  b_{-1}r_{1} & 0 & 0 & 0 & \cdots & 0\\
  0 & b_{-2}r_{0} & 0 & 0 & \cdots & 0\\ 
\end{bmatrix},\\
&M_{31}=
\begin{bmatrix}
  s_7 & 0 & c_{n-3}r_{n-1} & 0 & \cdots & 0\\
  0 & s_6 & 0 & c_{n-4}r_{n-2} & \cdots & 0\\
  \vdots & \vdots & \vdots & \vdots & \vdots & \vdots \\
  0 & 0 & 0 & 0 & \cdots & c_0r_2\\
  c_{-1}r_{1} & 0 & 0 & 0 & \cdots & 0\\
  0 & c_{-2}r_{0} & 0 & 0 & \cdots & s_0\\ 
\end{bmatrix},
M_{12}=
\begin{bmatrix}
  q_7 & 0 & a_{n-3}p_{n-1} & 0 & \cdots & 0\\
  0 & q_6 & 0 & a_{n-4}p_{n-2} & \cdots & 0\\
  \vdots & \vdots & \vdots & \vdots & \vdots & \vdots \\
  0 & 0 & 0 & 0 & \cdots & a_0p_2\\
  a_{-1}p_{1} & 0 & 0 & 0 & \cdots & 0\\
  0 & a_{-2}p_{0} & 0 & 0 & \cdots & q_0\\ 
\end{bmatrix},\\
%
&M_{22}=
\begin{bmatrix}
  0 & 0 & b_{n-3}p_{n-1} & 0 & \cdots & 0\\
  0 & 0 & 0 & b_{n-4}p_{n-2} & \cdots & 0\\
  \vdots & \vdots & \vdots & \vdots & \vdots & \vdots \\
  0 & 0 & 0 & 0 & \cdots & b_0p_2\\
  b_{-1}p_{1} & 0 & 0 & 0 & \cdots & 0\\
  0 & b_{-2}p_{0} & 0 & 0 & \cdots & 0\\ 
\end{bmatrix},M_{32}=
\begin{bmatrix}
  0 & 0 & c_{n-3}p_{n-1} & 0 & \cdots & 0\\
  0 & 0 & 0 & c_{n-4}p_{n-2} & \cdots & 0\\
  \vdots & \vdots & \vdots & \vdots & \vdots & \vdots \\
  0 & 0 & 0 & 0 & \cdots & c_0p_2\\
  c_{-1}p_{1} & 0 & 0 & 0 & \cdots & 0\\
  0 & c_{-2}p_{0} & 0 & 0 & \cdots & 0\\ 
\end{bmatrix},\\
&M_{13}=
\begin{bmatrix} 
0 & 0 & \cdots & 0 & a_{n-6} & 0 & \cdots & 0\\
0 & 0 & \cdots & 0 & 0 & a_{n-5} & \cdots & 0\\
\vdots & \vdots & \vdots & \vdots & \vdots & \vdots & \vdots\\
0 & 0 & \cdots & 0 & 0 & 0 & \cdots & a_0\\
a_{-1} & 0 & \cdots & 0 & 0 & 0 & 0 & 0\\
0 & a_{-2} & \cdots & 0 & 0 & 0 & 0 & 0\\,
\vdots & \vdots & \vdots & \vdots & \vdots & \vdots & \vdots\\
0 & 0 & \cdots & a_{-5} & 0 & 0 & \cdots & 0
\end{bmatrix},M_{23}=
\begin{bmatrix} 
0 & 0 & \cdots & 0 & b_{n-6} & 0 & \cdots & 0\\
0 & 0 & \cdots & 0 & 0 & b_{n-5} & \cdots & 0\\
\vdots & \vdots & \vdots & \vdots & \vdots & \vdots & \vdots\\
0 & 0 & \cdots & 0 & 0 & 0 & \cdots & b_0\\
b_{-1} & 0 & \cdots & 0 & 0 & 0 & 0 & 0\\
0 & b_{-2} & \cdots & 0 & 0 & 0 & 0 & 0\\,
\vdots & \vdots & \vdots & \vdots & \vdots & \vdots & \vdots\\
0 & 0 & \cdots & b_{-5} & 0 & 0 & \cdots & 0
\end{bmatrix},\\
&M_{33}=
\begin{bmatrix} 
0 & 0 & \cdots & 0 & c_{n-6} & 0 & \cdots & 0\\
0 & 0 & \cdots & 0 & 0 & c_{n-5} & \cdots & 0\\
\vdots & \vdots & \vdots & \vdots & \vdots & \vdots & \vdots\\
0 & 0 & \cdots & 0 & 0 & 0 & \cdots & c_0\\
c_{-1} & 0 & \cdots & 0 & 0 & 0 & 0 & 0\\
0 & c_{-2} & \cdots & 0 & 0 & 0 & 0 & 0\\,
\vdots & \vdots & \vdots & \vdots & \vdots & \vdots & \vdots\\
0 & 0 & \cdots & c_{-5} & 0 & 0 & \cdots & 0
\end{bmatrix}.
\end{align*}}Let $n=8, n_1=5, n_2=3, n_3=3$. Unlike in Example \ref{eg_TVL_NF}, here the matrix $M^{-1}_{11}M_{21}M^{-1}_{23}M_{13}$ is not identity and is equal to
\begin{align*}
\begin{bmatrix} 
1 & 0 & 0 & 0 & 0 & 0 & 0 & 0\\
0 & 1 & 0 & 0 & 0 & 0 & 0 & 0\\
0 & 0 & \frac{a_{-3}b_5}{a_5b_{-3}} & 0 & 0 & 0 & 0 & 0\\
0 & 0 & 0 & \frac{a_{-4}b_4}{a_4 b_{-4}} &  0 & 0 & 0 & 0\\
0 & 0 & 0 & 0 & \frac{a_{-5}b_3}{a_3b_{-5}} & 0 & 0 & 0\\
0 & 0 & 0 & 0 & 0 & 1 & 0 & 0\\
0 & 0 & 0 & 0 & 0 & 0 & 1 & 0\\
0 & 0 & 0 & 0 & 0 & 0 & 0 & 1
\end{bmatrix}.
\end{align*}
\begin{figure*}
\begin{align} \nonumber
&a_{-5}=1, a_{-4}=\beta, a_{-3}=1, a_{-2}=\beta, a_{-1}=\beta^6, a_0=\beta^2, a_1=\beta^{12}, a_2=\beta^7,\\ \nonumber
&a_3 =\beta^{26}, a_4=\beta^3, a_5=\beta^{31},\\ \nonumber
&b_{-5}=\beta^{54}, b_{-4}=\beta^{52}, b_{-3}=\beta^{13}, b_{-2}=\beta^{35}, b_{-1}=\beta^8, b_0=\beta^{48}, b_1=\beta^{27},\\ \nonumber
&b_2=\beta^{18}, b_3=\beta^4, b_4=\beta^{24}, b_5=\beta^{28},\\ \nonumber
&c_{-5}=\beta^{45}, c_{-4}=\beta^{54}, c_{-3}=1, c_{-2}=\beta, c_{-1}=\beta^6, c_0=\beta^2, c_1=\beta^{12}, c_2=\beta^7,\\ \nonumber
&c_3=\beta^{26}, c_4=\beta^3, c_5=\beta^{32},\\ \label{eqn-LEC-TVF}
&p_0=\beta^{45}, p_1=\beta^{49}, p_2=\beta^{38}, p_3=\beta^{28}, p_4=\beta^{41}, p_5=\beta^{19}, p_6=\beta^{56}, p_7=\beta^5,\\ \nonumber
&q_0=\beta^{24}, q_1=\beta^{33}, q_2=\beta^{16}, q_3=\beta^{14}, q_4=\beta^{52}, q_5=\beta^{36}, q_6=\beta^{54}, q_7=\beta^9,\\ \nonumber
&r_0=\beta^{62}, r_1=\beta^{25}, r_2=\beta^{11}, r_3=\beta^{34}, r_4=\beta^{31}, r_5=\beta^{17}, r_6=\beta^{47}, r_7=\beta^{15},\\ \nonumber
&s_0=\beta^{32}, s_1=\beta^{13}, s_2=\beta^{35}, s_3=\beta^8, s_4=\beta^{48}, s_5=\beta^{27}, s_6=\beta^{18}, s_7=\beta^4.
\end{align}
\hrule
\end{figure*}We now, choose $V_1=[W ~UW ~U^2W ~U^3W ~U^4W]$ where, the matrix $U$ is defined in (\ref{eqn-na2_matU}) and $W$ is the all-ones column vector of length $n=8$. We choose the matrices $A$ and $B$ (defined in (\ref{na2_V2_and_V3})) respectively to be equal to the first three columns and the second three columns of the identity matrix $I_5$.  The matrix $C$ is taken to be equal to the identity matrix $I_3$. Now, it can be easily verified that $g_{ij}=0$, where the rational-polynomial $g_{ij}$ is defined in (\ref{eqn-gij}). Now, for PBNA with time-varying LECs to be feasible, it remains to be verified if there exists an assignment to the LECs such that (\ref{thm3_pf4})-(\ref{thm3_pf6}) (given in Appendix \ref{appen_thm4}) are satisfied. The LECs involved are chosen (randomly) as given in (\ref{eqn-LEC-TVF}) where, $\beta$ is a primitive element of $GF(2^6)$ whose minimal polynomial is given by $(1+x+x^6)$. With these choice of LECs, it can be verified using the software {\em Mathematica} that (\ref{thm3_pf4})-(\ref{thm3_pf6}) are satisfied.
\end{example}

In the next section, we shall motivate a discussion on the potential of a non-asymptotic scheme, namely on-off scheme, to achieve a rate of half for every source-destination pair in $3$-S $3$-D MUN-D.
\section{Discussion on On-off Schemes}
\label{sec6}
PBNA for $3$-S $3$-D I-MUN was primarily motivated by the breakthrough result for $K$-user Gaussian interference channel (GIC) in \cite{CaJ} where IA helped achieve a sum-degrees of freedom (DoF) of $\frac{K}{2}$ asymptotically. The presence of diagonal network transfer matrices in $3$-S $3$-D I-MUN and the ability to diagonalize the network transfer matrices in $3$-S $3$-D MUN-D helped in readily adapting the IA problem formulation and the IA precoders proposed for the $K$-user GIC to the $3$-S $3$-D I-MUN and $3$-S $3$-D MUN-D settings. 

We now discuss if some simple on-off schemes can achieve a rate of half for every source-destination pair in $3$-S $3$-D MUN-D. This discussion is motivated by an interesting result for the $K$-user GIC with propagation delays \cite{CaJ2} where, it was shown that by appropriately adjusting the duration of transmission, at every destination all the interference symbols would arrive at even time slots while the desired symbol would arrive at odd time slots. Hence, using a simple on-off signaling each user could achieve a DoF of half almost surely. 

If on-off schemes could achieve a rate of half for every source-destination pair then, PBNA would be unnecessary for $3$-S $3$-D MUN-D. But, as we shall see, there exist networks where the proposed PBNA schemes are feasible while on-off schemes cannot achieve a rate of half for every source-destination pair. Unlike in the $K$-user GIC with propagation delays, the advantage offered by the on-off schemes cannot be completely realized in $3$-S $3$-D MUN-D because of the fundamental difference between the wireline system model for delay networks and the wireless system model involving propagation delays. The model discussed in Section \ref{sec2} assumed that the link delays are positive integer multiples of the symbol duration. This gave rise to the input-output relations in (\ref{Trmx}) and (\ref{trmx_tvarleks}). Whereas in the wireless setting, the symbol duration is independent of the propagation delays between the sources and destinations.

Formally, we define an on-off scheme in a $3$-S $3$-D MUN-D as follows.
\begin{definition}
 A transmission scheme where every source switches itself on and off so that the interference symbols at each of the destinations can be aligned in orthogonal time slots with respect to the desired symbols is defined as an on-off scheme.
\end{definition}

We shall now consider two examples of $3$-S $3$-D MUN-D where, in the first example, an on-off scheme can achieve a rate of half for every source-destination pair, and in the second, it is impossible to achieve a rate of half for every source-destination pair using on-off schemes.

\begin{example} \label{eg_on_off_feasible}
Consider the $3$-S $3$-D MUN-D in Fig. \ref{On_Off_Feasible} where all the links have unit delay. It can be easily verified that the reduced feasibility condition of Theorem \ref{thm_MRMJ} are satisfied and hence, by Theorem \ref{thm_GC}, PBNA using transform approach and block time-varying LECs is feasible.
\begin{figure}[htbp]
\centering
\includegraphics[width=6.4in,height=4.4in]{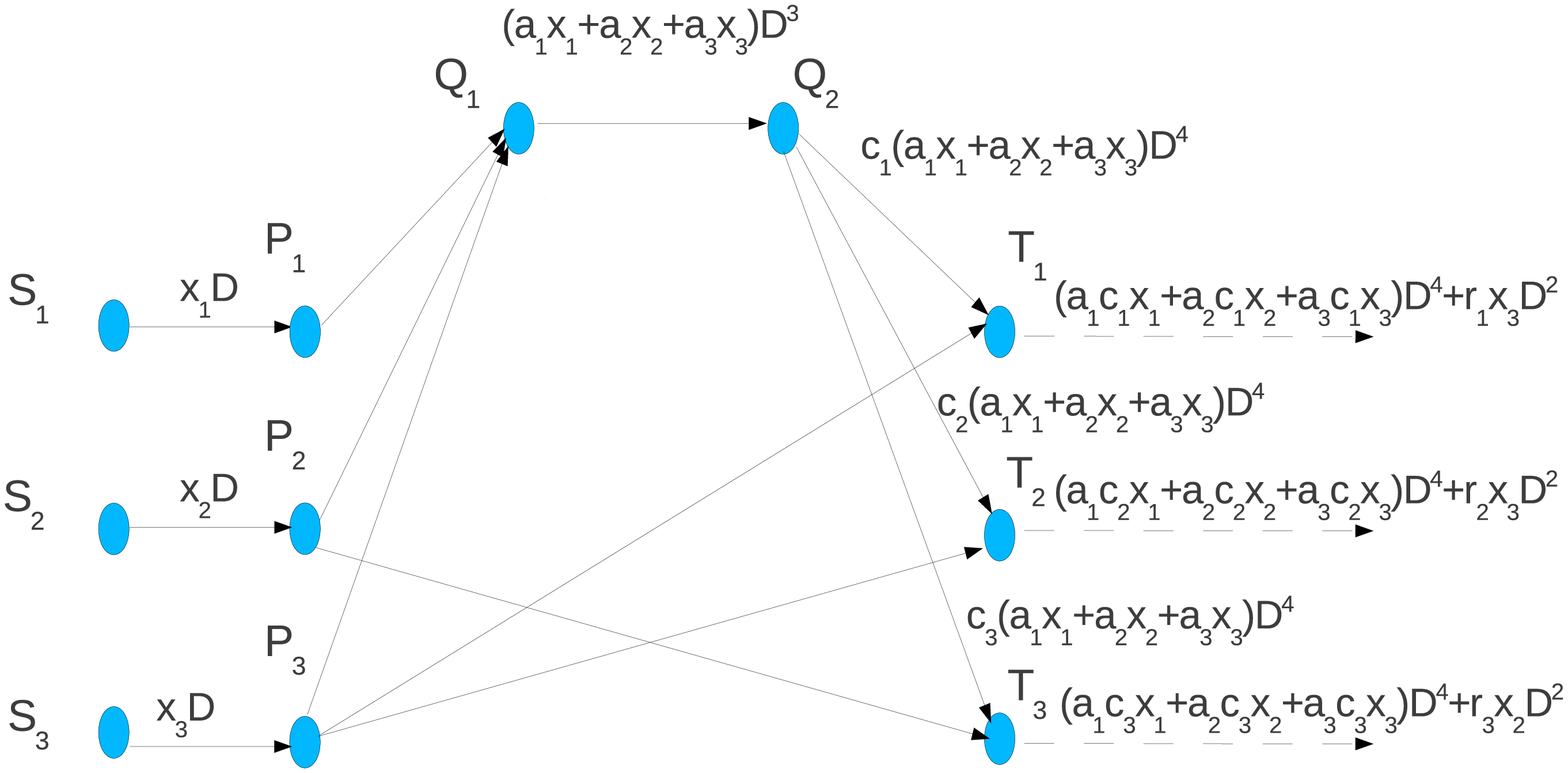}
\vspace{-1.5cm}
\caption{A $3$-S $3$-D MUN-D where, an on-off scheme can achieve a rate of half for every source-destination pair. PBNA using transform approach and block time-varying LEKs is feasible and achieves half-rate asymptotically for every source-destination pair.}
\label{On_Off_Feasible}	
\end{figure}
Now, consider the following on-off scheme. The destinations $T_1$ and $T_2$ delays the incoming symbol from the node $P_3$ by two time slots and then linear combines it with the incoming symbol from the node $Q_2$ so that the random process at the imaginary output links at $T_1$ and $T_2$ are given by $(a_1c_1x_1+a_2c_1x_2+a_3c_1x_3)D^4+r_1x_3D^4$ and $(a_1c_2x_1+a_2c_2x_2+a_3c_2x_3)D^4+r_2x_3D^4$ respectively. Choosing $r_1=-a_3c_1$ and $r_2=-a_3c_2$, the interference from $S_3$ at $T_1$ and $T_2$ are eliminated. Similarly, the interference from $S_2$ is eliminated at $T_3$. Let $a_i=c_i=1$, for all $i$ so that the output symbols at $T_1$, $T_2$, and $T_3$ at every time instant are given by $x_1+x_2$, $x_1+x_2$, and $x_1+x_3$ respectively. Now, the sources $S_1$, $S_2$, and $S_3$ are allowed to transmit only in odd, even, and even time slots respectively. Hence, the on-off scheme achieves a rate of half for every source-destination pair.
\end{example}

\begin{example}\label{eg_on_off_not_feasible}
Consider the $3$-S $3$-D MUN-D in Fig. \ref{On_Off_Not_Feasible} where all the links have unit delay. This network is essentially the same network as that in Example \ref{eg-PBNA_TINV_LEC} but with the last link before each of the destinations removed. Hence,  PBNA using transform approach and time-invariant LECs, and PBNA using transform approach and block time-varying LECs are feasible. 
\begin{figure}[htbp]
\centering
\includegraphics[width=6.4in,height=4.4in]{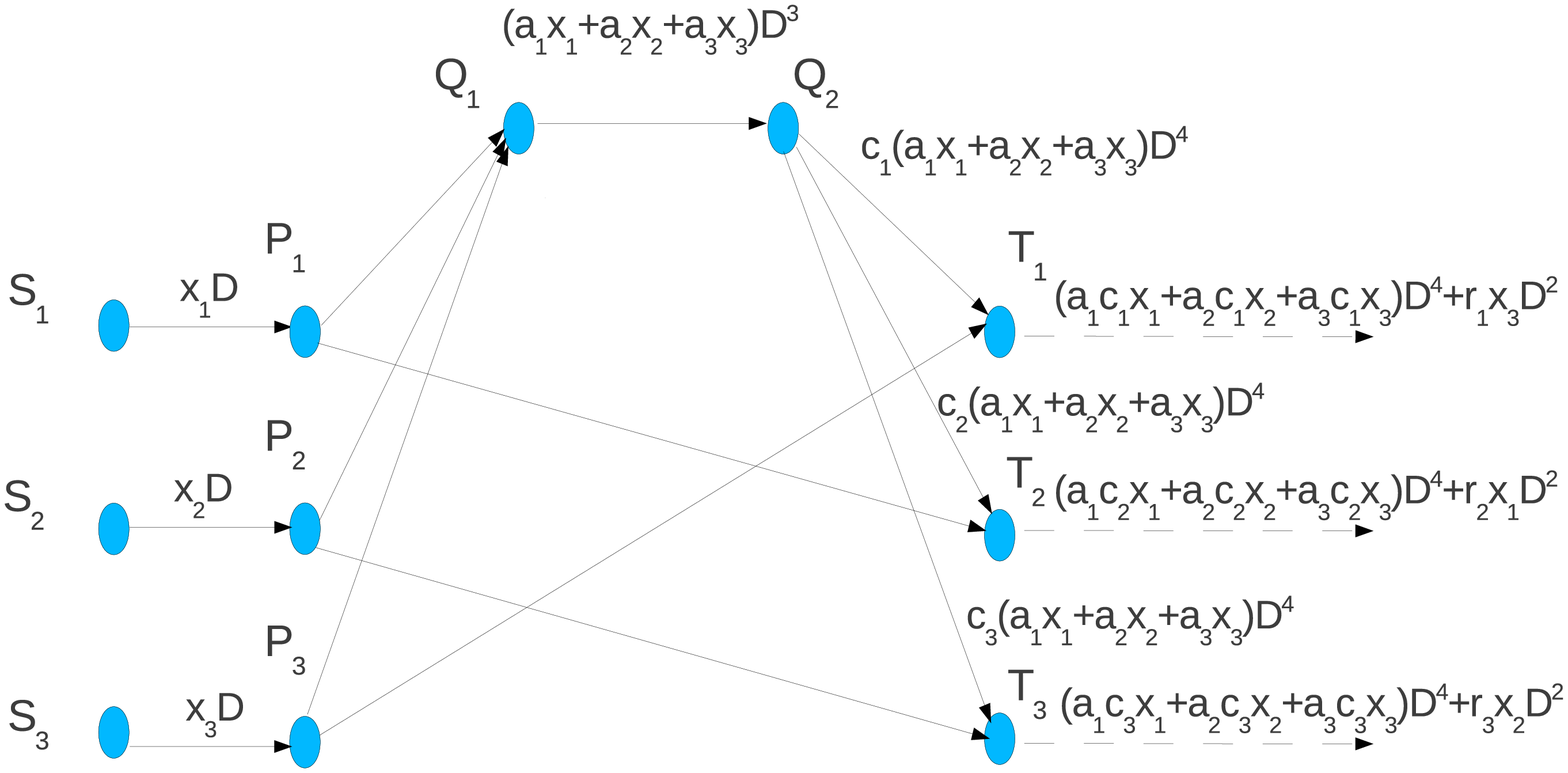}
\vspace{-1.5cm}
\caption{A $3$-S $3$-D MUN-D where, it is impossible to achieve a rate of half for every pair $S_i-T_i$ using on-off schemes while PBNA using transform approach and time-invariant LECs, and PBNA using transform approach and block time-varying LECs are feasible.}
\label{On_Off_Not_Feasible}	
\end{figure}
Now, consider the following on-off scheme. Like in Example \ref{eg_on_off_feasible}, the destinations $T_1$, $T_2$, and $T_3$ cancel the interference from $S_3$, $S_1$, and $S_2$ respectively by delaying the symbols received from the nodes $P_3$, $P_1$, and $P_2$. The output random process at $T_1$, $T_2$, and $T_3$ are now given by $(a_1c_1x_1+a_2c_1x_2)D^4$, $(a_2c_2x_2+a_3c_2x_3)D^4$, and $(a_3c_3x_3+a_1c_3x_1)D^4$ respectively. If $x_1$ and $x_2$ are to be received in orthogonal time slots at $T_1$ then, let $S_1$ transmit at odd time instants and $S_2$ transmit at even time instants. Similarly for $x_2$ and $x_3$ to be received in orthogonal time slots at $T_2$, $S_3$ has to transmit at odd time instants. Now, it is impossible for $x_3$ and $x_1$ to be received in orthogonal time slots at $T_3$. Similarly in the case of interference not being canceled at some destination nodes, an on-off scheme cannot achieve a rate of half for every source-destination pair.

To summarize, Example \ref{eg_on_off_not_feasible} showed that there exists a $3$-S $3$-D MUN-D where a PBNA scheme can achieve a rate of half (asymptotically) for every source-destination pair while on-off schemes cannot achieve a rate of half for every source-destination pair. This strengthens the case for PBNA in $3$-S $3$-D MUN-D. On the other hand,  Example \ref{eg_on_off_feasible} showed that there exists $3$-S $3$-D MUN-D where an on-off scheme can achieve a rate of half for every source-destination pair. Though one of the proposed PBNA schemes is feasible in Example \ref{eg_on_off_feasible}, it is unnecessary because it achieves a rate of half for every source-destination pair only asymptotically ($n'\rightarrow \infty$) unlike the on-off scheme. Nevertheless, identifying the class of $3$-S $3$-D MUN-D or wireline networks in general where on-off schemes can provide a rate guarantee of half for every source-destination pair remains open.

\end{example}

\section{Conclusion}
\label{sec7}
Using DFT, an acyclic network with delay was transformed into $n$-instantaneous networks {\em without making use of memory at the intermediate nodes}. This was then applied to $3$-S $3$-D MUN-D and it was shown that there can exist networks where PBNA is feasible even by using time-invariant LECs which is not possible in the delay-free counterpart. The conditions for feasibility of network alignment were then generalized with time-varying LECs and posed as an algebraic geometry problem in Section \ref{subsec2} under Case $1$. PBNA using transform approach and block time-varying LECs was proposed, and it was shown that its feasibility conditions is the same as the reduced feasibility conditions of $3$-S $3$-D I-MUN. It was also shown that if  PBNA using transform approach and block time-varying LECs is infeasible then, PBNA using transform approach and time-invariant LECs is also infeasible. The complete role of network topology in determining the feasibility of PBNA in $3$-S $3$-D I-MUN appears in \cite{MDRJMV} which is a more recent version of \cite{MRMJ}. Hence, the same is applicable to PBNA using transform approach and time-invariant LECs, and PBNA using transform approach and block time-varying LECs.

It was also shown that there exists a $3$-S $3$-D MUN-D where PBNA using time-varying LECs is feasible while the other two proposed PBNA schemes are infeasible. The following questions however remain open.
 
\begin{enumerate}
\item Under what condition is PBNA using time-varying LECs feasible when the reduced feasibility conditions are not satisfied?
\item The sufficiency of the reduced feasibility conditions for feasibility of PBNA using transform approach and time-invariant LECs remains open, i.e., is PBNA using transform approach and time-invariant LECs feasible whenever PBNA using transform approach and block time-varying LECs is feasible?  
\end{enumerate}

Optimizing Groebner basis algorithms for specific networks is crucial to verifying the condition of Theorem \ref{thm4} and hence, is of significant interest.

Optimality of PBNA, when feasible, for the class of $3$-S $3$-D MUN-D remains to be investigated. PBNA for $3$-S $3$-D \mbox{MUN-D} discussed in this paper as well as that for $3$-S $3$-D I-MUN, could be extended to the case where each source-destination pair has a min-cut greater than one. Another interesting direction of future research is extending PBNA to the case of arbitrary number of sources and destinations with arbitrary message demands.

Though the transform method described was claimed to be applicable for acyclic networks having $M(D)$ whose elements are only polynomial functions in $D$, it can also be applied to networks having $M(D)$ whose elements are rational functions in $D$ by multiplying by the LCM of all the denominators of the rational functions, at all the sinks. This gives a finite $d_{max}$. The same applies to cyclic networks too.

\section*{Acknowledgement}
This work was supported  partly by the DRDO-IISc program on Advanced Research in
Mathematical Engineering through a research grant as well as the INAE Chair
Professorship grant to B.~S.~Rajan.

\appendices
 \section{Proof of Theorem \ref{diagthm}}
\label{appen_thm1}
\begin{IEEEproof}
For $0 \leq j \leq n-1$, we have
\begin{align}
\nonumber A\begin{bmatrix}
   I_{\mu} \\
   \alpha^{j} I_{\mu} \\
   \alpha^{2j} I_{\mu} \\
   \vdots \\
   \alpha^{(n-1)j} I_{\mu} \\
 \end{bmatrix}_{n\mu\times\mu}
 =&\begin{bmatrix}
   \sum_{i=0}^{L}\alpha^{ij}A^{(i)} \\
   \sum_{i=0}^{L}\alpha^{(i+1)j}A^{(i)} \\
   \vdots\\
   \sum_{i=0}^{L}\alpha^{(i+n-1)j}A^{(i)} \\
 \end{bmatrix}_{n\nu\times\mu}\\
\label{diag3} =&\begin{bmatrix}
   I_{\nu} \\
   \alpha^{j}I_{\nu} \\
   \alpha^{2j}I_{\nu} \\
   \vdots \\
   \alpha^{(n-1)j}I_{\nu} \\
 \end{bmatrix}_{n\nu\times\nu}
 \left(\sum_{i=0}^{L}\alpha^{ij}A^{(i)}\right).
\end{align}
\noindent From (\ref{diag3}), we have
\begin{equation}
\label{eqn-diag-imp}
AQ_{\mu}=Q_{\nu}\hat{A}.
\end{equation}
The matrix $Q_{\mu}$ defined in (\ref{diagqmu}) can also be written as \mbox{$Q_{\mu}=F\otimes I_{\mu}$} ($i.e.$, the Kronecker product of $F$ and $I_{\mu}$). Similarly, $Q_{\nu}=F\otimes I_{\nu}$. The inverse of the matrix $F$ is given by \cite{Bla}
\begin{align*}
F^{-1}=&n^{-1}\begin{bmatrix}
    1 & 1 & 1 & \cdots & 1 \\
    1 & \alpha^{-1} & \alpha^{-2} & \cdots & \alpha^{-(n-1)} \\
    1 & \alpha^{-2} & \alpha^{-4} & \cdots & \alpha^{-2(n-1)} \\
    \vdots & \ddots & \ddots & \ddots & \ddots \\
    1 & \alpha^{-(n-1)} & \alpha^{-2(n-1)} & \cdots & \alpha^{-(n-1)(n-1)} \\
  \end{bmatrix}.
\end{align*}
\noindent 
Now, $det(Q_{\mu})=[det(F)]^{\mu}[det(I_{\mu})]^{n}\not=0$ and $Q_{\mu}^{-1}=F^{-1}\otimes I_{\mu}$ ($\because Q_{\mu}Q_{\mu}^{-1}=(F\otimes I_{\mu})(F^{-1}\otimes I_{\mu})=(FF^{-1})\otimes I_{\mu}=I_{n\mu}$). Hence, from (\ref{eqn-diag-imp}) 
\begin{equation*}
A=Q_{\nu}\hat{A}Q_{\mu}^{-1}.
\end{equation*}
\end{IEEEproof}

\section{Proof of Theorem \ref{thm2}}
\label{appen_thm2}
\begin{IEEEproof}
{\em \underline{If Part}: }
If both the conditions are satisfied after the assignment of values to $\underline{\varepsilon}$, then sink-$j$ can invert the matrix  $[\hat{M}_{i_1j}^{(k)}(l_{i_1}) ~\hat{M}_{i_2j}^{(k)}(l_{i_2}) ~\cdots ~\hat{M}_{i_{s'}j}^{(k)}(l_{i_{s'}}) ]$ and decode the required input symbols  without any interference.

{\em\underline{ Only If Part}: }If Condition 1) is not satisfied for some sink-$j$ then, sink-$j$ receives superposition of the required input symbols and interference from other input symbols, from which it cannot decode the required input symbols.

If Condition 2) is not satisfied for some sink-$j$ then, sink-$j$ cannot invert the matrix \\ {$[\hat{M}_{i_1j}^{(k)}(l_{i_1}) ~\hat{M}_{i_2j}^{(k)}(l_{i_2}) ~\cdots ~\hat{M}_{i_{s'}j}^{(k)}(l_{i_{s'}}) ]$} which is necessary for decoding the input symbols.
\end{IEEEproof}

\section{Proof of Lemma \ref{nonzeros}}
\label{appen_nonzeros}
\begin{IEEEproof}
Following the terminology developed so far, for some $n>>d_{max}$ and for $0\leq t \leq n-1,$ let
\[
\underline{X}^{(t)}=
\begin{bmatrix} \underline{X_1}^{(t)}\\
\underline{X_2}^{(t)} \\
\vdots\\
\underline{X_s}^{(t)}\\
\end{bmatrix}.
\]

Then, by (\ref{eqnno3}), (\ref{instant}) and the structure of the $\hat{M}_{ij}^{(t)}$ matrices, we have for $0\leq t \leq n-1,$
\begin{equation}
\label{eqnno4}
\underline{Y_j}^{(t)} =
\left(\sum_{d=0}^{d_{max}}{\alpha}^{d(n-1-t)}M_j^{(d)}\right)\underline{X}^{(t)}
,
\end{equation}
where $M_j^{(d)}$ is a $\nu_j \times \mu$ matrix over $\mathbb{F}_{p^m}$ (considered as a subfield of $\mathbb{F}_{p^b}$) such that  
\begin{equation}
\label{eqnno5}
M_j(D)=\sum_{d=0}^{d_{max}}M_j^{(d)}D^d.
\end{equation}

We define a collection of ring homomorphisms $\phi_t:\mathbb{F}_{p^m}(D) \rightarrow \mathbb{F}_{p^{b}}$ for $0\leq t \leq n-1,$ given by $\phi_t(D)={\alpha}^t$ and $\phi_t$ as an identity map on $\mathbb{F}_{p^m}$. For some matrix $P(D)$ over $\mathbb{F}_{p^m}(D),$ we also define $\phi_t(P(D))$ to be equal to the matrix $P$ with elements in $\mathbb{F}_{p^b}$ that are the $\phi_t$-images of the corresponding elements of $P(D).$ Then, from (\ref{eqnno4}) and (\ref{eqnno5}), we have 
\begin{equation}
\label{eqnno6}
\underline{Y_j}^{(n-1-t)} = \phi_t(M_j(D))\underline{X}^{(n-1-t)},
\end{equation}
for $0 \leq t \leq n-1.$ Clearly, the zero-interference conditions satisfied in the $M_j(D)$ matrices continue to hold in the $\phi_t(M_j(D))$ matrices, for $0 \leq t \leq n-1$ and for any sink-$j.$ Having satisfied the zero-interference conditions, to recover the source processes demanded by each sink-$j$ at time instant $n-1-t,$ the invertibility conditions also have to be satisfied, i.e., 
\begin{equation}
\label{eqnno7}
\prod_{j=1}^{r}det\left(\phi_t(M_j'(D))\right)\neq 0,
\end{equation}
where $M_j'(D)$ is the square submatrix of $M_j(D)$ indicating the source processes that are demanded by sink-$j.$ But then, we have
\begin{equation}
\label{eqnno10}
det\left(\phi_t(M_j'(D))\right) = \phi_t(det(M_j'(D)))
\end{equation}
and thus 
\begin{align*}
\prod_{j=1}^{r}det\left(\phi_t(M_j'(D))\right)&=\prod_{j=1}^{r}
\phi_t\left(det(M_j'(D))\right)\\
&=\phi_t\left(\prod_{j=1}^{r}det(M_j'(D))\right)\\
&=\phi_t(f(D))\\
&=f({\alpha}^t),\\
\end{align*}
where $f(D)$ is as defined in (\ref{eqnno1}). Clearly, $f({\alpha}^t)\neq 0$ implies that (\ref{eqnno7}) is satisfied and the source processes demanded at each sink can be recovered at time instant $n-1-t$ in the transform approach. Similarly, if the sink demands are satisfied at time instant $n-1-t$ in the transform approach, clearly we must have $f({\alpha}^t)\neq 0.$ This holds for $0 \leq t \leq n-1,$ thus proving the lemma.
\end{IEEEproof}

\section{Proof of Theorem \ref{thmexistence}}
\label{appen_thmexistence}
\begin{IEEEproof}
\textit{\underline{If part}:} Let $\mathbb{F}_{p^m}$ be the field over which the feasible network code has been obtained for $({\cal G},{\cal C}).$  Consider the polynomial $f(D)$ (given by (\ref{eqnno1})) with coefficients from $\mathbb{F}_{p^m}$. Let $\mathbb{F}_{p^{m'}}$ be the splitting field of this polynomial, i.e., a suitable smallest extension field of  $\mathbb{F}_{p^m}$ in which $f(D)$ splits into linear factors.

Let
\[
p^{m'}-1={\prod_{b=1}^{b=k}}p_{b}^{m_b'},
\]
where each $p_b$ is some prime and $m_b$ is some positive integer. 

By Lemma \ref{nonzeros}, the choice of $\alpha$ to be used for the DFT operations should be such that $f(\alpha^t) \neq 0,$ for any $0 \leq t \leq n-1.$ We now show that such an $\alpha$ exists and can be chosen. 

Let $\mathbb{F}_{p^{{m''}}}$ be an extension field of $\mathbb{F}_{p^{m'}}.$ Clearly, $\left(p^{m'}-1\right)|\left(p^{{m''}}-1\right).$ However, we further demand that $\mathbb{F}_{p^{{m''}}}$  is such that
\begin{equation}
\label{eqnno2}
p^{{m''}}-1={\prod_{b=1}^{b=k}}p_b^{m_b''}{\prod_{c=1}^{c=k'}}p_c^{m_c''},
\end{equation}
where each $p_c$ is some prime  and $m_b''$ and $m_c''$ are some positive integer such that $p_b \neq p_c$ for $1\leq b \leq k$ and $1\leq c \leq k'.$ Note that $m_b'' \geq m_b$ for $1\leq b \leq k.$ Such extensions of $\mathbb{F}_{p^{{m''}}}$ can indeed be obtained. For example, $\mathbb{F}_{p^{{m''}}}$  can be considered to be the smallest field which contains $\mathbb{F}_{p^{m'}}$ and $\mathbb{F}_{p^{\tilde{m}}},$ $\tilde{m}$ being some positive integer coprime with $m^\prime$. Then clearly $\mathbb{F}_{p^{{m''}}}$ is such that (\ref{eqnno2}) holds.

Following the notations of Section \ref{sec3}, we now pick $\alpha \in \mathbb{F}_{p^{{m''}}}$ (where ${m''}$ satisfies (\ref{eqnno2})) such that the following condition holds.
\begin{itemize}
\item The cyclic subgroup $\left\{1,\alpha,...,\alpha^{n-1}\right\}$ of
$\mathbb{F}_{p^{{m''}}}\backslash\left\{0\right\}$ with order $n (n>1)$ is such
that $n$ and ${\prod_{b=1}^{b=k}}p_b^{m_b''}$ are coprime.
\end{itemize}
Such an $\alpha$ can be obtained by choosing $\alpha$ from the subgroup of $\mathbb{F}_{p^{m''}}\backslash \left\{0\right\}$ with
$n={\prod_{c=1}^{c=k'}}p_c^{m_c''}$ elements. We now claim that using such an $\alpha$ for the DFT will result in a feasible transform network code for $({\cal G},{\cal C}).$ The proof is as follows. 

We first note that the zero-interference conditions are satisfied irrespective of the choice of $\alpha$ in the DFT operations. As for the invertibility conditions, by Lemma \ref{nonzeros}, it is clear that as long as $f(\alpha^t)\neq 0$ for $0 \leq t \leq n-1,$ we have a feasible transform network code for $({\cal G},{\cal C}).$ Suppose $f(\alpha^t)=0$ for some $1 \leq t \leq n-1.$ Let $n_t$ be the order of $\alpha^t,$ i.e., the number of elements in the cyclic group generated by $\alpha^t.$ Then $n_t|n$ and also $n_t|{\prod_{b=1}^{b=k}}p_b^{m_b''}$ as $\alpha^t \in \mathbb{F}_{p^{m'}}$ is a zero of $f(D).$ However this leads to a contradiction as $n$ shares no common prime factor with ${\prod_{b=1}^{b=k}}p_b^{m_b''}.$ Thus no $\alpha^t,~1 \leq t \leq n-1,$ can be a zero of $f(D).$ This, coupled with the given fact that $f(1)\neq 0,$ proves the claim and hence the if part of the theorem.

\textit{\underline{Only If part}:} Let $\mathbb{F}_{p^m}$ be the field over which a feasible transform network code has been defined for $({\cal G},{\cal C}),$ i.e., there exists a choice of LECs and for $\alpha$ from $\mathbb{F}_{p^m}$ using which the zero-interference and the invertibility constraints have been satisfied in the transform domain. Note that a choice for the LECs implies that the matrices $M_j(D)$ given by (\ref{eqnno3}) are well defined. We will now prove that the invertibility and the zero-interference constraints also hold in these $M_j(D)$ matrices for all sinks, i.e., for $1 \leq j \leq r.$ 

We first prove the invertibility conditions. Towards that end, let $\hat{M}_j^{(n-1)}$ be defined as the $\nu_j \times \mu$ transfer matrix at time instant $n-1$ from all the sources to sink-$j$ in the transform approach, i.e.,
\begin{equation}
\label{eqnno8}
\hat{M}_j^{(n-1)} = \left[\hat{M}_{1j}^{(n-1)} \hat{M}_{2j}^{(n-1)} ...
\hat{M}_{sj}^{(n-1)} \right].
\end{equation}
By the structure of the $\hat{M}_{ij}^{(n-1)}$ matrices, we have $\hat{M}_j^{(n-1)} = \sum_{d=0}^{d=d_{max}}M_j^{(d)}=M_j(D)|_{D=1}.$ Let $\hat{M}_j^{'(n-1)}$ be the submatrix of $\hat{M}_j^{(n-1)}$ which is known to be invertible, as it is given that the invertibility conditions for the transform network code are all satisfied. 

The invertibility conditions for sink-$j$ of the usual (non-transform) network code for $({\cal G},{\cal C})$ demand a suitable submatrix $M'_j(D)$ of the matrix $M_j(D)$ to be invertible. Note however that $M'_j(D)|_{D=1} = \hat{M}_j^{'(n-1)},$ by (\ref{eqnno8}). Therefore, we have $det\left(\hat{M}_j^{'(n-1)}\right) = det\left(M'_j(D)|_{D=1}\right)\neq 0.$ As in (\ref{eqnno10}), we have $det\left(M'_j(D)\right)|_{D=1}=det\left(M'_j(D)|_{D=1}\right) \neq 0.$ Therefore, $det\left(M'_j(D)\right) \neq 0,$ i.e., $det\left(M'_j(D)\right)$ is a non-zero polynomial in $D.$ Because the choice of the sink was arbitrary, it is clear that the invertibility conditions hold for each sink in the usual network code for $({\cal G},{\cal C}).$  By (\ref{eqnno1}), we also have $(D-1)\nmid f(D).$

We now prove the zero-interference conditions. The zero-interference conditions in the transform domain can be interpreted as follows. Having ordered the input processes at the source-$i,$ suppose the sink-$j$ does not demand the $k^{th}$ process from the source-$i.$ Then the matrix $\hat{M}_{ij}$ is such that $k^{th}$ column of $\hat{M}_{ij}^{(t)}$ is an all-zero column for all $0 \leq t \leq n-1.$ To prove that the zero-interference conditions continue to hold in the usual network code for $({\cal G},{\cal C}),$ we must then prove that for each source-$i$, each particular sink-$j$  and each $k$ (such that the $k^{th}$ input process at source-$i$ is not demanded at sink-$j$, the $k^{th}$ columns of $M_{ij}^{(d)}$ matrices are all-zero for $0\leq d \leq d_{max}$ where, $M_{ij}^{(d)}, 0 \leq d \leq d_{max}$  are matrices such that 
\[
M_{ij}(D) = \sum_{d=0}^{d_{max}}M_{ij}^{(d)}D^d.
\] 

This is seen by observing the structure of the $M_{ij}$ matrix, which is defined by (\ref{bigeqn2}). Using Theorem \ref{diagthm} and with  $\beta_a = \alpha^{a}$, we have (\ref{eqnno9}). Comparing the submatrices of $M_{ij}$ from (\ref{bigeqn2}) and (\ref{eqnno9}), we see that if the $k^{th}$ column of the $\hat{M}_{ij}^{(t)}$ matrices is all-zero for all $0 \leq t \leq n-1,$ then the $k^{th}$ columns of $M_{ij}^{(d)}$ matrices are all-zero for $0\leq d \leq d_{max}$. As the choice of source-$i$ and sink-$j$ are arbitrary, it is clear that the zero-interference conditions continue to hold in the $M_{ij}(D)$ matrices for all $1 \leq i \leq s$ and $1 \leq j \leq r.$ This proves the only if part of the theorem and hence, the theorem is proved.	 
\begin{figure*}
\small
\begin{align}
\nonumber
&M_{ij}=Q_{\nu_{j}}\hat{M}_{ij}Q_{\mu_{i}}^{-1}\\
\nonumber
&=
\begin{bmatrix}
I_{\nu_j} & I_{\nu_j} & I_{\nu_j} & \cdots & I_{\nu_j} \\
I_{\nu_j} & {\beta}_{1}I_{\nu_j} & {\beta}_{1}^2I_{\nu_j} & \cdots &
{\beta}_{1}^{n-1}I_{\nu_j} \\
\vdots & \vdots & \vdots & \cdots & \vdots \\
I_{\nu_j} & {\beta}_{n-1}I_{\nu_j} & {\beta}_{n-1}^{2}I_{\nu_j} & \cdots &
{\beta}_{n-1}^{n-1}I_{\nu_j} \\
\end{bmatrix}
\begin{bmatrix}
          {\hat{M}_{ij}^{(n-1)}} & 0 & 0 & \cdots & 0 \\
          0 & {\hat{M}_{ij}^{(n-2)}} & 0 & \cdots & 0 \\
          \vdots & \vdots & \vdots & \cdots & \vdots \\
          0 & 0 & 0 & \cdots & {\hat{M}_{ij}^{(0)}} \\
\end{bmatrix}\\
\nonumber
&\hspace{2cm}\times \begin{bmatrix}
I_{\mu_i} & I_{\mu_i} & I_{\mu_i} & \cdots & I_{\mu_i} \\
I_{\mu_i} & {\beta}_{1}^{-1}I_{\mu_i} & {\beta}_{1}^{-2}I_{\mu_i} & \cdots &
{\beta}_{1}^{-(n-1)}I_{\mu_i} \\
\vdots & \vdots & \vdots & \cdots & \vdots \\
I_{\mu_i} & {\beta}_{n-1}^{-1}I_{\mu_i} & {\beta}_{n-1}^{-2}I_{\mu_i} & \cdots &
{\beta}_{n-1}^{-(n-1)}I_{\mu_i} \\
\end{bmatrix} \\ 
\label{eqnno9}
&=
\begin{bmatrix}
\sum_{t=0}^{n-1}\hat{M}_{ij}^{(t)} &
\sum_{t=0}^{n-1}\beta^{-1}_{n-1-t}\hat{M}_{ij}^{(t)} &  \cdots &
\sum_{t=0}^{n-1}\beta^{-(n-1)}_{n-1-t}\hat{M}_{ij}^{(t)} \\
\sum_{t=0}^{n-1}\beta^{n-1-t}_{1}\hat{M}_{ij}^{(t)} &
\sum_{t=0}^{n-1}\hat{M}_{ij}^{(t)} &  \cdots &
\sum_{t=0}^{n-1}\beta^{n-1-t}_{1}\beta_{n-1-t}^{-(n-1)}\hat{M}_{ij}^{(t)} \\
\vdots & \vdots & \cdots & \vdots \\
 \sum_{t=0}^{n-1}\beta^{n-1-t}_{n-1}\hat{M}_{ij}^{(t)} &
\sum_{t=0}^{n-1}\beta^{n-1-t}_{n-1}\beta_{n-1-t}^{-1}\hat{M}_{ij}^{(t)} & \cdots
&
\sum_{t=0}^{n-1}\hat{M}_{ij}^{(t)} \\ 
\end{bmatrix}.
\end{align}
~\\
\hrule
\end{figure*}
\end{IEEEproof}

\section{Proof of Lemma \ref{lemma2}}
\label{appen_lemma2}
\begin{IEEEproof}
Consider $M_{ij}$ as defined in (\ref{bigeqn2}) which is a circulant matrix of size $(2n'+1)\times(2n'+1)$. The diagonal elements of $\hat{M}_{ij}$, i.e., $\hat{M}_{ij}^{(k)}$, for $k =0,1,\cdots,2n'$, are the eigen values of the matrix $M_{ij}$. Note that the eigen values are equal to $(2n'+1)$-point finite-field DFT of the first row of $M_{ij}$. Since, the min-cut from source-$i$ to sink-$j$ is greater than or equal to $1$, by Menger's Theorem, there exists at least one link-disjoint directed path from source-$i$ to sink-$j$. Consider one such directed path consisting of links $e_1$, $e_2,~.., ~e_t$. Now, we can assign the values $\alpha_{1,e_1}=1$, $\beta_{e_i,e_{i+1}}=1$ $(i \in \{1,2,..,t-1\})$, $\epsilon_{e_t,1}=1$ and assign values of $0$ to all the other LECs. By such an assignment of values to the LECs, exactly one among $M_{ij}^{(0)}$, $M_{ij}^{(1)}$, $.. $, $M_{ij}^{(d_{max})}$ is equal to $1$. This implies that all the eigen values of $M_{ij}$ are non-zero. Hence, the diagonal elements of $\hat{M}
_{ij}$ are non-zero polynomials in $\underline{\varepsilon}$ and so is its determinant.
\end{IEEEproof}

\section{Proof of Lemma \ref{lemma_newly_added}}
\label{appen_lemma_newly_added}
\begin{IEEEproof}
\underline{\em If part:} Euler's theorem \cite{DuF-book} states that if two positive integers $a$ and $b$ are coprime then, $b$ divides $a^{\phi(b)}-1$ where $\phi$ represents the Euler's totient function. If $2n'+1 < p$ then, $2n'+1$ and $p$ are coprime. If $2n'+1 \geq p$ then, $p$ and $2n'+1$ are coprime iff $p$ does not divide $2n'+1$. Hence, by Euler's theorem, $2n'+1|p^{\phi(2n'+1)}-1$ if $p \nmid 2n'+1$. Thus if $p \nmid 2n'+1$ then, $2n'+1|p^{m}-1$, for all $m$ such that $\phi(2n'+1)|m$.\\
\underline{\em Only If part:} If $2n'+1$ divides $p^m-1$ for some positive integer $m$ then, $p^m-1=r(2n'+1)$ for some positive integer $r$. So, $p^m-(2n'+1)r=1$ which means that $p$ and $2n'+1$ must be coprime. Since $p$ is prime, $p \nmid 2n'+1$.
\end{IEEEproof} 

\section{Proof of Theorem \ref{thm3}}
\label{appen_thm3}
\begin{IEEEproof}
To exactly recover ${X_1^\prime}^{n'+1}$, ${X_2^\prime}^{n'}$ and ${X_3^\prime}^{n'}$ at the sinks-$1$, $2$ and $3$ respectively, it is sufficient that the following network alignment conditions are satisfied.
\begin{align}
\label{thm2_pf1}
&\hat{M}_{21}V_2 = \hat{M}_{31}V_3\\
\label{thm2_pf2}
&\text{Col}\left(\hat{M}_{32}V_3\right) \subset \text{Col}\left(\hat{M}_{12}V_1\right)\\
\label{thm2_pf3}
&\text{Col}\left(\hat{M}_{23}V_2\right) \subset \text{Col}\left(\hat{M}_{13}V_1\right)\\
\nonumber
&\mbox{Rank}[\hat{M}_{11}V_1 ~~\hat{M}_{21}V_2]=2n'+1 \\
\label{thm2_pf4}
\Leftrightarrow ~&\mbox{Rank}[V_1 ~~\hat{M}_{11}^{-1}\hat{M}_{21}V_2]=2n'+1\\
\nonumber
&\mbox{Rank}[\hat{M}_{22}V_2 ~~\hat{M}_{12}V_1]=2n'+1\\
\label{thm2_pf5}
\Leftrightarrow ~&\mbox{Rank}[\hat{M}_{12}^{-1}\hat{M}_{22}V_2 ~~V_1]=2n'+1\\
\nonumber
&\mbox{Rank}[\hat{M}_{33}V_3 ~~\hat{M}_{13}V_1]=2n'+1\\
\label{thm2_pf6}
\Leftrightarrow ~&\mbox{Rank}[\hat{M}_{13}^{-1}\hat{M}_{33}V_3 ~~V_1]=2n'+1
\end{align}
Note that from Lemma \ref{lemma2}, inverse of $\hat{M}_{ij}$ $\forall$ $(i,j)\in \{1,2,3\}$ is well-defined. It is easily seen that the choice of $V_1$, $V_2$ ,and $V_3$ in (\ref{thm2_pf7})-(\ref{thm2_pf9}) satisfy the conditions (\ref{thm2_pf1})-(\ref{thm2_pf3}). Suppose that (\ref{thm2_pf4})-(\ref{thm2_pf6}) are satisfied. Let
\begin{align*}
&f_1(\underline{\varepsilon})= det([V_1 ~~\hat{M}_{11}^{-1}\hat{M}_{21}V_2])\\
&f_2(\underline{\varepsilon})= det([\hat{M}_{12}^{-1}\hat{M}_{22}V_2 ~~V_1])\\
&f_3(\underline{\varepsilon})=det([\hat{M}_{13}^{-1}\hat{M}_{33}V_3 ~~V_1])\\
&f_4(\underline{\varepsilon})=\prod_{(i,j)\in\{1,2,3\}}det(M_{ij})\\
&f(\underline{\varepsilon})=\prod_{i=1}^{4}f_i(\underline{\varepsilon}).
\end{align*} Since $f_1(\underline{\varepsilon})$, $f_2(\underline{\varepsilon})$ and $f_3(\underline{\varepsilon})$ are non-zero
polynomials in $\underline{\varepsilon}$, $f(\underline{\varepsilon})$ is also a non-zero polynomial in $\underline{\varepsilon}$. Hence, by Lemma $1$ in \cite{KoM}, for a sufficiently large field size, there exists an assignment of
values to $\underline{\varepsilon}$ such that the network alignment conditions are satisfied. Since $p \nmid 2n'+1$, by Lemma \ref{lemma_newly_added}, for a sufficiently large $m$ (in particular, $m$ such that $\phi(2n'+1)|m$ where $\phi$ represents the Euler's totient function), there exists an assignment of values to $\underline{\varepsilon}$ such that the network alignment conditions are satisfied. Hence, the theorem is proved.
\end{IEEEproof}

\section{Proof of Lemma \ref{lemma3}}
\label{appen_lemma3}
\begin{IEEEproof}
If we assign
$\underline{\varepsilon}^{(-d_{max})}=\underline{\varepsilon}^{(-d_{max}+1)} =\ldots= \underline{\varepsilon}^{(n-1)}=\underline{\varepsilon}$, $M_{ij}$ in (\ref{bigeqn1_unicast}) becomes a circulant matrix. Since, the min-cut from source-$i$ to sink-$j$ is greater than or equal to $1$, by Menger's Theorem, there exists at least one link-disjoint directed path from source-$i$ to sink-$j$. Consider one such directed path consisting of links $e_1$, $e_2,~.., ~e_t$. Now, assign the values $\alpha_{1,e_1}=1$, $\beta_{e_i,e_{i+1}}=1$, for \mbox{$i =1,2,\cdots,t-1$}, $\epsilon_{e_t,1}=1$ and assign values of $0$ to all the other LECs. By such an assignment, $M_{ij}$ becomes a permuted identity matrix whose determinant is non-zero. Hence, the determinant of $M_{ij}$ is a non-zero polynomial in $\underline{\varepsilon}'$.
\end{IEEEproof}

\section{Proof of Theorem \ref{thm4}}
\label{appen_thm4}
\begin{IEEEproof}
To exactly recover ${X_1^\prime}^{n_1}$, ${X_2^\prime}^{n_2}$ and
${X_3^\prime}^{n_3}$ at $T_1$, $T_2$ and $T_3$ respectively, it is sufficient if the following network alignment conditions are satisfied.
\begin{align}
\label{thm3_pf1}
&\text{Span}({M}_{31}V_3) \subset \text{Span}({M}_{21}V_2)\\
\label{thm3_pf2}
&\text{Span}({M}_{32}V_3) \subset \text{Span}({M}_{12}V_1)\\
\label{thm3_pf3}
&\text{Span}({M}_{23}V_2) \subset \text{Span}({M}_{13}V_1)\\
\nonumber
&\mbox{Rank}[{M}_{11}V_1 ~~{M}_{21}V_2]=n_1+n_2 \\
\label{thm3_pf4}
\Leftrightarrow ~&\mbox{Rank}[V_1 ~~{M}_{11}^{-1}{M}_{21}V_2]=n_1+n_2\\
\nonumber
&\mbox{Rank}[{M}_{22}V_2 ~~{M}_{12}V_1]=n_1+n_2\\
\label{thm3_pf5}
\Leftrightarrow ~&\mbox{Rank}[{M}_{12}^{-1}{M}_{22}V_2 ~~V_1]=n_1+n_2\\
\nonumber
&\mbox{Rank}[{M}_{33}V_3 ~~{M}_{13}V_1]=n_1+n_3\\
\label{thm3_pf6}
\Leftrightarrow ~&\mbox{Rank}[{M}_{13}^{-1}{M}_{33}V_3 ~~V_1]=n_1+n_3
\end{align}
\noindent Note that  (\ref{thm3_pf4})-(\ref{thm3_pf6}) are also necessary conditions whereas (\ref{thm3_pf1})-(\ref{thm3_pf3}) are necessary when $(n_1+n_2)=(n_1+n_3)=n$ \mbox{($\because$ $n_1 \geq n_2 \geq n_3$)}. Clearly, (\ref{thm3_pf4}) and (\ref{thm3_pf6}) cannot be satisfied when $(n_1+n_2)>n$ and $(n_1+n_3)>n$ respectively. Therefore, $(n_1+n_2) \leq n$. 

The choice of $V_2$ and $V_3$ in (\ref{na2_V2_and_V3}) ensures that the conditions (\ref{thm3_pf2}) and (\ref{thm3_pf3}) are satisfied. To satisfy (\ref{thm3_pf1}), we have to ensure that
\begin{align}
\nonumber
&M_{31}V_3=M_{21}V_2C\\
\nonumber
\Leftrightarrow
~&M_{32}^{-1}M_{12}V_1B=M_{31}^{-1}M_{21}M_{23}^{-1}M_{13}V_1AC\\
\label{matrices_eq}
\Leftrightarrow ~&V_1B=UV_1AC
\end{align}
\noindent is satisfied. In order to satisfy (\ref{matrices_eq}), every element
of $UV_1A$ must be equal to every element of $V_1BC$, i.e.,
\begin{align*}
 g_{ij}=0, ~\text{for } ~i =1,2,\cdots,n,~j =1,2,\cdots,n_3.
\end{align*}
\noindent To ensure that (\ref{thm3_pf4}) is satisfied, we require that at least one among $f^{(k)}_1$, $k = 1,2,\ldots, {n \choose n_1+n_2}$, take a non-zero value after some assignment to the variables and LECs. This necessitates that, firstly, $f^{(k)}_1$ should be a non-zero rational polynomial for some $k$. It can be easily seen that $f^{(k)}_1$ is a non-zero rational polynomial for some $k$ iff {\small $f_1=\left(1-\prod_{k=1}^{n \choose n_1+n_2}(1-\delta^{(k)}_1f^{(k)}_1)\right)$} is a non-zero rational polynomial. Similarly, from (\ref{thm3_pf5}) and (\ref{thm3_pf6}), {\small $f_2=\left(1-\prod_{k=1}^{n \choose n_1+n_2}(1-\delta^{(k)}_2f^{(k)}_2)\right)$} and \mbox{{\small $f_3=\left(1-\prod_{k=1}^{n \choose n_1+n_3}(1-\delta^{(k)}_3f^{(k)}_3)\right)$}} must be non-zero rational polynomials. Hence, to satisfy (\ref{thm3_pf4})-(\ref{matrices_eq}) we need to find an assignment to  $V_1,~\underline{\varepsilon}',A$, and $B$, such that $f \neq 0$ and $g_{ij}=0$, for all  $(i,j)$. This means that there 
must exist an assignment such that $f^{(nr)} \neq 0$ and $g^{(nr)}_{ij}=0$. After the assignment to the variables, we require that $f^{(dr)} \neq 0$ and $g^{(dr)}_{ij}\neq 0$ as dividing by zero is prohibited. In order to formulate this as an algebraic problem, introduce a new variable $\delta$ and consider the polynomial $\left(1-\delta f^{(nr)}f^{(dr)} \prod_{(i,j)}g^{(dr)}_{ij}\right)$. From Weak Nullstellensatz \cite{CLO-book}, an
assignment to $\delta, ~V_1,~\underline{\varepsilon}',~A$, $B$,  $C$, and $\delta^{(k)}_i$, for all $(i,k)$, exist such that $g^{(nr)}_{ij}=0$, for all $(i,j)$, and $\left(1-\delta f^{(nr)}f^{(dr)} \prod_{(i,j)}g^{(dr)}_{ij}\right)=0$ iff $1$ does not belong to the ideal generated by the polynomials $g^{(nr)}_{ij}$ for all $(i,j)$ and $\left(1-\delta f^{(nr)}f^{(dr)} \prod_{(i,j)}g^{(dr)}_{ij}\right)$.
\end{IEEEproof}

\section{Proof of Theorem \ref{thm5}}
\label{appen_thm5}. 
\begin{IEEEproof}
To exactly recover ${X_1^\prime}^{n_1}$, ${X_2^\prime}^{n_2}$ and
${X_3^\prime}^{n_3}$ at $T_1$, $T_2$ and $T_3$ respectively, it is sufficient if the following network alignment conditions are satisfied.
\begin{align}
\label{thm5_pf1}
&\text{Span}(M_{32}V_3) \subset \text{Span}(M_{12}V_1)\\
\label{thm5_pf2}
&\text{Span}(M_{23}V_2) \subset \text{Span}(M_{13}V_1)\\
\nonumber
&\mbox{Rank}[M_{11}V_1 ~~M_{31}V_3]=n_1+n_3 \\
\label{thm5_pf3}
\Leftrightarrow ~&\mbox{Rank}[V_1 ~~M_{11}^{-1}M_{31}V_3]=n_1+n_3\\
\nonumber
&\mbox{Rank}[M_{22}V_2 ~~M_{12}V_1]=n_1+n_2\\
\label{thm5_pf4}
\Leftrightarrow ~&\mbox{Rank}[M_{12}^{-1}M_{22}V_2 ~~V_1]=n_1+n_2\\
\nonumber
&\mbox{Rank}[M_{33}V_3 ~~M_{13}V_1]=n_1+n_3\\
\label{thm5_pf5}
\Leftrightarrow ~&\mbox{Rank}[M_{13}^{-1}M_{33}V_3 ~~V_1]=n_1+n_3
\end{align}
\noindent Note that  (\ref{thm5_pf3})-(\ref{thm5_pf5}) are also necessary conditions whereas (\ref{thm5_pf1}) and (\ref{thm5_pf2}) are necessary when $(n_1+n_2)=(n_1+n_3)=n$. Clearly, (\ref{thm5_pf4}) and (\ref{thm5_pf5}) cannot be satisfied when $(n_1+n_2)>n$ and $(n_1+n_3)>n$ respectively.  Therefore, $(n_1+n_2) \leq n$. 

It is easily seen that the choice of $V_2$ and $V_3$, in (\ref{mincut0_poss1_V2_and_V3}), satisfy the conditions (\ref{thm5_pf1}) and (\ref{thm5_pf2}). If (\ref{thm5_pf3})-(\ref{thm5_pf5}) are satisfied then, the determinants of at least one of the $(n_1+n_3) \times (n_1+n_3)$ submatrices of {\small $[M_{11}V_1 ~~M_{31}V_3]$}, $(n_1+n_2) \times (n_1+n_2)$ submatrices of {\small $[M_{12}^{-1}M_{22}V_2 ~~V_1]$}, and $(n_1+n_3) \times (n_1+n_3)$ submatrices of {\small $[M_{13}^{-1}M_{33}V_3 ~~V_1]$}
will be non-zero rational polynomials. Let $f$ be the product of the numerators and denominators of these non-zero rational polynomials. Hence, by Lemma 1 in \cite{KoM}, for a sufficiently large field size, there exists an assignment of values to  $\underline{\varepsilon}$ and other variables involved such that the network alignment conditions are satisfied. Hence, the theorem is proved.
\end{IEEEproof}

\section{Proof of Theorem \ref{thm_GC}}
\label{appen_thm_GC}. 
\begin{IEEEproof}
{\em If part}: Using the precoding matrices given in (\ref{eqn-precod-mats}), if $X'_i(0)$ can be recovered from $Y^{(0 \oplus k)}_i$ for all $i$ then, the determinants in (\ref{eqn-NA-nec-suff}) are non-zero polynomials in $\left(\underline{\varepsilon_1}, \cdots, \underline{\varepsilon_{2n'+1}}\right)$ for $q=0$. Note that the transfer matrices $M_{ij}\left(\varepsilon,\alpha^q\right)$, for all $q$, in the $3$-S $3$-D MUN-D can also be simulated as the ones obtained from its instantaneous network counterpart (i.e., $q=0$) by multiplying each of the LEC by $\alpha^q$. Suppose that one of the determinants in (\ref{eqn-NA-nec-suff}) is a zero-polynomial for some $q$ (i.e., $X'_i(q)$ cannot be recovered from $Y^{(q \oplus k)}_i$ for at least one $i$). Then, this determinant is also a zero polynomial with $\left(\underline{\varepsilon_1}, \cdots, \underline{\varepsilon_{2n'+1}}\right)$ replaced by $\left({\underline{\varepsilon_1}}/{\alpha^q}, \cdots, {\underline{\varepsilon_{2n'+1}}}/{\alpha^q}\right)$ where, 
${\underline{\varepsilon_{l}}}/{\alpha^q}$ denotes each of the LECs multiplied by the inverse of $\alpha^q$. But this contradicts the fact that all the determinants  in (\ref{eqn-NA-nec-suff}) are non-zero polynomials in $\left(\underline{\varepsilon_1}, \cdots, \underline{\varepsilon_{2n'+1}}\right)$, in the $q=0$ case. Hence, the determinants in (\ref{eqn-NA-nec-suff}) are non-zero polynomials for all $q$ and all $i$. Using a sufficiently large field size such that $k|2^m-1$, by Lemma $1$ in \cite{KoM}, there exists an assignment to $\left(\underline{\varepsilon_1}, \cdots, \underline{\varepsilon_{2n'+1}}\right)$ such that determinants in (\ref{eqn-NA-nec-suff}) are non-zero for all $q$. 

{\em Only-if part}: Using the precoding matrices given in (\ref{eqn-precod-mats}), if $X'_i(q)$ can be recovered from $Y^{(0 \oplus k)}_i$ for all $i$ and for some $q \neq 0$ then, the determinants in (\ref{eqn-NA-nec-suff}) are non-zero polynomials in $\left(\underline{\varepsilon_1}, \cdots, \underline{\varepsilon_{2n'+1}}\right)$ for some $q=q' \neq 0$. Suppose that one of the determinants in (\ref{eqn-NA-nec-suff}) is a zero-polynomial for $q=0$ (i.e., $X'_i(0)$ cannot be recovered from $Y^{(0 \oplus k)}_i$ for at least one $i$). Then, this determinant is also a zero polynomial with $\left(\underline{\varepsilon_1}, \cdots, \underline{\varepsilon_{2n'+1}}\right)$ replaced by $\left({\underline{\varepsilon_1}}{\alpha^{q'}}, \cdots, {\underline{\varepsilon_{2n'+1}}}{\alpha^{q'}}\right)$ where, ${\underline{\varepsilon_{l}}}{\alpha^{q'}}$ denotes each of the LECs multiplied by $\alpha^{q'}$. But this contradicts the fact that all the determinants in  (\ref{eqn-NA-nec-suff}) are non-zero polynomials in the 
$q=q'$ case. Thus, the determinants in (\ref{eqn-NA-nec-suff}) are non-zero polynomials for $q=0$ and all $i$. Hence, using the ``If part'', the determinants in (\ref{eqn-NA-nec-suff}) are non-zero polynomials for all $q$ and all $i$. Using a sufficiently large field size such that $k|2^m-1$, by Lemma $1$ in \cite{KoM}, there exists an assignment to $\left(\underline{\varepsilon_1}, \cdots, \underline{\varepsilon_{2n'+1}}\right)$ such that determinants in (\ref{eqn-NA-nec-suff}) are non-zero for all $q$. 
\end{IEEEproof}

\end{document}